\let\latexarabic\arabic
\let\latexdocument\document
\let\latexenddocument\enddocument

\RequirePackage[thmmarks]{ntheorem}
\makeatletter
\renewtheoremstyle{plain} 
  {\item[\hskip\labelsep \theorem@headerfont ##1\ \textup{##2}\theorem@separator]} 
  {\item[\hskip\labelsep \theorem@headerfont ##1\ \textup{##2}\ (##3)\theorem@separator]}
\makeatother

\documentclass[lineno]{biometrika_yiqun}

\let\document\latexdocument
\let\enddocument\latexenddocument
\AtEndDocument{\printhistory}
\let\arabic\latexarabic
\def\rm{}
\usepackage{amsmath, amssymb, verbatim, color, graphics,multirow}
\usepackage{booktabs,physics}
\usepackage[norelsize,ruled]{algorithm2e}
\usepackage{mathtools}
\usepackage[export]{adjustbox}

\usepackage{bbm}
\usepackage{bm}

\usepackage{tablefootnote}

\usepackage[labelformat=simple]{subcaption}

\usepackage{mathtools}
\makeatletter
\newcommand{\vast}{\bBigg@{4}}
\newcommand{\Vast}{\bBigg@{5}}

\makeatother

\usepackage[colorlinks=true,linkcolor=blue,citecolor=blue]{hyperref}%
\DeclareMathAlphabet{\mathcal}{OMS}{cmsy}{m}{n}

\newcommand{\pK}{\ensuremath{p_{\text{selective}}}}
\newcommand{\pKhat}{\ensuremath{\hat{p}_{\text{selective}}}}
\newcommand{\pKSigma}{\ensuremath{p_{\Sigma,\text{selective}}}}

\newcommand{\pr}{\text{pr}}

\newcommand{\E}{\text{E}}

\renewcommand{\forall}{\text{ for all }}

\makeatletter
\renewcommand{\algocf@captiontext}[2]{#1\algocf@typo. \AlCapFnt{}#2} 
\def\@algocf@capt@plain{top}
\renewcommand{\algocf@makecaption}[2]{%
  \addtolength{\hsize}{\algomargin}%
  \sbox\@tempboxa{\algocf@captiontext{#1}{#2}}%
  \ifdim\wd\@tempboxa >\hsize
    \hskip .5\algomargin%
    \parbox[t]{\hsize}{\algocf@captiontext{#1}{#2}}
  \else%
    \global\@minipagefalse%
    \hbox to\hsize{\box\@tempboxa}
  \fi%
  \addtolength{\hsize}{-\algomargin}%
}
\makeatother

\begin{document}
\nolinenumbers

\markboth{Y. Chen and D. Witten}{Selective inference for $k$-means clustering}
\title{Selective inference for $k$-means clustering}

\author{Yiqun T. Chen}
\affil{Department of Biostatistics, University of Washington, Seattle, Washington, 98195, U.S.A.
\email{yiqunc@uw.edu}}

\author{Daniela M. Witten}
\affil{Departments of Biostatistics and Statistics, University of Washington, Seattle, Washington, 98195, U.S.A.
\email{dwitten@uw.edu}}

\date{}
\maketitle
\begin{abstract}
We consider the problem of testing for a difference in means between clusters of observations identified via k-means clustering. In this setting, classical hypothesis tests lead to an inflated Type I error rate. To overcome this problem, we take a selective inference approach. We propose a finite-sample p-value that controls the selective Type I error for a test of the difference in means between a pair of clusters obtained using $k$-means clustering, and show that it can be efficiently computed. We apply our proposal in simulation, and on hand-written digits data and single-cell RNA-sequencing data.

\end{abstract}

\begin{keywords}
Post-selection inference; Unsupervised learning; Hypothesis testing; Type I error
\end{keywords}




\section{Introduction}
\label{section:intro}
Testing for a difference in means between two groups is one of the most fundamental tasks in statistics, with numerous applications. If the groups under investigation are \emph{pre-specified}, i.e., not a function of the observed data, then classical hypothesis tests will control the Type I error rate. However, it is increasingly common to want to test for a difference in means between groups that are \emph{defined through the observed data}, e.g., via the output of a clustering algorithm. For instance, in single-cell RNA-sequencing analysis, researchers often first cluster the cells, and then test for a difference in the expected gene expression levels between the clusters to quantify up- or down-regulation of genes, annotate known cell types, and identify new cell types~\citep{Grun2015-lh,Aizarani2019-xy,Lahnemann2020-xf,Zhang2019-lx,Doughty2020-ss}. In fact, the inferential challenges resulting from testing data-guided hypotheses have been described as a ``grand challenge'' in the field of genomics~\citep{Lahnemann2020-xf}, and papers in the field continue to overlook this issue: as an example, \texttt{seurat}~\citep{Stuart2019-on}, the state-of-the-art single-cell RNA sequencing analysis tool, tests for differential gene expression between groups obtained via clustering, with a note that ``$p$-values [from these hypotheses] should be interpreted cautiously, as the genes used for clustering are the same genes tested for differential expression.'' Testing data-guided hypothesis also arises in the field of neuroscience~\citep{Kriegeskorte2009-lo,Button2019-hz}, social psychology~\citep{Hung2020-no}, and physical sciences~\citep{Friederich2020-aw,Pollice2021-ap}. When the null hypothesis is a function of the data, classical tests that do not account for this will fail to control the Type I error. 

In this paper, we develop a test for a difference in means between two clusters estimated from applying $k$-means clustering~\citep{Lloyd2006-io,MacQueen1967-dd}, an extremely popular clustering algorithm with numerous applications \citep{Xu2008-lp,Wang2008-gy,Steinley2006-cs,Hand2015-gx}. We consider the following simple and well-studied model~\citep{Gao2020-yt,Loffler2021-vj,Lu2016-ic} for $n$ observations and $q$ features: 
\begin{align}
\label{eq:data_gen}
    {X} \sim \mathcal{MN}_{n\times q}\left({\mu}, \mathbf{I}_n, \sigma^2 \mathbf{I}_q \right),
\end{align} 
where ${\mu}\in\mathbb{R}^{n\times q}$ has unknown rows ${\mu}_i$, and $\sigma^2>0$ is known. Given a realization ${x}\in\mathbb{R}^{n\times q}$ of ${X}$, we first apply the $k$-means clustering algorithm to obtain $\mathcal{C}({x})$, a partition of the samples $\{1,\ldots,n\}$. We might then consider testing the null hypothesis that the mean is the same across two \emph{estimated} clusters, i.e.,
\begin{align}
 H_0: \sum_{i\in {\hat{\mathcal{C}}}_1}{\mu}_i/|\hat{\mathcal{C}}_1| = \sum_{i\in \hat{\mathcal{C}}_2}{\mu}_i/|\hat{\mathcal{C}}_2|   \mbox{ versus }  H_1: \sum_{i\in \hat{\mathcal{C}}_1}{\mu}_i/|\hat{\mathcal{C}}_1| \neq \sum_{i\in \hat{\mathcal{C}}_2}{\mu}_i/|\hat{\mathcal{C}}_2|,
  \label{eq:null_clustering_mean}
\end{align}
where $\hat{\mathcal{C}}_1,\hat{\mathcal{C}}_2 \in \mathcal{C}({x})$ are estimated clusters with cardinality $|\hat{\mathcal{C}}_1|$ and $|\hat{\mathcal{C}}_2|$. This is equivalent to testing $H_0: {\mu}^\top \nu = {0}_q  \mbox{ versus }  H_1: {\mu}^\top \nu \neq {0}_q$, where 
\begin{align}
\label{eq:nu_def}
    \nu_i = 1\qty{i\in\hat{\mathcal{C}}_1}/|\hat{\mathcal{C}}_1| - 1\qty{i\in\hat{\mathcal{C}}_2}/|\hat{\mathcal{C}}_2|, \quad i = 1,\ldots, n,
\end{align}
and $1\qty{A}$ equals 1 if the event $A$ holds, and 0 otherwise.
At first glance, we can test the hypothesis in \eqref{eq:null_clustering_mean} by applying a Wald test, with $p$-value given by
\begin{align} 
\label{eq:wald_pval}
p_{\text{Naive}} = \pr_{H_0}\qty( \Vert {X}^\top \nu \Vert_2 \geq  \Vert {x}^\top \nu \Vert_2 ),
\end{align}
where $\Vert {X}^\top \nu \Vert_2 \sim \qty(\sigma \Vert \nu\Vert_2) \chi_q $ under $H_0$. But this ``naive'' approach ignores the fact that $H_0$ is chosen based on the data, i.e., we constructed the contrast vector in \eqref{eq:nu_def} because $\hat{\mathcal{C}}_1$ and $\hat{\mathcal{C}}_2$ were obtained by clustering. Therefore, we will observe substantial differences between the cluster centroids $\sum_{i\in {\hat{\mathcal{C}}}_1}{x}_i/|\hat{\mathcal{C}}_1|$ and $\sum_{i\in \hat{\mathcal{C}}_2}{x}_i/|\hat{\mathcal{C}}_2|$, even in the absence of true differences in their population means (left panel Figure~\ref{fig:motivation}). The center panel of Figure~\ref{fig:motivation} illustrates that the test based on \eqref{eq:wald_pval} does not control the selective Type I error: that is, the probability of falsely rejecting a null hypothesis, given that we decided to test it~\citep{Fithian2014-ow}.
\begin{figure}[htbp!]
\begin{subfigure}{\linewidth}
  \centering
  \includegraphics[width=.3\linewidth]{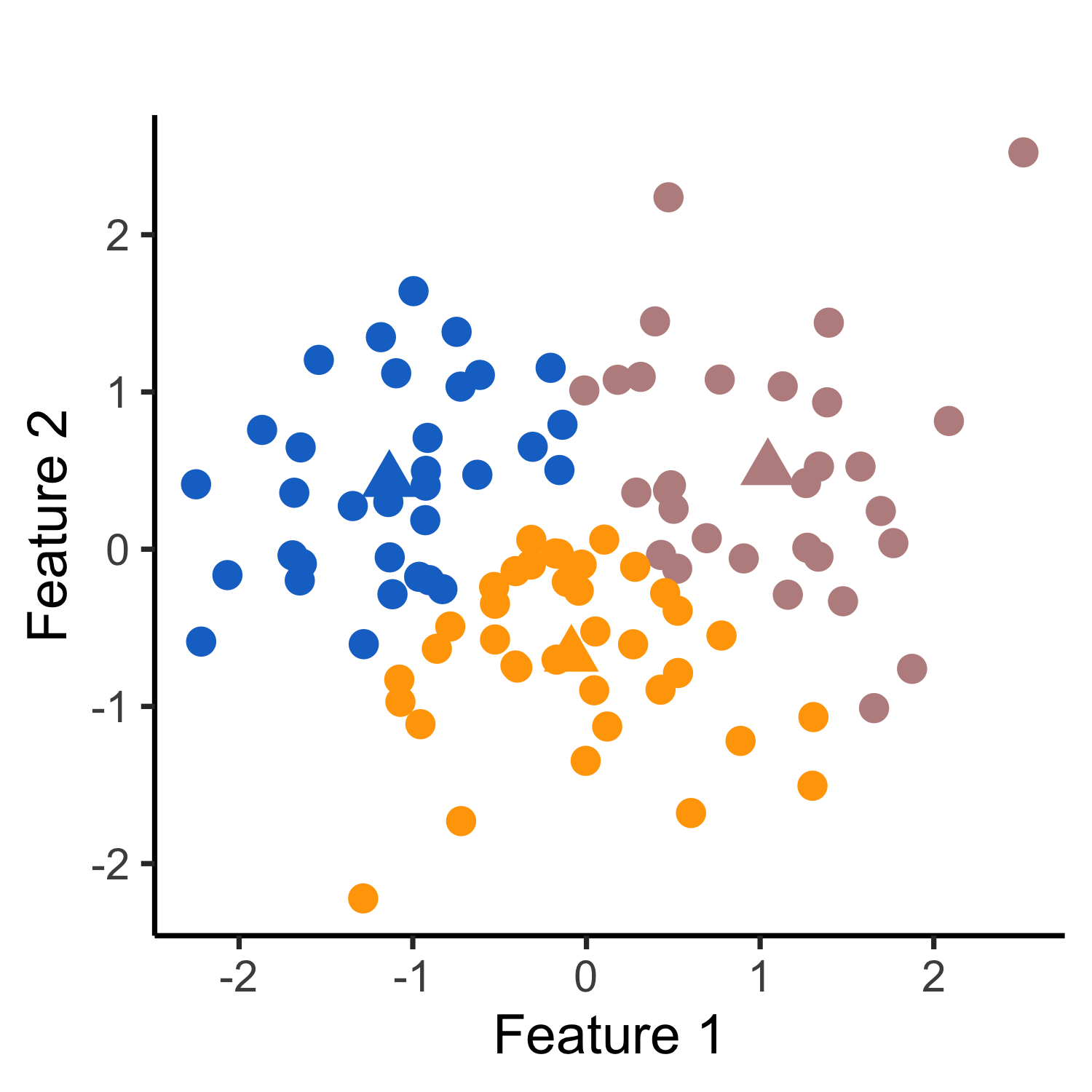}
  \includegraphics[width=.3\linewidth]{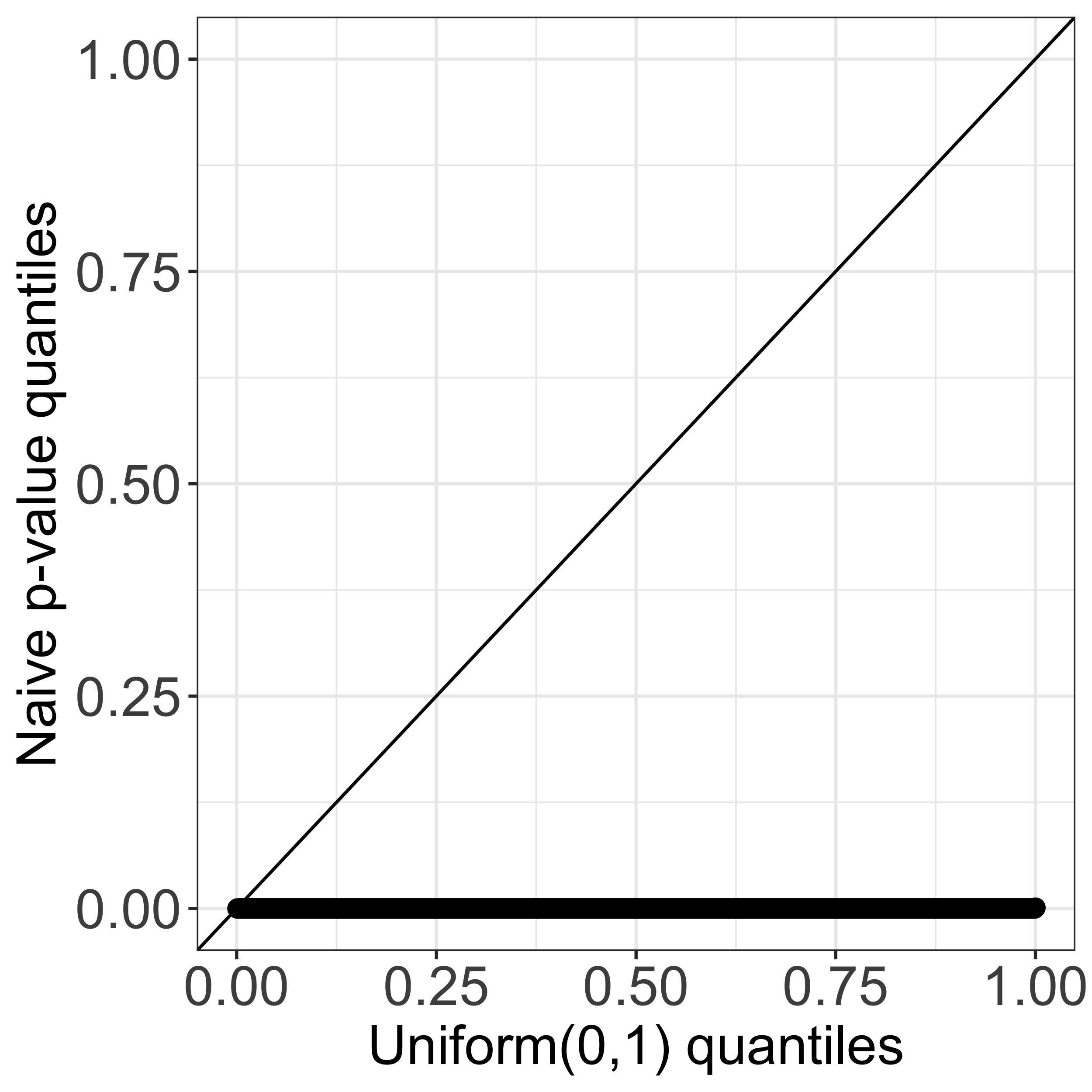}
  \includegraphics[width=.3\linewidth]{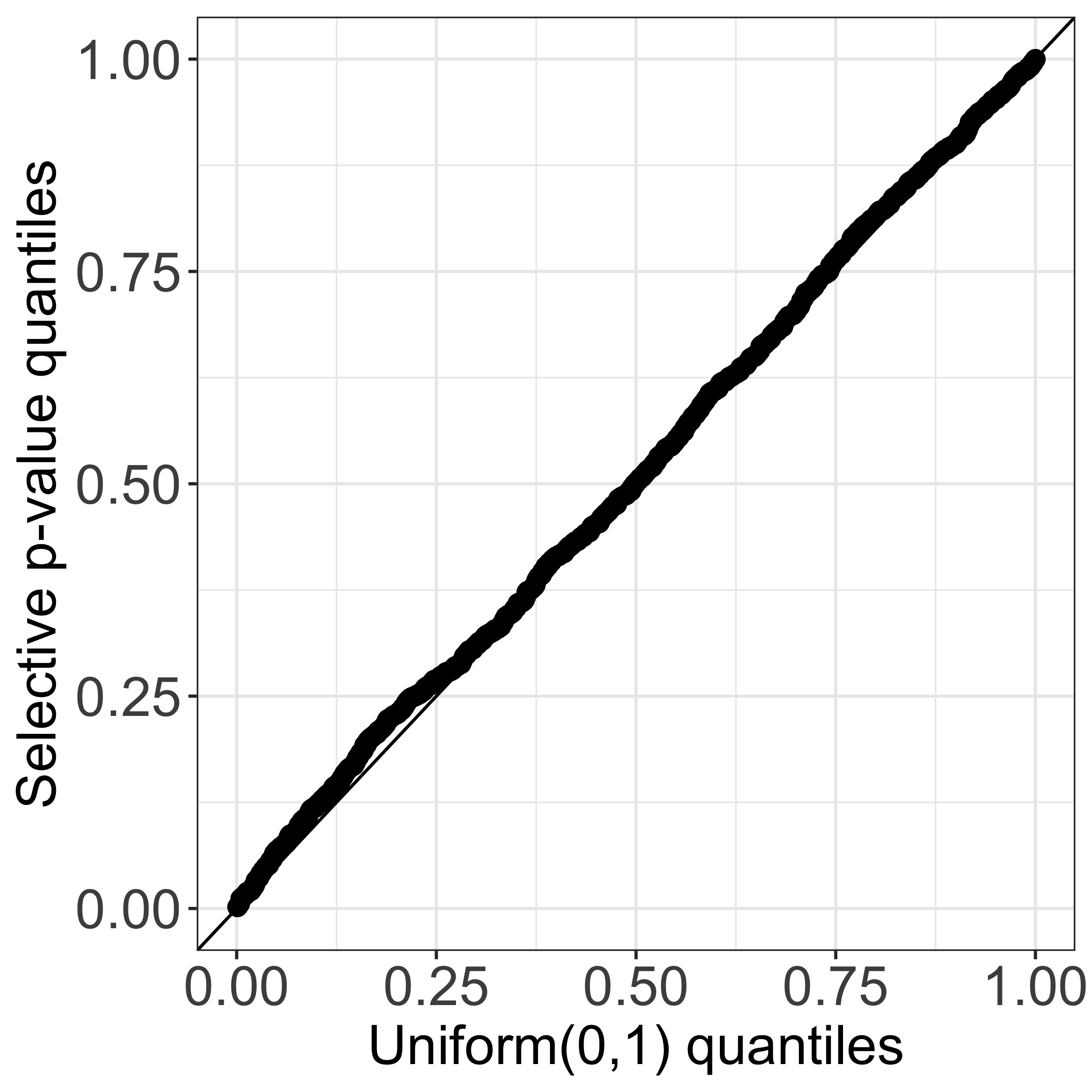}
\end{subfigure}
\vspace*{-4mm}
\caption{\textit{Left: } One simulated dataset generated from \eqref{eq:data_gen} with ${\mu} = {0}_{100\times 2}$ and $\sigma = 1$. We apply $k$-means clustering to obtain three clusters. The cluster centroids are displayed as triangles. \textit{Center: } Quantile-quantile plot of the naive $p$-values (defined in \eqref{eq:wald_pval}) applied to 2,000 simulated datasets from \eqref{eq:data_gen} with ${\mu} = {0}_{100\times 2}$ and $\sigma = 1$. 
\textit{Right: }Quantile-quantile plot of our proposed $p$-values (defined in \eqref{eq:p_val_k_means_chen}) applied to the same simulated datasets.}
\label{fig:motivation}
\end{figure}

Notably, the problem of testing for a difference in means between two groups obtained via clustering cannot be easily overcome by sample splitting, as pointed out in~\citet{Gao2020-yt} and \citet{Zhang2019-lx}. To see why, we divide the observations into a training and a test set. We apply $k$-means clustering on {only} the training set (left panel of Figure~\ref{fig:sample_split}), and then assign the test set observations to those clusters (to obtain the center panel of Figure~\ref{fig:sample_split}, we applied a 3-nearest neighbor classifier). Finally, we compute the naive $p$-values \eqref{eq:wald_pval} \emph{only} on the test set. Unfortunately, this approach does not work: while we clustered only the training data, we still used the test data to label the test observations, and consequently to construct the contrast vector $\nu$ in \eqref{eq:nu_def}. Therefore, the Wald test based on sample-splitting remains extremely anti-conservative, as shown in the right panel of Figure~\ref{fig:sample_split}, and does not lead to a valid test of $H_0$ in \eqref{eq:null_clustering_mean}.

\noindent
\begin{figure}[htbp!]
\begin{subfigure}{\linewidth}
  \centering
  \includegraphics[width=.3\linewidth]{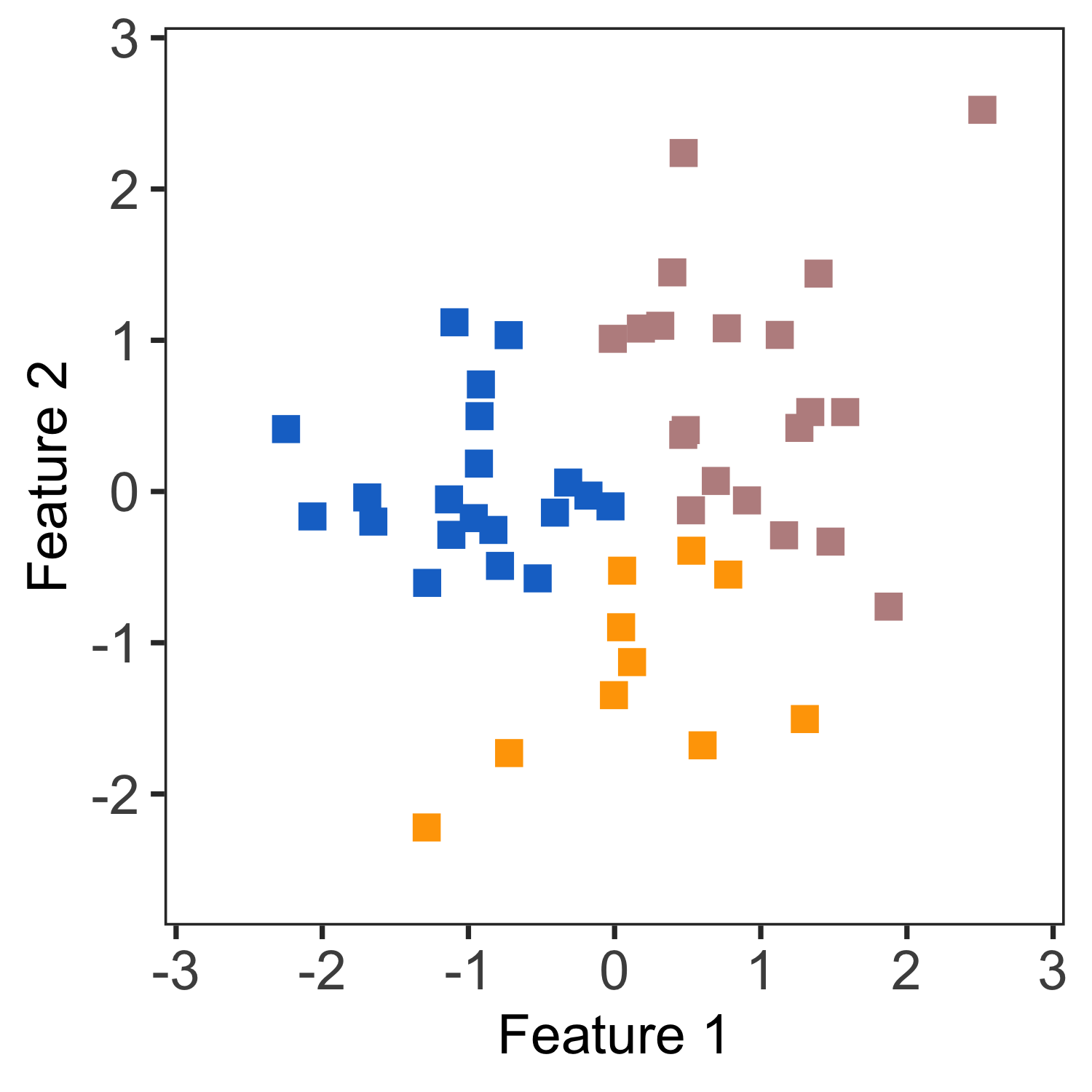}
  \includegraphics[width=.3\linewidth]{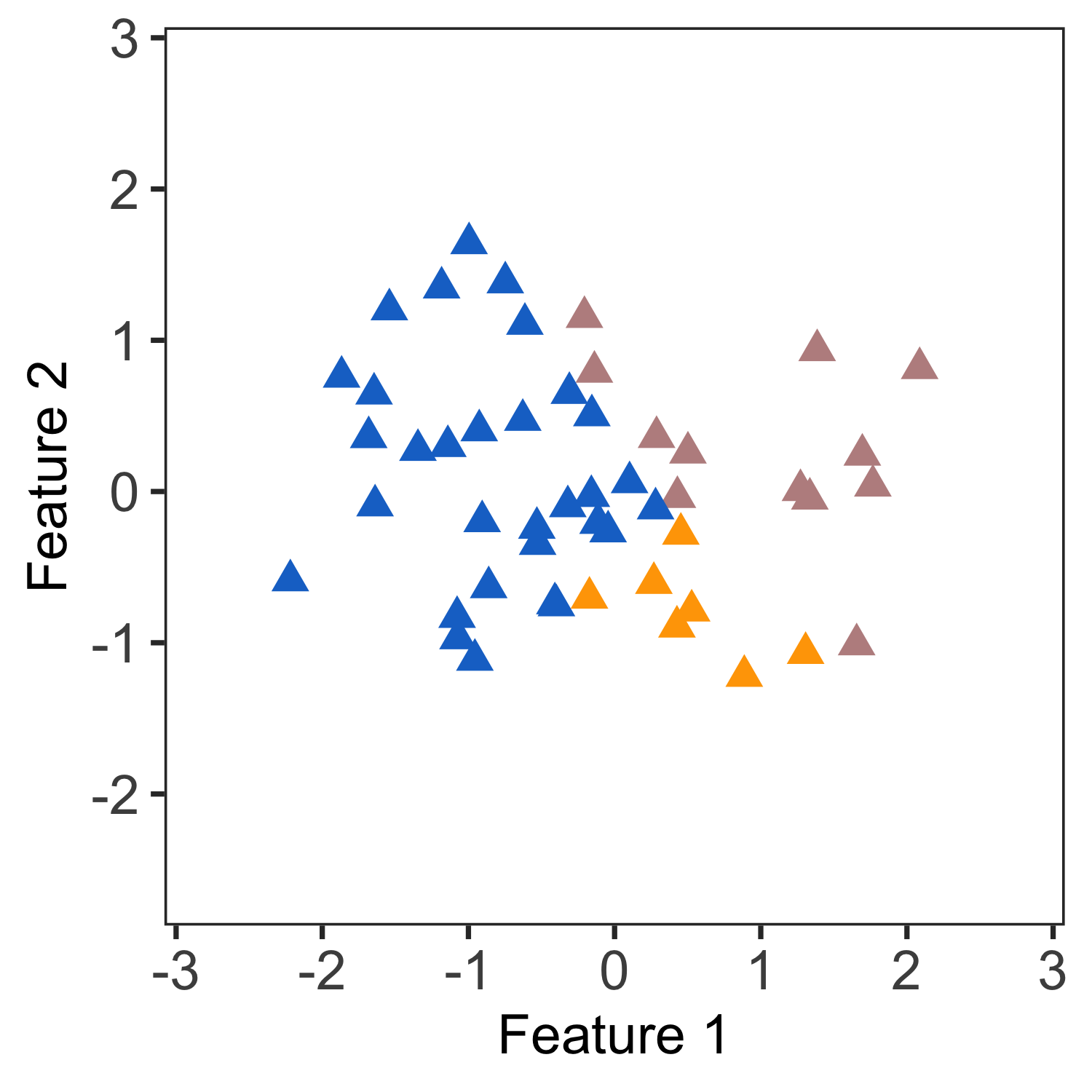}
  \includegraphics[width=.3\linewidth]{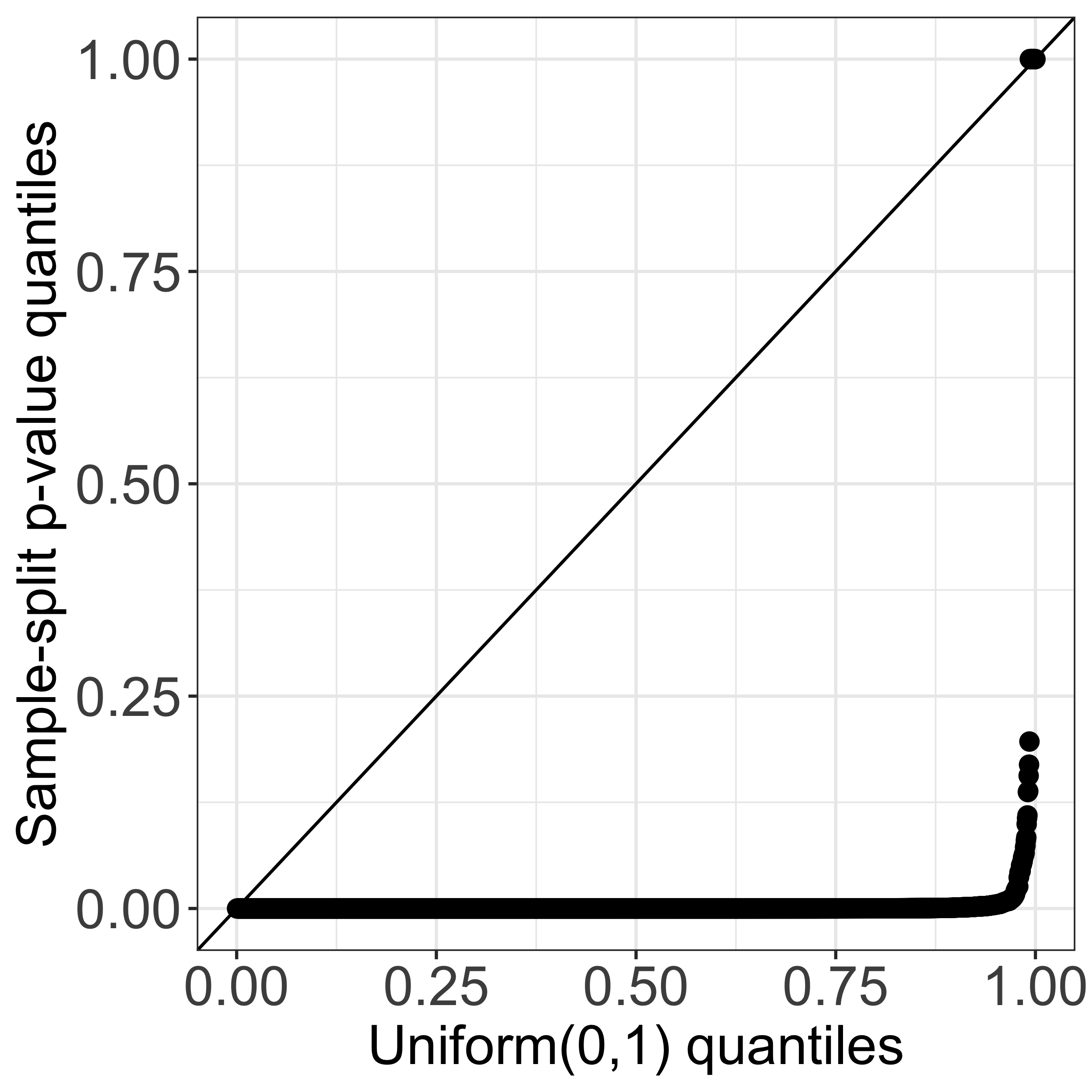}
\end{subfigure}
\vspace*{-4mm}
\caption{\textit{Left: } One simulated dataset generated from \eqref{eq:data_gen} with ${\mu} = {0}_{100\times 2}$ and $\sigma = 1$. We apply $k$-means clustering on the training set to obtain three clusters. \textit{Center: } We apply the training set clusters to the test set using a 3-nearest neighbors classifier. \textit{Right: } Quantile-quantile plot of the naive $p$-values \eqref{eq:wald_pval} applied to the test set, aggregated over 2,000 simulated datasets.}
\label{fig:sample_split}
\end{figure}

In this paper, we develop a test of $H_0$ that controls the selective Type I error. That is, we wish to ensure that the probability of rejecting $H_0$ at level $\alpha$, given that $H_0$ holds and we decided to test it, is no greater than $\alpha$:
\begin{align}
\label{eq:selective_type_1}
\pr_{H_0}\qty(\text{reject $H_0$ at level $\alpha$} \;\middle\vert\; H_0 \text{ is tested}   )\leq \alpha, \;\forall \alpha \in (0,1).
\end{align}
To develop the test, we leverage the selective inference framework, which has been applied extensively in high-dimensional linear modeling~\citep{Lee2016-te,Tibshirani2016-bx,Fithian2014-ow,Rugamer2022-yz,Schultheiss2021-zb,Taylor2018-mx,Charkhi2018-lz,Yang2016-km,Loftus2014-eq}, changepoint detection~\citep{jewell2019testing,Hyun2018-pe,Hyun2018-ta,Chen2021-qr,Le_Duy2021-iy,Duy2020-jw,Benjamini2019-yc}, and clustering~\citep{Zhang2019-lx,Gao2020-yt,Watanabe2021-dq}. The key insight behind selective inference is as follows: to obtain a valid test of $H_0$, we need to \emph{condition} on the aspect of the data that led us to test it. In our case, we chose to test the null hypothesis in \eqref{eq:null_clustering_mean} because $\hat{\mathcal{C}}_1$ and $\hat{\mathcal{C}}_2$ were obtained via $k$-means clustering. Therefore, we compute a $p$-value conditional on the event that $k$-means clustering yields $\hat{\mathcal{C}}_1$ and $\hat{\mathcal{C}}_2$. This results in selective Type I error control \eqref{eq:selective_type_1}, as seen in the right panel of Figure~\ref{fig:motivation}.

There is a rich literature on estimating and quantifying the uncertainty in the number of clusters~\citep{Li2010-tc,Chen2009-kn,Chen2004-th,McLachlan2019-ey,Dobriban2020-lm}, as well as assessing cluster stability and heterogeneity~\citep{Suzuki2006-rd,Kerr2001-ph,Kimes2017-qq,Chung2020-ia,Jin2016-iq,Aw2021-zv,Chung2015-le}. Others have examined the asymptotic properties of clustering models from a Bayesian perspective~\citep{Guha2019-ec,Nobile2004-kf,Cai2020-zp}. In addition, $k$-means clustering is a special case of the expectation-maximization algorithm, which allows us to 
tap into the active line of research on the statistical guarantees of the expectation-maximization algorithm~\citep{Zhang2014-fj,Wang2015-gy,Tony_Cai2019-ud,Yi2015-an,Balakrishnan2017-eg}. However, most prior work focused on scenarios where the number of clusters is correctly specified, and the estimated clusters memberships are close to the truth. By contrast, we are interested in a correctly-sized test for the null hypothesis \eqref{eq:null_clustering_mean}, even when $\hat{\mathcal{C}}_1,\hat{\mathcal{C}}_2$ do not correspond to true clusters. In addition, existing work often relies on asymptotic approximations and bootstrap resampling. Two recent exceptions include \citet{Zhang2019-lx} and~\citet{Gao2020-yt}, who took a selective inference approach and computed finite-sample $p$-values for testing the difference in means between estimated clusters obtained via linear classification rules and hierarchical clustering, respectively. Our work is closest to~\citet{Gao2020-yt}, and extends their framework to $k$-means clustering. We provide an exact, finite-sample test of the difference in means between a pair of clusters estimated via $k$-means clustering under model~\eqref{eq:data_gen}, without the need for sample splitting. 

Methods developed in this paper are implemented in the \texttt{R} package \texttt{KmeansInference}, available at \texttt{https://github.com/yiqunchen/KmeansInference}. Data and code for reproducing the results in this paper can be found at \texttt{https://github.com/yiqunchen/KmeansInference-experiments}. 


Throughout this paper, we will use the following notational conventions. For a matrix ${A}$, ${A}_i$ denotes the $i$th row and $A_{ij}$ denotes the $(i,j)$th entry. For a vector $\nu \in \mathbb{R}^n$, $\Vert\nu\Vert_2$ denotes its $\ell_2$ norm, and ${\Pi}_\nu^\perp$ is the projection matrix onto the orthogonal complement of $\nu$, i.e., ${\Pi}_\nu^\perp = \textbf{I}_n-\nu\nu^{\top}/\Vert\nu\Vert_2^2$, where $\textbf{I}_n$ is the $n$-dimensional identity matrix. Moreover, $\text{dir}(\nu) = \nu/\Vert \nu \Vert_2$ if $\nu \neq {0}_n$ and ${0}_n$ otherwise, where $0_n$ is the $n$-vector of zeros. We let $\langle \cdot, \cdot\rangle$ and $1\qty{\cdot}$ denote the inner product of two vectors and the indicator function, respectively.

\section{Selective inference for k-means clustering}
\label{section:hypothesis}


\subsection{A brief review of $k$-means clustering}
\label{section:k_means_review}
In this section, we review the $k$-means clustering algorithm. 
Given samples $x_1,\ldots,x_n \in \mathbb{R}^q$ and a positive integer $K$, $k$-means clustering partitions the $n$ samples into disjoint subsets $\hat{\mathcal{C}}_1,\ldots,\hat{\mathcal{C}}_K$  by solving the following optimization problem:
\begin{align}
\begin{split}
\label{eq:k_means_objective}
  &\underset{\mathcal{C}_1,\ldots,\mathcal{C}_K}{\text{minimize}}\;\qty{ \sum_{k=1}^K \sum_{i \in \mathcal{C}_k} \bigg\Vert x_i -  \sum_{i \in \mathcal{C}_k} x_i/|\mathcal{C}_k| \bigg\Vert_2^2 } \\
  &\text{subject to} \;
 \bigcup_{k=1}^K \mathcal{C}_k = \{1,\ldots, n\},\mathcal{C}_k\cap \mathcal{C}_{k'} = \emptyset ,\forall k\neq k'.
\end{split}
\end{align} 
It is not typically possible to solve for the global optimum in \eqref{eq:k_means_objective}~\citep{Aloise2009-on}. A number of algorithms are available to find a local optimum~\citep{Hartigan1979-up,Zha2002-cm,Arthur2007-zc}; one such approach is Lloyd’s algorithm~\citep{Lloyd2006-io}, given in Algorithm~\ref{algo:k_means_alt_min}. We first sample $K$ out of $n$ observations as initial centroids (step 1 in Algorithm~\ref{algo:k_means_alt_min}). We then assign each observation to the closest centroid (step 2). Next, we iterate between re-computing the centroids and updating the cluster assignments (steps 3a. and 3b.) until the cluster assignments stop changing. The algorithm is guaranteed to converge to a local optimum~\citep{Hastie2001-mr}.

In what follows, we will sometimes use $c_i^{(t)}({x})$ and $m_k^{(t)}({x})$ rather than $c_i^{(t)}$ and $m_k^{(t)}$ to emphasize the dependence of the cluster labels and centroids on the data ${x}$.

\begin{algorithm}
\DontPrintSemicolon 
\KwIn{Data $x_1,\ldots,x_n \in \mathbb{R}^q$, number of output clusters $K$, maximum iteration $T$, random seed $s$.}
\KwOut{Cluster assignments $\qty(c_1^{(t)},\ldots, c_n^{(t)})$.}
1. Initialize the centroids $\qty(m_1^{(0)},\ldots,m_K^{(0)})$ by sampling 
 $K$ observations from $x_1,\ldots, x_n$ without replacement, using the random seed $s$. \\ 
2. Compute assignments  $c_i^{(0)} \leftarrow \underset{1\leq k\leq K}{\text{argmin}} \left\Vert x_{i}-m_k^{(0)}\right\Vert_2^2,  \, i = 1,\ldots, n. $ \; \label{algo:step2}
3. Initialize $t=0$. \;
\While{$t\leq T$}{
  a. Update centroids:  $m_k^{(t+1)} \leftarrow \qty(\sum_{i:c_i^{(t)}=k} x_i)/\sum_{i=1}^n 1\qty{c_i^{(t)}=k}, \,  k = 1,\ldots, K. $\;
  b. Update assignment:  $c_i^{(t+1)} \leftarrow \underset{1\leq k\leq K}{\text{argmin}} \left\Vert x_{i}-m_k^{(t+1)}\right\Vert_2^2,  \, i = 1,\ldots, n. $\;
  c. \textbf{if} $c_i^{(t+1)} = c_i^{(t)}$ \emph{for all} $1\leq i \leq n$ \\
  \quad\quad \textbf{break} \\
  \textbf{else} \\
  \quad\quad $t \leftarrow t+1. $\;
}
\Return{$\qty(c_1^{(t)},\ldots, c_n^{(t)})$}.\;
\caption{Lloyd's algorithm for $k$-means clustering~\citep{Lloyd2006-io}}
\label{algo:k_means_alt_min}
\end{algorithm}

\subsection{A test of \eqref{eq:null_clustering_mean} for clusters obtained via $k$-means clustering}
\label{section:selective_k_means_review}
Here, we briefly review the proposal of \citet{Gao2020-yt} for selective inference for hierarchical clustering, and outline a selective test for \eqref{eq:null_clustering_mean} for $k$-means clustering.  

 \citet{Gao2020-yt} proposed a selective inference framework for testing hypotheses based on the output of a clustering algorithm. Let $\mathcal{C}(\cdot)$ denote the clustering operator, i.e., a partition of the observations resulting from a clustering algorithm. Since $H_0$ in \eqref{eq:null_clustering_mean} is chosen because $\qty{ \hat{\mathcal{C}}_1, \hat{\mathcal{C}}_2 \in \mathcal{C}({x})}$, where $\hat{\mathcal{C}}_1,\hat{\mathcal{C}}_2$ are the two estimated clusters under consideration in \eqref{eq:null_clustering_mean}, \citet{Gao2020-yt} proposed to reject $H_0$ if
\begin{align}
\pr_{H_0}\qty{ \Vert {X}^{\top}\nu \Vert_2 \geq  \Vert {x}^{\top}\nu  \Vert_2 \;\middle\vert\;  \hat{\mathcal{C}}_1, \hat{\mathcal{C}}_2 \in \mathcal{C}({X}),{\Pi}_{\nu}^\perp {X} = {\Pi}_{\nu}^\perp {x}, \text{dir}({X}^{\top}\nu) =  \text{dir}({x}^{\top}\nu) }
\label{eq:p_val_gao}
\end{align} 
is small. In \eqref{eq:p_val_gao}, conditioning on $\qty{{\Pi}_{\nu}^\perp {X} = {\Pi}_{\nu}^\perp {x}, \text{dir}({X}^{\top}\nu) =  \text{dir}({x}^{\top}\nu)}$ eliminates the nuisance parameters ${\Pi}_{\nu}^\perp {\mu}$ and $\text{dir}({\mu}^{\top}\nu)$, where ${\Pi}_{\nu}^\perp = \textbf{I}_n- \nu\nu^\top/\Vert \nu \Vert_2 $ and $\text{dir}({\mu}^{\top}\nu) = {\mu}^{\top}\nu/\Vert {\mu}^{\top}\nu \Vert_2$ (see, e.g., Section 3.1 of~\citet{Fithian2014-ow}).~\citet{Gao2020-yt} showed that the test that rejects $H_0$ when \eqref{eq:p_val_gao} is below $\alpha$ controls the selective Type I error at level $\alpha$, in the sense of \eqref{eq:selective_type_1}. Furthermore, under \eqref{eq:data_gen}, the conditional distribution of $\Vert {X}^{\top}\nu \Vert_2$ in \eqref{eq:p_val_gao} is $\qty(\sigma \Vert \nu\Vert_2) \chi_q$, truncated to a set. When the operator $\mathcal{C}(\cdot)$ denotes hierarchical clustering, this set can  be analytically characterized and efficiently computed, leading to an efficient algorithm for computing \eqref{eq:p_val_gao}.

We now extend these ideas to $k$-means clustering \eqref{eq:k_means_objective}. Since the $k$-means algorithm partitions all $n$ observations, it is natural to condition on the cluster assignments of \emph{all} observations rather than just on $\qty{ \hat{\mathcal{C}}_1, \hat{\mathcal{C}}_2 \in \mathcal{C}({X})}$. This leads to the $p$-value
{\small
\begin{equation}
\begin{aligned} 
\label{eq:p_val_k_means_ideal}
\pr_{H_0}\bigg[ \Vert {X}^{\top}\nu \Vert_2 \geq  \Vert {x}^{\top}\nu  \Vert_2 \;\bigg\vert\;   \bigcap_{i=1}^{n}\left\{c_i^{(T)}\left({X}\right) = c_i^{(T)}\left({x}\right)\right\},\,{\Pi}_{\nu}^\perp {X} = {\Pi}_{\nu}^\perp {x},\,
  \text{dir}({X}^{\top}\nu) =  \text{dir}({x}^{\top}\nu)\bigg],
\end{aligned} 
\end{equation}
}
where $c_i^{(T)}\left({X}\right)$ is the cluster assigned  to the $i$th observation at the final iteration of Algorithm~\ref{algo:k_means_alt_min}. However, computing \eqref{eq:p_val_k_means_ideal} requires characterizing $\bigcap_{i=1}^{n}\left\{c_i^{(T)}\left({X}\right) = c_i^{(T)}\left({x}\right)\right\}$, which is not straightforward, and may necessitate enumerating over possibly an exponential number of intermediate cluster assignments $c_i^{(t)}(\cdot)$ for $t=1,\ldots,T-1$. Hence, we also condition on \emph{all of the intermediate clustering assignments} in Algorithm~\ref{algo:k_means_alt_min}:
\begin{equation}
\begin{aligned} 
\label{eq:p_val_k_means_chen}
\pK =  \pr_{H_0}\bigg[ \Vert {X}^{\top}\nu \Vert_2 \geq  \Vert {x}^{\top}\nu  \Vert_2 \;\bigg\vert\;  \bigcap_{t=0}^{T} \bigcap_{i=1}^{n}\left\{c_i^{(t)}\left({X}\right) = c_i^{(t)}\left({x}\right)\right\},\,{\Pi}_{\nu}^\perp {X} = {\Pi}_{\nu}^\perp {x},\,  \\
\text{dir}({X}^{\top}\nu) =  \text{dir}({x}^{\top}\nu)\bigg].
\end{aligned} 
\end{equation}

In \eqref{eq:p_val_k_means_chen}, $c_i^{(t)}\left({X}\right)$ is the cluster assigned to the $i$th observation at the $t$th iteration of Algorithm~\ref{algo:k_means_alt_min}. Roughly speaking, this $p$-value answers the question:
\begin{quote}
Assuming that there is no difference between the population means of $\hat{\mathcal{C}}_1$ and $\hat{\mathcal{C}}_2$, what is the probability of observing such a large difference between their centroids, {among all the realizations of ${X}$ that yield identical results in every iteration of the $k$-means algorithm}?
\end{quote}
The $p$-value in \eqref{eq:p_val_k_means_chen} is the focus of this paper. We establish its key properties below. 

\begin{proposition}
\label{prop:single_param_p}
Suppose that ${x}$ is a realization from \eqref{eq:data_gen}, and let $\phi\sim (\sigma\Vert\nu\Vert_2)\chi_q$. Then, under $H_0: {\mu}^\top \nu = 0$ with $\nu$  defined in \eqref{eq:nu_def},
\begin{equation}
\label{eq:p_val_single_param}
\pK = \emph{\pr}\left[  \phi \geq \Vert{x}^\top \nu \Vert_2 \;\middle\vert\; \bigcap_{t=0}^{T}\bigcap_{i=1}^{n} \left\{c_i^{(t)}\left({x}'(\phi)\right) = c_i^{(t)}\left({x}\right)\right\}   \right],
\end{equation} 
where $\pK$ is defined in \eqref{eq:p_val_k_means_chen}, and 
 \begin{equation} {x}'(\phi) =  {x} + \qty( \phi-\Vert {x}^{\top}\nu  \Vert_2) \qty(\nu/\Vert \nu  \Vert_2^2 ) \left\{\emph{\text{dir}}({x}^{\top}\nu)\right\}^{\top}.
\label{eq:xphi}
\end{equation} 
Moreover, the test that rejects $H_0: {\mu}^\top \nu = 0$ when $\pK\leq \alpha$ controls the selective Type I error at level $\alpha$, in the sense of \eqref{eq:selective_type_1}.
\end{proposition}
 Proposition~\ref{prop:single_param_p} states that $\pK$ can be recast as the survival function of a scaled $\chi_q$ random variable, truncated to the set 
\begin{equation}
\label{eq:s_set}
  \mathcal{S}_T = \left\{ \phi\in \mathbb{R}: \bigcap_{t=0}^{T}\bigcap_{i=1}^{n} \left\{c_i^{(t)}\left({x}'(\phi)\right) = c_i^{(t)}\left({x}\right)\right\}   \right\},
\end{equation}
where ${x}'(\phi)$ is defined in \eqref{eq:xphi}. Therefore, to compute $\pK$, it suffices to characterize the set $ \mathcal{S}_T$. 
\noindent
\begin{figure}[htbp!]
\begin{subfigure}{\linewidth}
  \centering
  \includegraphics[width=.3\linewidth]{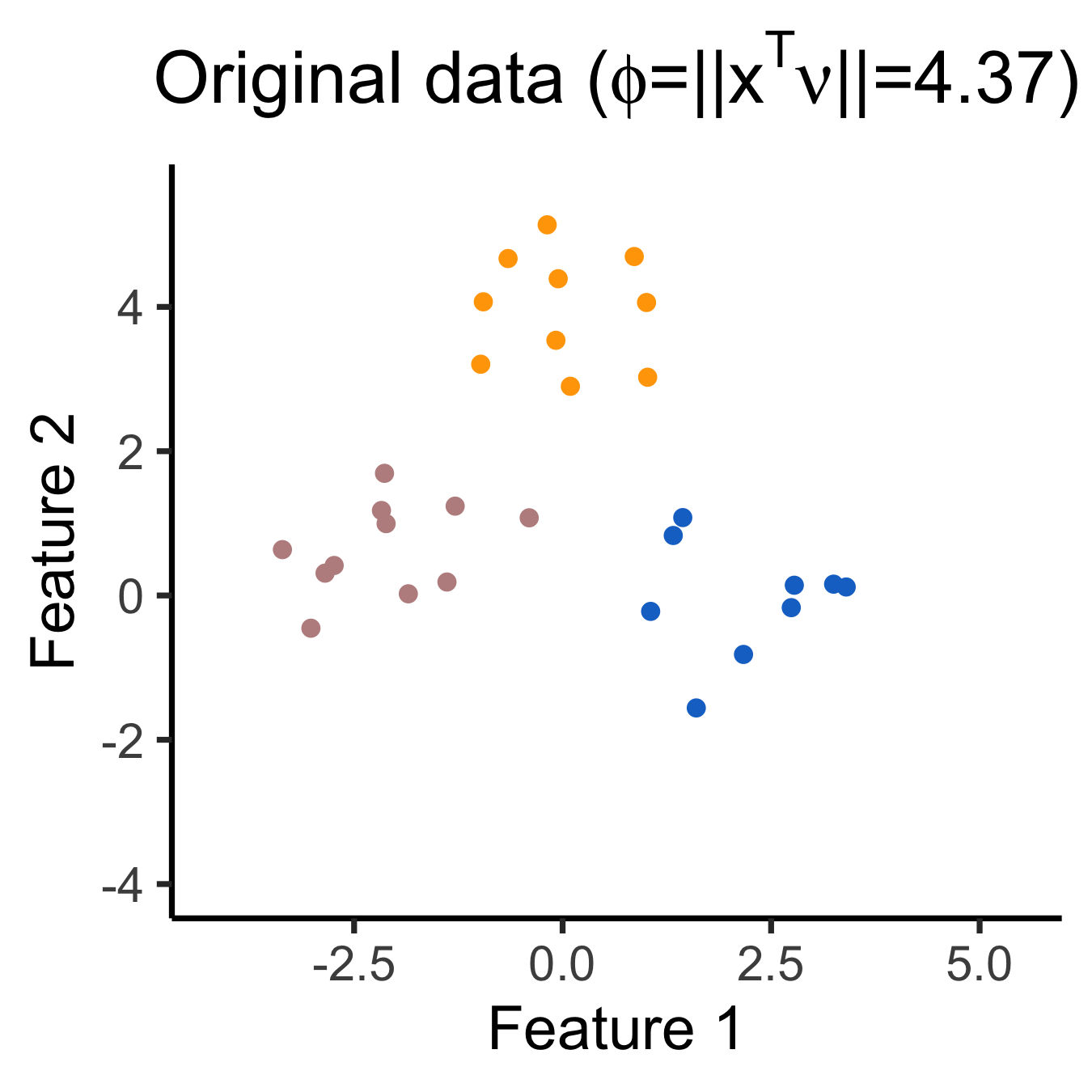}
  \includegraphics[width=.3\linewidth]{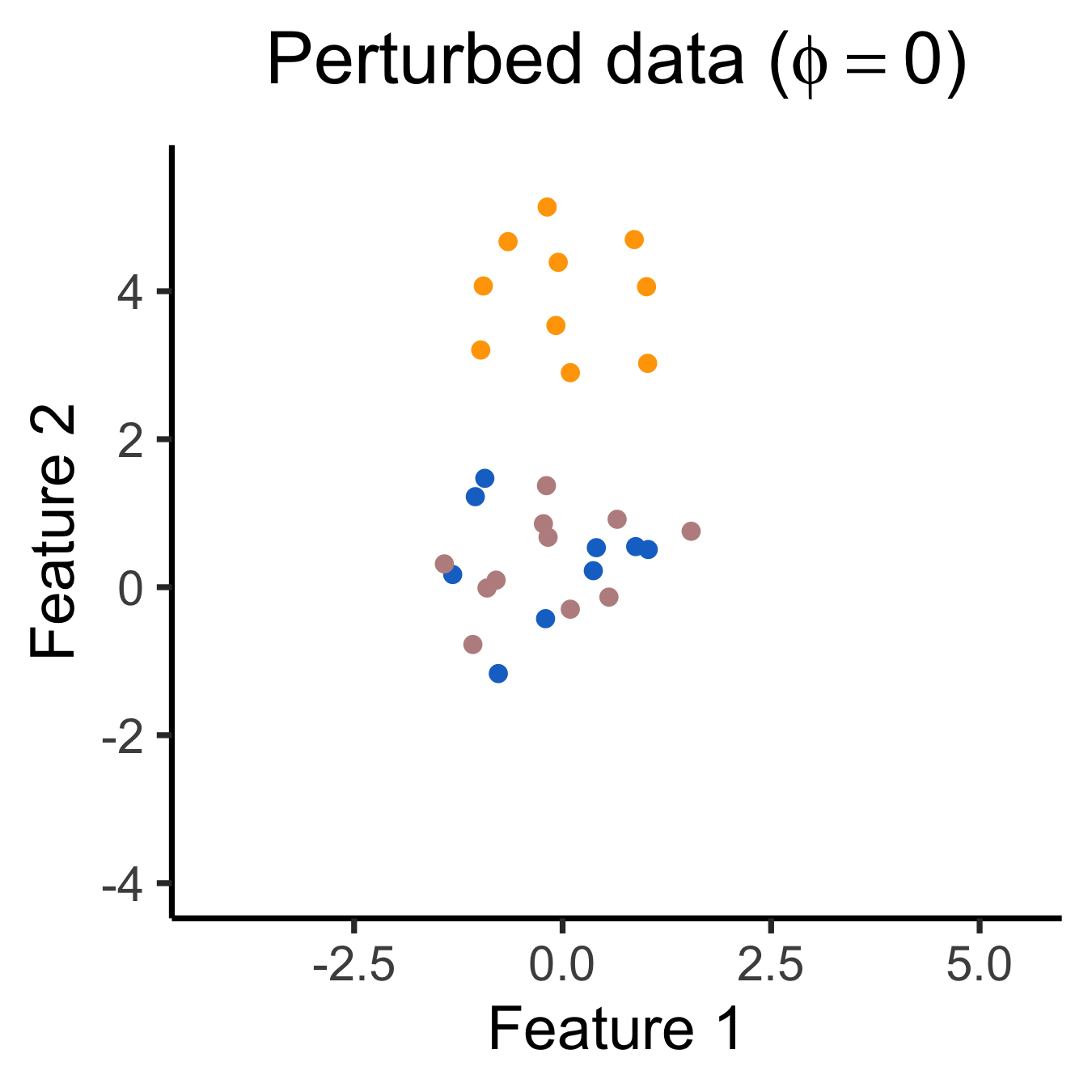}
  \includegraphics[width=.3\linewidth]{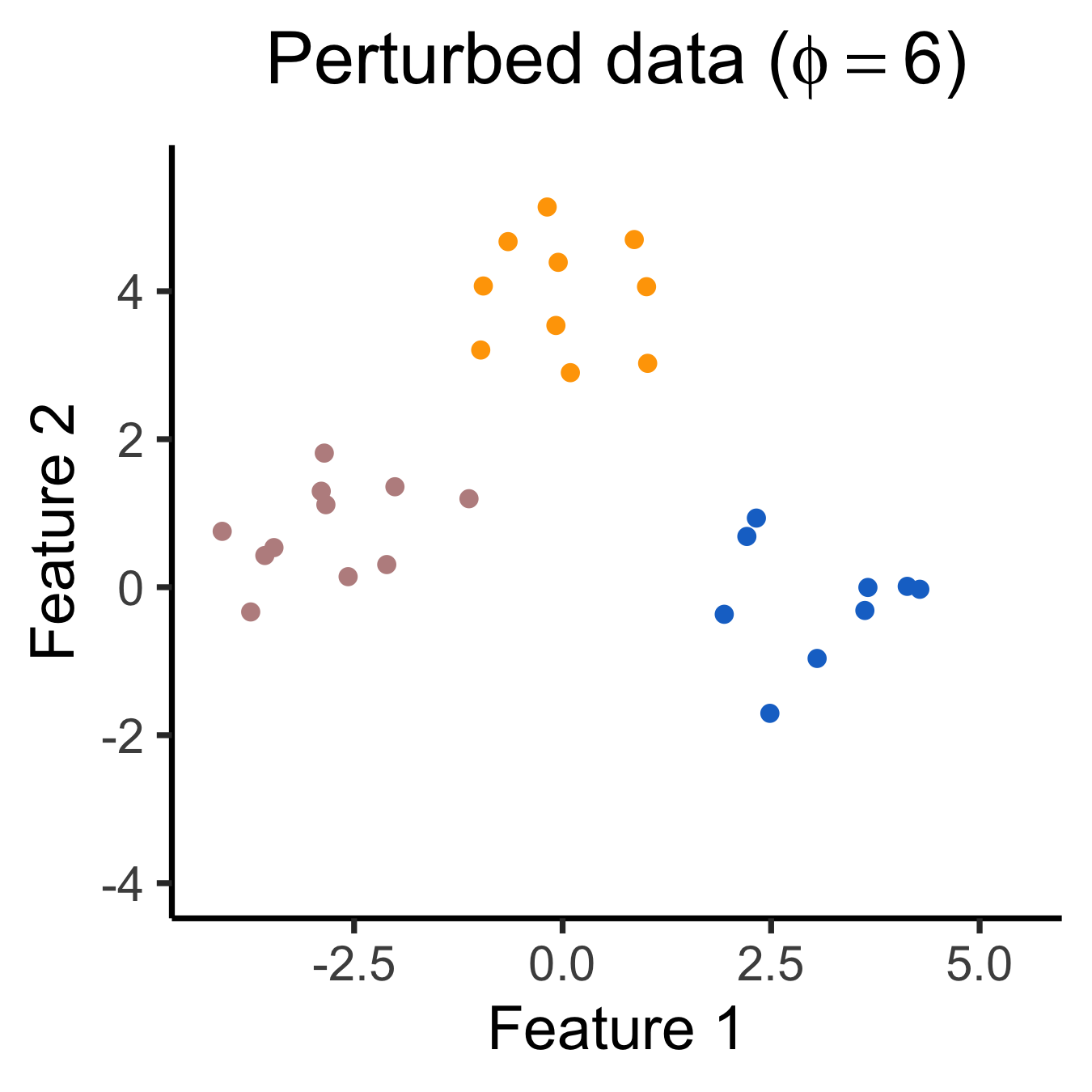}
\end{subfigure}
\vspace*{-4mm}
\caption{One simulated dataset generated from model \eqref{eq:data_gen} with ${\mu}_{i} = 1\qty{1\leq i\leq 10}[\begin{smallmatrix}2.5 \\ 0 \end{smallmatrix}] + 1\qty{11 \leq i \leq 20} [\begin{smallmatrix}0 \\ -2.5\end{smallmatrix}] + 1\qty{21\leq i\leq 30} [\begin{smallmatrix}\surd{18.75} \\ 0 \end{smallmatrix}]$ and $\sigma = 1$. \textit{Left: } The original data ${x}$ corresponds to $\phi = \Vert {x}^\top \nu \Vert_2 = 4.37$. Applying $k$-means clustering with $K=3$ yields three clusters, displayed in pink, blue, and orange. Here, $\nu$ is chosen to test for a difference in means between $\hat{\mathcal{C}}_1$ (pink) and $\hat{\mathcal{C}}_2$ (blue). \textit{Center: } The perturbed data ${x}'(\phi)$ with $\phi = 0$. Applying $k$-means clustering with $K=3$ does not yield the same set of clusters as in the left panel. \textit{Right: } The perturbed data ${x}'(\phi)$ with $\phi = 6$. Applying $k$-means clustering with $K=3$ yields the same set of clusters as in the left panel.}
\label{fig:x_phi}
\end{figure}
In \eqref{eq:xphi}, ${x}'(\phi)$ results from applying a perturbation to the observed data ${x}$, along the direction of ${x}^\top \nu$, the difference between the two cluster centroids of interest. Figure~\ref{fig:x_phi} illustrates a realization of \eqref{eq:data_gen} for $k$-means clustering with $K=3$. The left panel displays the observed data ${x}$, which corresponds to ${x}'(\phi)$ with $\phi=\Vert {x}^\top \nu \Vert_2 =4.37$. Here, $\nu$ defined in \eqref{eq:nu_def} was chosen to test the difference between $\hat{\mathcal{C}}_1$ (shown in pink) and $\hat{\mathcal{C}}_2$ (shown in blue). The center and right panels of Figure~\ref{fig:x_phi} display ${x}'(\phi)$ with $\phi = 0$ and $\phi = 6$, respectively. In the center panel, with $\phi=0$, the blue and pink clusters are ``pushed together'', resulting in $\Vert {x}'(\phi)^\top \nu \Vert_2 = 0$; that is, there is no difference in empirical means between the two clusters under consideration. Applying $k$-means clustering no longer results in these clusters. By contrast, in the right panel, with $\phi=6$, the blue and pink clusters are ``pulled apart'' along the direction of ${x}^\top \nu$, which results in an increased distance between the centroids of the blue and pink clusters, and $k$-means clustering does yield the same clusters as on the original data. In this example, $\mathcal{S}_T=(3.59,\infty)$.

\section{Computation of the selective p-value}
\label{section:method}

In Section~\ref{section:hypothesis}, we have shown that the $p$-value $\pK$ \eqref{eq:p_val_k_means_chen} involves the set $\mathcal{S}_T$ in \eqref{eq:s_set}.
Here, we start with a high-level summary of our approach to characterizing $\mathcal{S}_T$ in \eqref{eq:s_set}. First, we rewrite $\mathcal{S}_T$ as $\left\{ \phi\in\mathbb{R}:\bigcap_{i=1}^{n} \left\{c_i^{(0)}\left({x}'(\phi)\right) = c_i^{(0)}\left({x}\right)\right\}   \right\} \cap \left\{ \phi\in\mathbb{R}: \bigcap_{t=1}^{T}\bigcap_{i=1}^{n} \left\{c_i^{(t)}\left({x}'(\phi)\right) = c_i^{(t)}\left({x}\right)\right\} \right\}$. Next, we consider the first term in the intersection: according to step 2. of Algorithm~\ref{algo:k_means_alt_min}, for $i=1,\ldots,n$, $c_i^{(0)}\left({x}'(\phi)\right) = c_i^{(0)}\left({x}\right)$ if and only if for $i=1,\ldots, n$, the initial randomly-sampled centroid to which $\qty[{x}'(\phi)]_i$ is closest coincides with the initial centroid to which ${x}_i$ is closest. This condition can be expressed using $K-1$ inequalities. Furthermore, the same intuition holds for the second term in the intersection, except that the centroids are a function of the cluster assignments in the previous iteration. We formalize this intuition in Proposition~\ref{prop:phi_ineq}.

\begin{proposition}
\label{prop:phi_ineq}
Suppose that we apply the $k$-means clustering algorithm (Algorithm~\ref{algo:k_means_alt_min}) to a matrix ${x}\in\mathbb{R}^{n\times q}$, to obtain $K$ clusters in at most $T$ steps. Define
 \begin{align}
 \label{eq:w_i_t_k}
 w_{i}^{(t)}(k) = 1{\left\{c_{i}^{(t)}({x}) = k\right\}} /\sum_{i'=1}^n 1\qty{c_{i'}^{(t)}({x})= k}.
 \end{align}
Then, for the set $\mathcal{S}_T$ defined in \eqref{eq:s_set}, 
we have that
{\footnotesize
\begin{align}
&\mathcal{S}_T = \qty( \bigcap_{i=1}^{n}\bigcap_{k=1}^{K} \qty{\phi:  \left\Vert \qty[{x}'(\phi)]_{i} - m^{(0)}_{c_i^{(0)}({x})}\qty({x}'(\phi)) \right\Vert_2^2 \leq  \left\Vert \qty[{x}'(\phi)]_{i} - m^{(0)}_{k}\qty({x}'(\phi))  \right\Vert_2^2}) \cap \label{eq:S_T_quad_ineq_part_1} \\
&\hspace{-8mm}\qty(\bigcap_{t=1}^{T}\bigcap_{i=1}^{n}\bigcap_{k=1}^{K}\qty{ \phi: \left\Vert \qty[{x}'(\phi)]_i -  \sum_{i'=1}^n w_{i'}^{(t-1)}\qty(c_{i}^{(t)}({x}))\qty[{x}'(\phi)]_{i'}   \right\Vert_2^2
 \leq \left\Vert \qty[{x}'(\phi)]_i -  \sum_{i'=1}^n w_{i'}^{(t-1)}\qty(k)\qty[{x}'(\phi)]_{i'} \right\Vert_2^2  }) 
 \label{eq:S_T_quad_ineq_part_2}.
\end{align} }
\end{proposition}

Recall that $c_i^{(t)}({x})$ denotes the cluster to which the $i$th observation is assigned in step 3b. of Algorithm~\ref{algo:k_means_alt_min} during the $t$th iteration, and that $m^{(0)}_{k}\qty({x})$ denotes the $k$th centroid sampled from the data ${x}$ during step 1 of Algorithm~\ref{algo:k_means_alt_min}. In words, Proposition~\ref{prop:phi_ineq} says that $\mathcal{S}_T$ can be expressed as the intersection of $\mathcal{O}(nKT)$ sets. Therefore, it suffices to characterize the sets in \eqref{eq:S_T_quad_ineq_part_1} and \eqref{eq:S_T_quad_ineq_part_2}.
We now present two lemmas.

\begin{lemma}[Lemma 2 in \citet{Gao2020-yt}]
\label{lemma:simple_norm}
For $\nu$ in \eqref{eq:nu_def} and ${x}'(\phi)$ in \eqref{eq:xphi}, we have that 
{\small $\left\Vert\qty[{x}'(\phi)]_i- \qty[{x}'(\phi)]_j\right\Vert_2^2 = a \phi^2 + b \phi + \gamma$}, where {\small $a = \qty{(\nu_i-\nu_j)/\Vert\nu\Vert_2^2 }^2 $}, {\small $b=2\big[ (\nu_i-\nu_j)/\Vert\nu\Vert_2^2 \left\langle  {x}_i-{x}_j,  \emph{\text{dir}}({x}^{\top}\nu) \right\rangle-\qty{(\nu_i-\nu_j)/\Vert\nu\Vert_2^2 }^2 \Vert{x}^{\top}\nu\Vert_2 \big]$}, and {\small $\gamma=\left\Vert {x}_i  -  {x}_j -(\nu_i-\nu_j) ({x}^{\top}\nu)/||\nu||_2^2 \right\Vert_2^2$}. 
\end{lemma}
\begin{lemma}
\label{lemma:canonical_norm}
For $\nu$ in \eqref{eq:nu_def}, ${x}'(\phi)$ in \eqref{eq:xphi}, and $ w_{i}^{(t)}(k)$ in \eqref{eq:w_i_t_k},  we have that
{\small $\left\Vert \qty[{x}'(\phi)]_i - \sum_{i'=1}^n w_{i'}^{(t-1)}(k)\qty[{x}'(\phi)]_{i'} \right\Vert_2^2 = \tilde{a} \phi^2 + \tilde{b}\phi +\tilde{\gamma}$}, where 
 {\small$\tilde{a} = \qty(\nu_i - \sum_{i'=1}^n w_{i'}^{(t-1)}(k) \nu_{i'} )^2/\Vert\nu\Vert_2^4$}, 
{\small $ \tilde{b} =  \qty(2/\Vert\nu\Vert_2^2)\bigg\{  \qty(\nu_i - \sum_{i'=1}^n  w_{i'}^{(t-1)}(k) \nu_{i'} ) \left\langle {x}_i - \sum_{i'=1}^n  w_{i'}^{(t-1)}(k) {x}_{i'},\emph{\text{dir}}({x}^{\top}\nu) \right\rangle - \qty(\nu_i - \sum_{i'=1}^n w_{i'}^{(t-1)}(k) \nu_{i'} )^2 $\\
$(\Vert{x}^{\top}\nu \Vert_2)/\Vert\nu\Vert_2^4\bigg\}$} 
and {\small $\tilde{\gamma} =  \left\Vert {x}_i - \sum_{i'=1}^n w_{i'}^{(t-1)}(k)  {x}_{i'} - \qty(\nu_i - \sum_{i'=1}^n w_{i'}^{(t-1)}(k) \nu_{i'} )({x}^\top \nu)/\Vert\nu\Vert_2^2 \right\Vert_2^2$}.
\end{lemma}

 It follows from Lemmas~\ref{lemma:simple_norm} and \ref{lemma:canonical_norm} that all of the inequalities in \eqref{eq:S_T_quad_ineq_part_1} and \eqref{eq:S_T_quad_ineq_part_2} are in fact \emph{quadratic} in $\phi$, with coefficients that can be analytically computed. Therefore, computing the set $\mathcal{S}_T$ requires solving $\mathcal{O}(nKT)$ quadratic inequalities of $\phi$.

\begin{proposition}
\label{prop:time_complexity}
Suppose that we apply the $k$-means clustering algorithm (Algorithm~\ref{algo:k_means_alt_min}) to a matrix ${x}\in\mathbb{R}^{n\times q}$, to obtain $K$ clusters in at most $T$ steps. Then, the set $\mathcal{S}_T$ defined in \eqref{eq:s_set} can be computed in $\mathcal{O}\qty(nKT\qty(n+q)+nKT\log\qty(nKT))$ operations.
\end{proposition}


\section{Extensions}
\label{section:extension}

\subsection{Non-spherical covariance matrix}
\label{section:full_cov}

Thus far, we have assumed that the observed data ${x}$ is a realization of \eqref{eq:data_gen}, which implies that $\text{cov}({X}_i) = \sigma^2 \textbf{I}_q$. However, this assumption is often violated in practice. For example, expression levels of genes are highly correlated, and neighbouring pixels in an image tend to be more similar. For a known positive definite matrix ${\Sigma}$, we now let
\begin{align}
\label{eq:data_full_cov_gen}
    {X} \sim \mathcal{MN}_{n\times q}\left({\mu}, \mathbf{I}_n, {\Sigma} \right).
\end{align} 
Under \eqref{eq:data_full_cov_gen}, we can whiten the data by applying the transformation ${x}_i \to {\Sigma}^{-\frac{1}{2}}{x}_i$~\citep{Bell1997-op}, where ${\Sigma}^{-\frac{1}{2}}$ is the unique symmetric  positive definite square root of ${\Sigma}^{-1}$~\citep{Horn2012-vd}. Note that ${\Sigma}^{-\frac{1}{2}}{X}_i \sim \mathcal{N}({\Sigma}^{-\frac{1}{2}}{\mu}_i, \textbf{I}_q)$. Moreover, as ${\Sigma}^{-\frac{1}{2}}\succ 0$, testing the null hypothesis in \eqref{eq:null_clustering_mean} is equivalent to testing 
{ 
\begin{align}
\hspace{-5mm}
 H_0: \sum_{i\in {\hat{\mathcal{C}}}_1}{\Sigma}^{-\frac{1}{2}}{\mu}_i/|\hat{\mathcal{C}}_1| = \sum_{i\in \hat{\mathcal{C}}_2}{\Sigma}^{-\frac{1}{2}}{\mu}_i/|\hat{\mathcal{C}}_2|   \mbox{ versus }  H_1: \sum_{i\in \hat{\mathcal{C}}_1}{\Sigma}^{-\frac{1}{2}}{\mu}_i/|\hat{\mathcal{C}}_1| \neq \sum_{i\in \hat{\mathcal{C}}_2}{\Sigma}^{-\frac{1}{2}}{\mu}_i/|\hat{\mathcal{C}}_2|.
  \label{eq:null_mean_zca}
\end{align}
}
Therefore, to get a correctly-sized test under model~\eqref{eq:data_full_cov_gen}, we can simply carry out our proposal in Section~\ref{section:hypothesis} on the transformed data ${\Sigma}^{-\frac{1}{2}}{x}_i$ instead of the original data ${x}_i$. 

Instead of applying the whitening transformation, we can directly accommodate a known covariance matrix ${\Sigma}$ by considering the following extension of $\pK$ in \eqref{eq:p_val_k_means_chen}:
{
\begin{align} 
\label{eq:full_cov_p_val_k_means_chen}
\begin{split}
\pKSigma &= \pr_{H_0}\Big[ \Vert {\Sigma}^{-\frac{1}{2}}{X}^{\top}\nu \Vert_2 \geq  \Vert {\Sigma}^{-\frac{1}{2}}{x}^{\top}\nu  \Vert_2 \;\big\vert\;  \bigcap_{t=0}^{T} \bigcap_{i=1}^{n}\left\{c_i^{(t)}\left({X}\right) = c_i^{(t)}\left({x}\right)\right\},\, \\
&{\Pi}_{\nu}^\perp {X} = {\Pi}_{\nu}^\perp {x},\, \text{dir}\qty({\Sigma}^{-\frac{1}{2}}{X}^{\top}\nu) =  \text{dir}\qty({\Sigma}^{-\frac{1}{2}}{x}^{\top}\nu)\Big].
\end{split}
\end{align} 
}
\begin{proposition}
\label{prop:full_cov_single_param_p}
Suppose that ${x}$ is a realization from \eqref{eq:data_full_cov_gen}, and let $\phi\sim ( \Vert\nu\Vert_2)\chi_q$. Then, under $H_0: {\mu}^\top \nu = 0$ with $\nu$ defined in \eqref{eq:nu_def},
{\small
\begin{equation}
\label{eq:p_val_single_param_full_cov}
\hspace{-5mm}\pKSigma = \emph{\pr}\left[  \phi \geq \Vert{\Sigma}^{-\frac{1}{2}} {x}^\top \nu \Vert_2 \;\middle\vert\; \bigcap_{t=0}^{T}\bigcap_{i=1}^{n} \left\{c_i^{(t)}\left( {\Pi}^\perp_\nu{x} + \qty(\phi  \frac{\nu}{\Vert\nu\Vert_2^2})  \qty{\emph{\text{dir}}\qty({\Sigma}^{-\frac{1}{2}} {x}^\top \nu)}^\top {\Sigma}^{\frac{1}{2}} \right) = c_i^{(t)}\left({x}\right)\right\} \right],
\end{equation} 
}
where $\pKSigma$ is defined in \eqref{eq:full_cov_p_val_k_means_chen}. Furthermore, the test that rejects $H_0: {\mu}^\top \nu = 0$ when $\pKSigma\leq \alpha$ controls the selective Type I error at level $\alpha$.
\end{proposition}
In addition, we can adapt the results in Section~\ref{section:method} to compute the set 
{ $\qty{ \phi\in \mathbb{R}: \bigcap_{t=0}^{T}\bigcap_{i=1}^{n} \left\{c_i^{(t)}\left( {\Pi}^\perp_\nu{x} + \qty(\phi{\nu}/\Vert\nu\Vert_2^2) \qty{{\text{dir}}\qty({\Sigma}^{-\frac{1}{2}} {x}^\top \nu)}^\top {\Sigma}^{\frac{1}{2}} \right) = c_i^{(t)}\left({x}\right)\right\}}$} by modifying the results in Lemmas~\ref{lemma:simple_norm} and \ref{lemma:canonical_norm}. Details are in Section~\ref{appendix:proof_p_val_prop_full_cov} of the Appendix.
\subsection{Unknown variance}
\label{sec:unknown_var}
 When $\sigma$ is unknown, we can plug in an estimate $\hat\sigma$ in \eqref{eq:p_val_k_means_chen}:
\begin{equation}
\label{eq:p_val_chen_estimate_sigma}
\pKhat(\hat\sigma) = \pr\left[  \phi(\hat{\sigma}) \geq {\Vert{x}^\top \nu \Vert_2} \;\middle\vert\; \bigcap_{t=0}^{T}\bigcap_{i=1}^{n} \left\{c_i^{(t)}\left({x}'(\phi(\hat{\sigma}))\right) = c_i^{(t)}\left({x}\right)\right\}   \right],
\end{equation} 
where $\phi(\hat{\sigma})\sim (\hat\sigma\Vert \nu \Vert_2) \chi_q$. 
If we use a consistent estimator of $\sigma$, then a test based on the $p$-value in \eqref{eq:p_val_chen_estimate_sigma} provides selective Type I error control \eqref{eq:selective_type_1} asymptotically. 
\begin{proposition}
\label{prop:estimated_sigma}
For $q=1,2,\ldots,$ suppose that ${X}^{(q)} \sim \mathcal{MN}_{n\times q}\qty({\mu}^{(q)}, \emph{\textbf{I}}_{n}, \sigma^2 \emph{\textbf{I}}_{q})$. Let ${x}^{(q)}$ be a realization from ${X}^{(q)}$ and let $c_i^{(t)}(\cdot)$ be the cluster to which the $i$th observation is assigned during the $t$th iteration of step 3b. in Algorithm~\ref{algo:k_means_alt_min}. Consider the sequence of null hypotheses $H_0^{(q)}: {\mu^{(q)}}^\top \nu^{(q)} = 0_q$, where $\nu^{(q)}$ defined in \eqref{eq:nu_def} is the contrast vector resulting from applying $k$-means clustering on ${x}^{(q)}$. Suppose that (i) $\hat\sigma$ is a consistent estimator of $\sigma$, i.e., for all $\epsilon>0, \lim_{q\to\infty}\emph{\pr}\qty(|\hat\sigma({X}^{(q)})-\sigma|\geq \epsilon) = 0$; and (ii) there exists $\delta \in (0,1)$ such that $\lim_{q\to\infty}\emph{\pr}_{H_0^{(q)}}\qty[\bigcap_{t=0}^{T}\bigcap_{i=1}^{n} \left\{c_i^{(t)}\left({X}^{(q)}\right) = c_i^{(t)}\left({x}^{(q)}\right)\right\} ] > \delta$.
Then, for all $\alpha \in (0,1)$, we have that 
$\lim_{q\to\infty}\emph{\pr}_{H_0^{(q)}}\qty[\pKhat(\hat\sigma)\leq \alpha \;\middle\vert\; \bigcap_{t=0}^{T}\bigcap_{i=1}^{n} \left\{c_i^{(t)}\left({X}^{(q)}\right) = c_i^{(t)}\left({x}^{(q)}\right)\right\} ]= \alpha$.
\end{proposition}
In practice, we propose to use the following estimator of $\sigma$~\citep{Huber1981-aw}: 
\begin{align}
\label{eq:sigma_hat_MED}
\hat\sigma_{\text{MED}}({x}) = \qty{\underset{1\leq i\leq n, 1\leq j\leq q}{\text{median}}\qty(\tilde{{x}}_{ij}^2)/M_{\chi_1^2}}^{1/2},
\end{align}
 where $\tilde{{x}}$ is obtained from subtracting the median of each column in ${x}$, and $M_{\chi_1^2}$ is the median of the $\chi^2_1$ distribution. If ${\mu}$ is sparse, i.e., $\sum_{i=1}^n\sum_{j=1}^q 1\qty{\mu_{ij}\neq 0}$ is small, then \eqref{eq:sigma_hat_MED} is consistent with appropriate assumptions; see Appendix~\ref{appendix:variance_estimation}.

\section{Simulation study}
\label{section:sim}
\subsection{Overview}
Throughout this section, we consider testing the null hypothesis $H_0:{\mu}^\top \nu = {0}_q \mbox{ versus } H_1:{\mu}^\top \nu \neq {0}_q$, where, 
unless otherwise stated, $\nu$ defined in $\eqref{eq:nu_def}$  is based on a {randomly-chosen} pair of clusters $\hat{\mathcal{C}}_1$ and $\hat{\mathcal{C}}_2$ from $k$-means clustering. We consider four $p$-values: $p_{\text{Naive}}$ in \eqref{eq:wald_pval}, $\pK$ in \eqref{eq:p_val_k_means_chen}, $\pKhat$ in \eqref{eq:p_val_chen_estimate_sigma} with $\hat\sigma_{\text{MED}}$ defined in \eqref{eq:sigma_hat_MED}, and $\pKhat$ in \eqref{eq:p_val_chen_estimate_sigma} with $\hat\sigma_{\text{Sample}} = \qty{\sum_{i=1}^n\sum_{j=1}^q \qty({x}_{ij}-\bar{{x}}_j)^2/(nq-q)}^{1/2}$, where $\bar{{x}}_j=\sum_{i=1}^n {x}_{ij}/n$. In the simulations that follow, we compare the selective Type I error \eqref{eq:selective_type_1} and power of the tests that reject $H_0$ when these $p$-values are less than $\alpha=0.05$. 

\subsection{Selective Type I error under the global null}
\label{section:type_1_error}

We generate data from \eqref{eq:data_gen} with ${\mu} = {0}_{n\times q}$; therefore, $H_0$ in \eqref{eq:null_clustering_mean} holds for any pair of estimated clusters. We simulate 3,000 datasets with $n=150$,\,$\sigma=1$, and $q=2,10,50,100$.

For each simulated dataset, we apply $k$-means clustering with $K=3$, and then compute $p_{\text{Naive}}$, $\pK$, $\pKhat(\hat\sigma_{\text{MED}})$, and $\pKhat(\hat\sigma_{\text{Sample}})$ for a randomly-chosen pair of clusters. Figure~\ref{fig:sim_type_1} displays the observed $p$-value quantiles versus the Uniform(0,1) quantiles. We see that for all values of $q$, (i) the naive $p$-values in \eqref{eq:wald_pval} are stochastically smaller than a Uniform(0,1) random variable, and the test based on $p_{\text{Naive}}$ leads to an inflated Type I error rate; (ii) tests based $\pK$, $\pKhat(\hat\sigma_{\text{MED}})$, and $\pKhat(\hat\sigma_{\text{Sample}})$ control the selective Type I error rate in the sense of \eqref{eq:selective_type_1}.

\begin{figure}[htbp!]
\centering
\includegraphics[width=\linewidth]{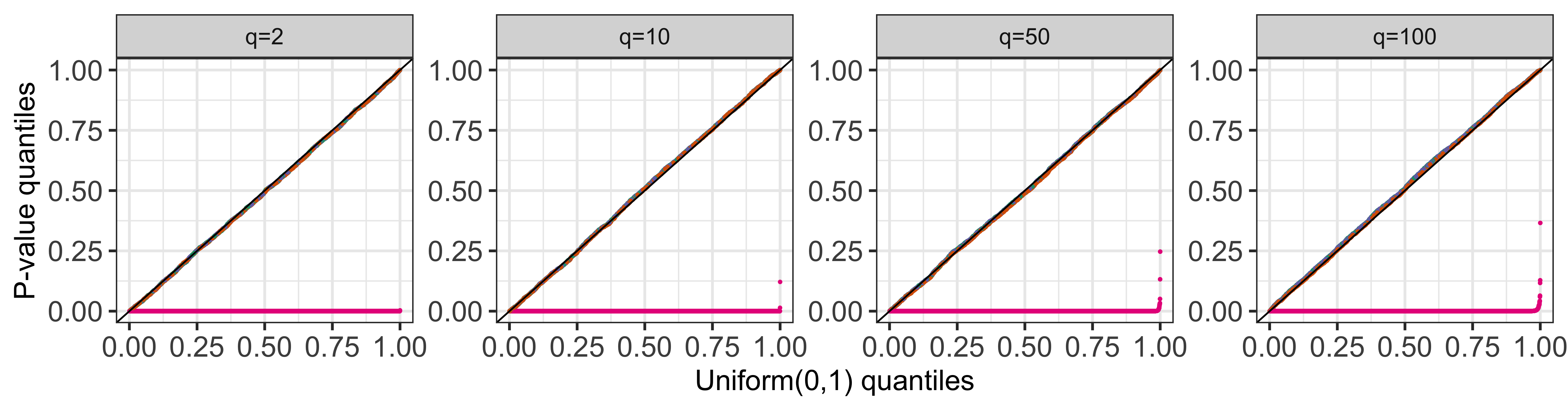}
\vspace*{-8mm}
\caption{Quantile-quantile plots for $p_{\text{Naive}}$ (pink), $\pK$ (green), $\pKhat(\hat\sigma_{\text{MED}})$ (orange), and $\pKhat(\hat\sigma_{\text{Sample}})$ (purple) under \eqref{eq:data_gen} with ${\mu} = {0}_{n\times q}$, stratified by $q$.}
\label{fig:sim_type_1}
\end{figure}

\subsection{Conditional power and detection probability}
\label{section:cond_power}
In this section, we show that the tests based on our proposal ($\pK$, $\pKhat(\hat\sigma_{\text{MED}})$, and $\pKhat(\hat\sigma_{\text{Sample}})$) have substantial power to reject $H_0$ when it is not true. We generate data from \eqref{eq:data_gen} with $n=150$ and 
\begin{align}
\label{eq:power_model}
\hspace{-5mm}
{\mu}_1 =\ldots = {\mu}_{\frac{n}{3}} = \begin{bmatrix}
-\frac{\delta}{2} \\ 0_{q-1}
\end{bmatrix}, \; {\mu}_{\frac{n}{3}+1}=\ldots = {\mu}_{\frac{2n}{3}} = \begin{bmatrix}
 0_{q-1} \\ \frac{\surd{3}\delta}{2}
\end{bmatrix} ,\;
{\mu}_{\frac{2n}{3}+1}=\ldots = {\mu}_{n} = \begin{bmatrix}
\frac{\delta}{2} \\ 0_{q-1}
\end{bmatrix}.
\end{align}
Here, we can think of $\mathcal{C}_1 = \{1,\ldots,n/3\},\mathcal{C}_2 = \{(n/3)+1,\ldots,(2n/3)\},\mathcal{C}_3 = \{(2n/3)+1,\ldots,n\}$ as the ``true clusters''. Moreover, these clusters are equidistant in the sense that the pairwise distance between each pair of population means is $|\delta|$.
Recall that we test $H_0$ in \eqref{eq:null_clustering_mean} for a pair of estimated clusters $\hat{\mathcal{C}}_1$ and $\hat{\mathcal{C}}_2$, which may not be true clusters. Hence, we will separately consider the \emph{conditional power} and \emph{detection probability} of our proposed tests \citep{Gao2020-yt,jewell2019testing,Hyun2018-pe}. The conditional power is the probability of rejecting $H_0$ in \eqref{eq:null_clustering_mean}, given that $\hat{\mathcal{C}}_1$ and $\hat{\mathcal{C}}_2$ are true clusters. Given $M$ simulated datasets, we estimate it as 
\begin{align}
\label{eq:conditional_power}
\text{Conditional power} = \frac{\sum_{m=1}^M 1\qty{ \qty{\hat{\mathcal{C}}^{(m)}_1,\hat{\mathcal{C}}^{(m)}_2} \subseteq \qty{\mathcal{C}_1,\ldots,  \mathcal{C}_L }, p^{(m)}\leq \alpha}}{\sum_{m=1}^M 1\qty{ \qty{\hat{\mathcal{C}}^{(m)}_1,\hat{\mathcal{C}}^{(m)}_2} \subseteq \qty{\mathcal{C}_1,\ldots,  \mathcal{C}_L } }},
\end{align}
where $\qty{\mathcal{C}_1,\ldots,  \mathcal{C}_L }$ are true clusters, and $p^{(m)}$ and $\hat{\mathcal{C}}^{(m)}_1, \hat{\mathcal{C}}^{(m)}_2$ correspond to the $p$-value and clusters under consideration for the $m$th simulated dataset. Because the quantity in \eqref{eq:conditional_power} conditions on the event that 
$\hat{\mathcal{C}}_1$ and $\hat{\mathcal{C}}_2$ are true clusters, we also estimate how often that event occurs:
\begin{align}
\label{eq:detect_p}
\text{Detection probability} = \sum_{m=1}^M 1\qty{ \{\hat{\mathcal{C}}^{(m)}_1,\hat{\mathcal{C}}^{(m)}_2\} \subseteq \qty{\mathcal{C}_1,\ldots,  \mathcal{C}_L }}/M.
\end{align}

\begin{figure}[htbp!]
\centering
\includegraphics[width=0.8\linewidth]{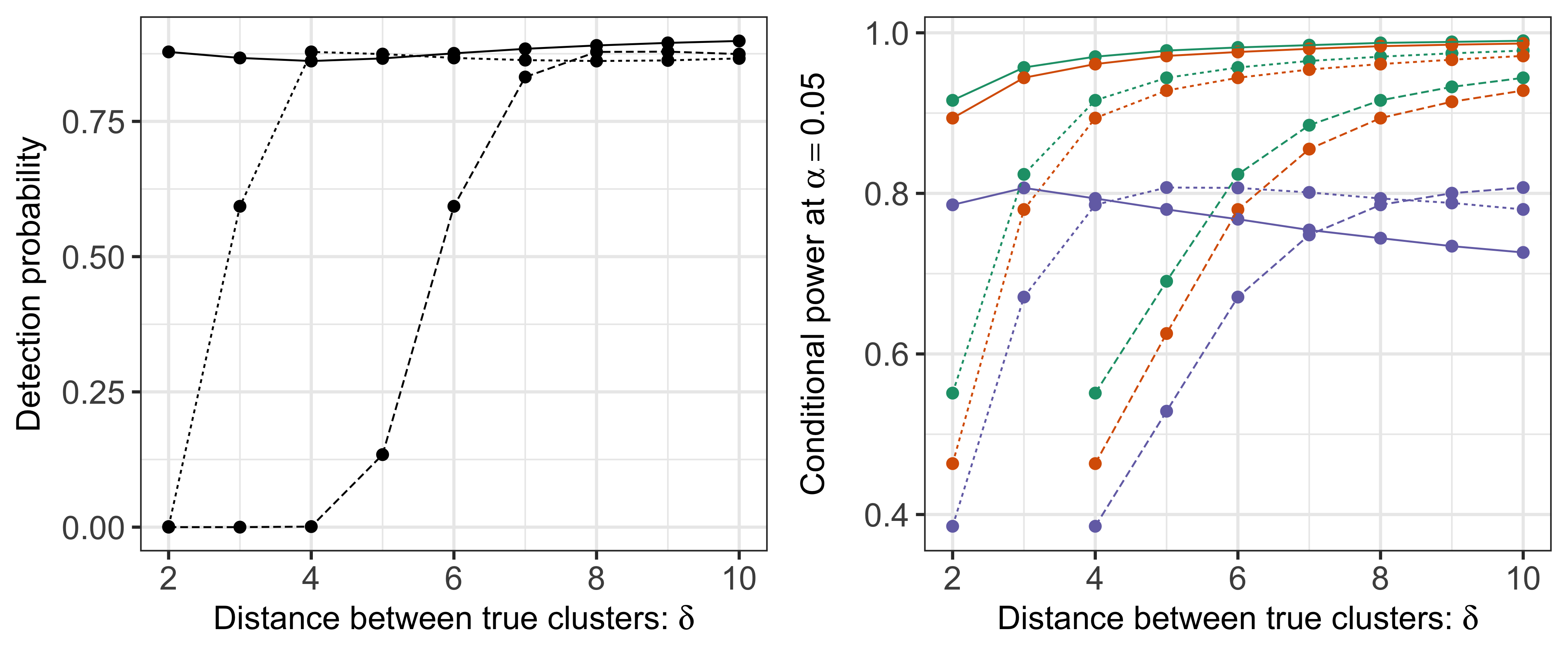}
\vspace*{-4mm}
\caption{\textit{Left: } The detection probability \eqref{eq:detect_p} for $k$-means clustering with $K=3$ under model \eqref{eq:data_gen} with $\mu$ defined in \eqref{eq:power_model}, and  $\sigma=0.25$ (solid lines), $0.5$ (dashed lines), and $1$ (long-dashed lines). 
\textit{Right: } The conditional power \eqref{eq:conditional_power} at $\alpha=0.05$ for the tests based on $\pK$ (green), $\pKhat(\hat{\sigma}_{\text{MED}})$ (orange), and $\pKhat(\hat{\sigma}_{\text{Sample}})$ (purple), under model \eqref{eq:data_gen} with $\mu$ defined in \eqref{eq:power_model} and $\sigma=0.25,0.5,1$. The conditional power is not displayed for $\delta=2,3, \sigma=1$ because the true clusters were never recovered in simulation.}
\label{fig:sim_power}
\end{figure}

We generate $M=200,000$ datasets from \eqref{eq:power_model} with $q=10,\sigma=0.25, 0.5,1,$ and $\delta=2,3,\ldots,10$. For each simulated dataset, we apply $k$-means clustering with $K=3$ and reject $H_0:{\mu}^\top \nu = 0_q$ if $\pK$, $\pKhat(\hat\sigma_{\text{MED}})$, or $\pKhat(\hat\sigma_{\text{Sample}})$ is less than $\alpha = 0.05$. In Figure~\ref{fig:sim_power}, the left panel displays the detection probability \eqref{eq:detect_p} of $k$-means clustering as a function of $\delta$ in \eqref{eq:power_model}, and the right panel displays the conditional power \eqref{eq:conditional_power} for the tests based on $\pK$, $\pKhat(\hat\sigma_{\text{MED}})$, and $\pKhat(\hat\sigma_{\text{Sample}})$. Under model \eqref{eq:data_gen}, the detection probability and conditional power increase as a function of $\delta$ in \eqref{eq:power_model} for all values of $\sigma$. For a given value of $\delta$, a larger value of $\sigma$ leads to lower detection probability and conditional power. The conditional power is not displayed for $\delta=2,3, \sigma=1$ because the true clusters were never recovered in simulation. Moreover, for a given value of $\delta$ and $\sigma$, the test based on $\pK$ has the highest conditional power, followed closely by the test based on $\pKhat(\hat{\sigma}_{\text{MED}})$. Using $\hat{\sigma}_{\text{Sample}}$ in $\pKhat$ leads to a less powerful test, especially for large values of $\delta$. This is because $\hat{\sigma}_{\text{Sample}}$ is a conservative estimator of $\sigma$ in \eqref{eq:data_gen}, and its bias is an increasing function of $\delta$, the distance between true clusters. By contrast, $\hat{\sigma}_{\text{MED}}$ is a consistent estimator under model \eqref{eq:power_model} (see Appendix~\ref{appendix:variance_estimation}). 

As an alternative to the conditional power in \eqref{eq:conditional_power}, in Appendix~\ref{appendix:additional_simulation}, we consider a notion of power that does not condition on having correctly estimated the true clusters.

\section{Real data applications}
\label{section:real_data}

\subsection{Palmer Penguins~\citep{palmerpenguins}}

Here we analyze the Palmer penguins dataset from the \texttt{palmerpenguins} package in \texttt{R}~\citep{palmerpenguins}. We consider the 165 female penguins with complete observations, and apply $k$-means clustering with $K=4$ to two of the collected features: bill depth and flipper length. Figure~\ref{fig:real_data_penguin} displays the estimated clusters.

\begin{figure}[htbp!]
\centering
\begin{minipage}{0.35\textwidth}
\includegraphics[trim={0 0mm 0 0mm},clip,width=\linewidth]{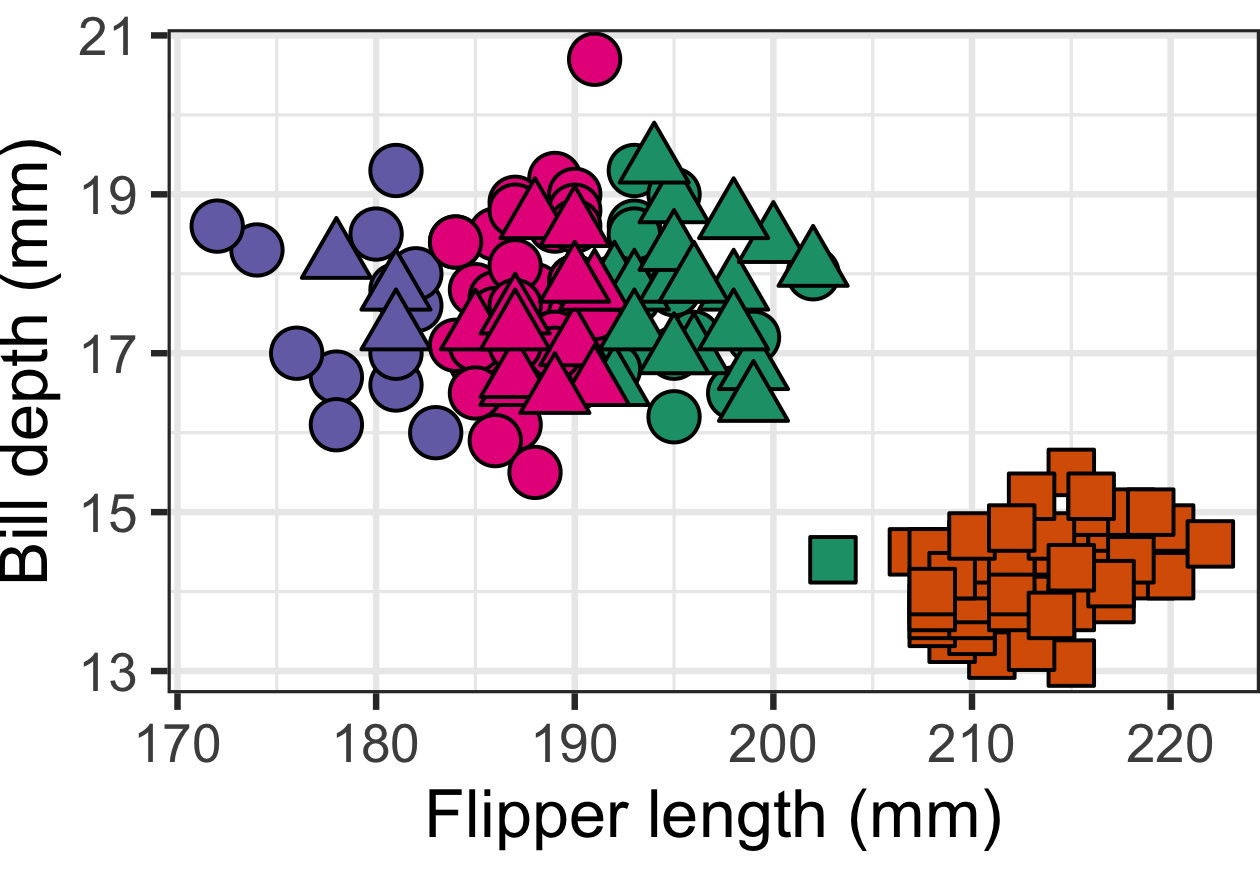}
\end{minipage}
\begin{minipage}[c]{0.40\linewidth}
\resizebox*{\linewidth}{0.14\textheight}{%
\begin{tabular}{ccccc}
  $H_0$  & $p_{\text{Naive}}$ & $\pKhat(\hat{\sigma}_{\text{MED}})$\\
$\bar{{\mu}}_1 = \bar{{\mu}}_2$ & $<10^{-10}$& $3.3\times 10^{-5}$ \\
$\bar{{\mu}}_1 = \bar{{\mu}}_3$ &  $<10^{-10}$& {$0.007$}  \\
$\bar{{\mu}}_1 = \bar{{\mu}}_4$ &  $<10^{-10}$ & 0.24 \\
$\bar{{\mu}}_2 = \bar{{\mu}}_3$ & $<10^{-10}$ & $2.2\times 10^{-6}$  \\
$\bar{{\mu}}_2 = \bar{{\mu}}_4$ &  $<10^{-10}$ &$7.8\times 10^{-6}$ \\
$\bar{{\mu}}_3 = \bar{{\mu}}_4$&  $<10^{-10}$ & 0.10  \\
\end{tabular}}
\end{minipage}
\vspace*{-4mm}
\caption{\textit{Left: } The bill depths and flipper lengths of female Palmer penguins, along with true species labels (Adelie: circle; Gentoo: square; Chinstrap: triangle) and clusters estimated using $k$-means clustering (cluster 1: green; cluster 2: orange; cluster 3: purple; cluster 4: pink).
\textit{Right: } We test the null hypothesis that the means of two estimated clusters are equal, for each pair of clusters estimated via $k$-means clustering, using $p_{\text{Naive}}$ in \eqref{eq:wald_pval} and $\pKhat(\hat{\sigma}_{\text{MED}})$ in \eqref{eq:p_val_chen_estimate_sigma} with $\hat{\sigma}_{\text{MED}}$ defined in \eqref{eq:sigma_hat_MED}. Here, $\bar{{\mu}}_i = \sum_{j\in\hat{\mathcal{C}}_i}{\mu}_j/|\hat{\mathcal{C}}_i|$. }
\label{fig:real_data_penguin}
\end{figure}

We assess the equality of the means of each pair of estimated clusters using $p_{\text{Naive}}$ in \eqref{eq:wald_pval} and $\pKhat(\hat{\sigma}_{\text{MED}})$ in \eqref{eq:p_val_chen_estimate_sigma} with $\hat{\sigma}_{\text{MED}}$ defined in \eqref{eq:sigma_hat_MED}. The results are in Figure~\ref{fig:real_data_penguin}. The naive p-values are small for all pairs of estimated clusters, even when the underlying species distributions are nearly identical (e.g., both clusters 1 and 4 are a mix of Chinstrap and Adelie penguins). By contrast, our proposal results in large $p$-values when testing for a difference in means between clusters composed of the same species (clusters 1 and 4, clusters 3 and 4), and small $p$-values when the clusters correspond to different species (e.g., clusters 1 and 2, clusters 2 and 3).  

\subsection{MNIST Dataset~\citep{Lecun1998-uw}}

In this section, we apply our method to the MNIST dataset~\citep{Lecun1998-uw}, which consists of 60,000 gray-scale images of handwritten digits. Each image has an accompanying label in $\{0,1,\ldots,9\}$, and is stored as a $28\times28$ matrix that takes on values in $[0,255]$. We first divide the entries of all the images by 255. Next, since there is no variation in the peripheral pixels of the images~\citep{Gallaugher2018-ew}, which violates model~\eqref{eq:data_gen}, we add an independent perturbation $\mathcal{N}(0,0.01)$ to each element of the image. Finally, we vectorize each image to obtain a vector $x_i \in \mathbb{R}^{784}$. 

\begin{figure}[htbp!]

\begin{minipage}{0.55\textwidth}
\begin{centering}
\hspace{10mm}  \\
\end{centering}
\begin{center}
\includegraphics[width=0.23\linewidth]{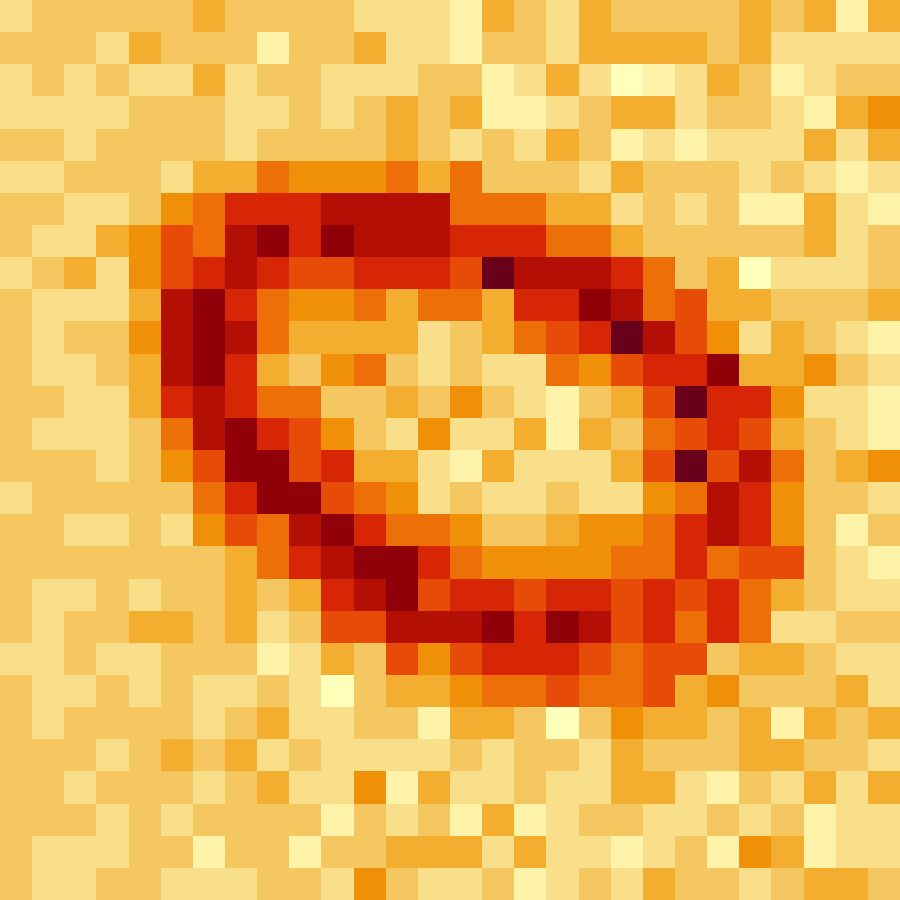}
\includegraphics[width=0.23\linewidth]{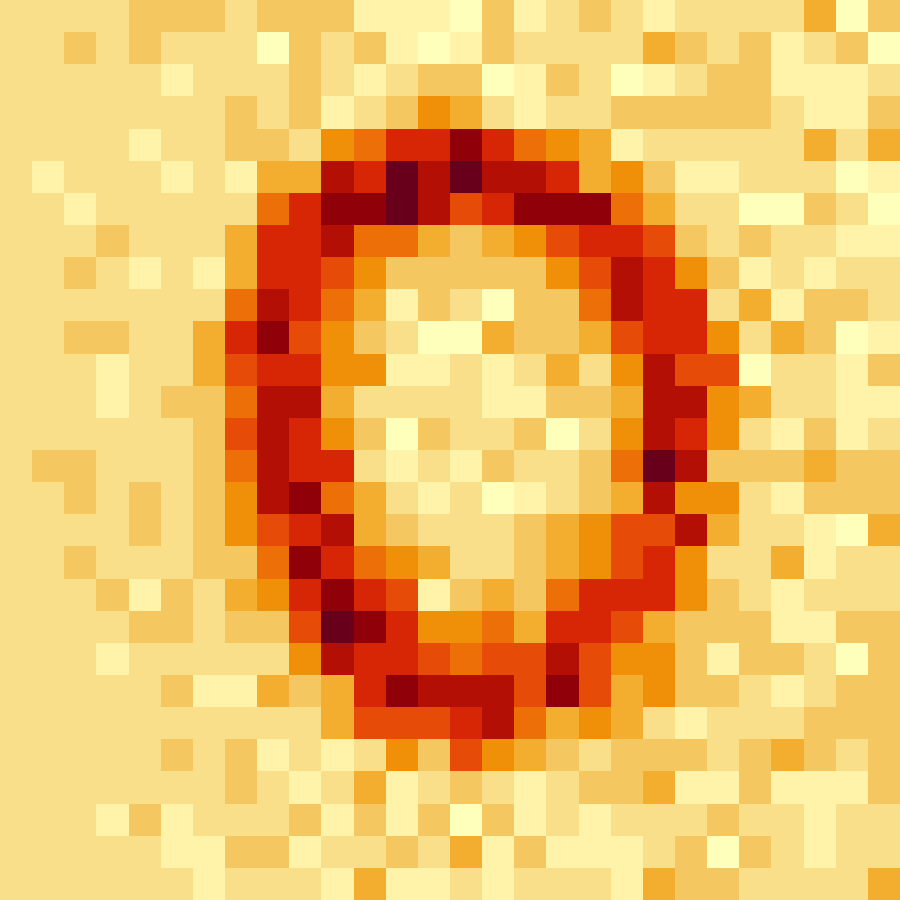}
\includegraphics[width=0.23\linewidth]{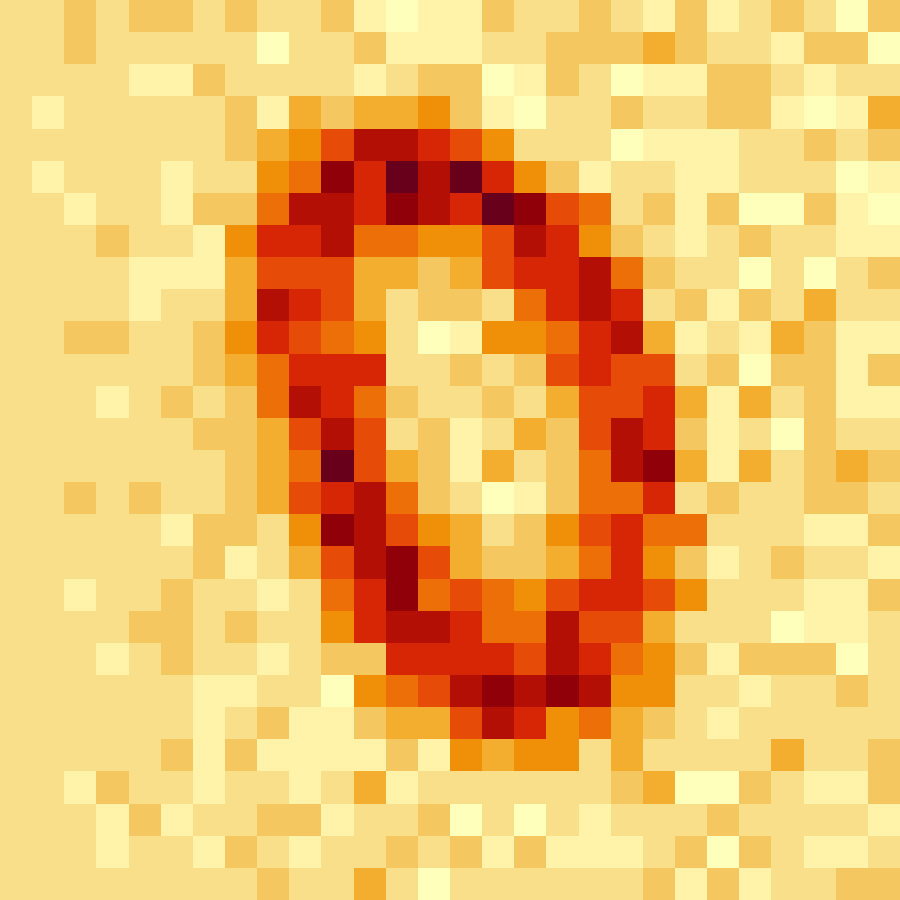} \\
\includegraphics[width=0.23\linewidth]{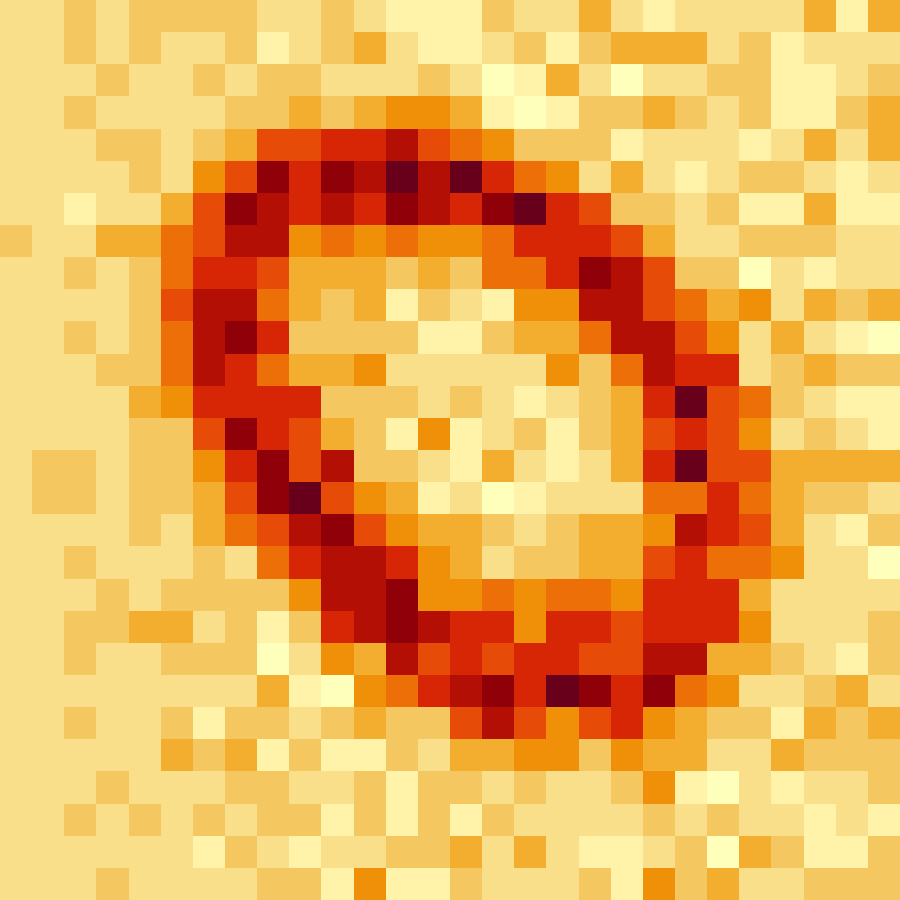}
\includegraphics[width=0.23\linewidth]{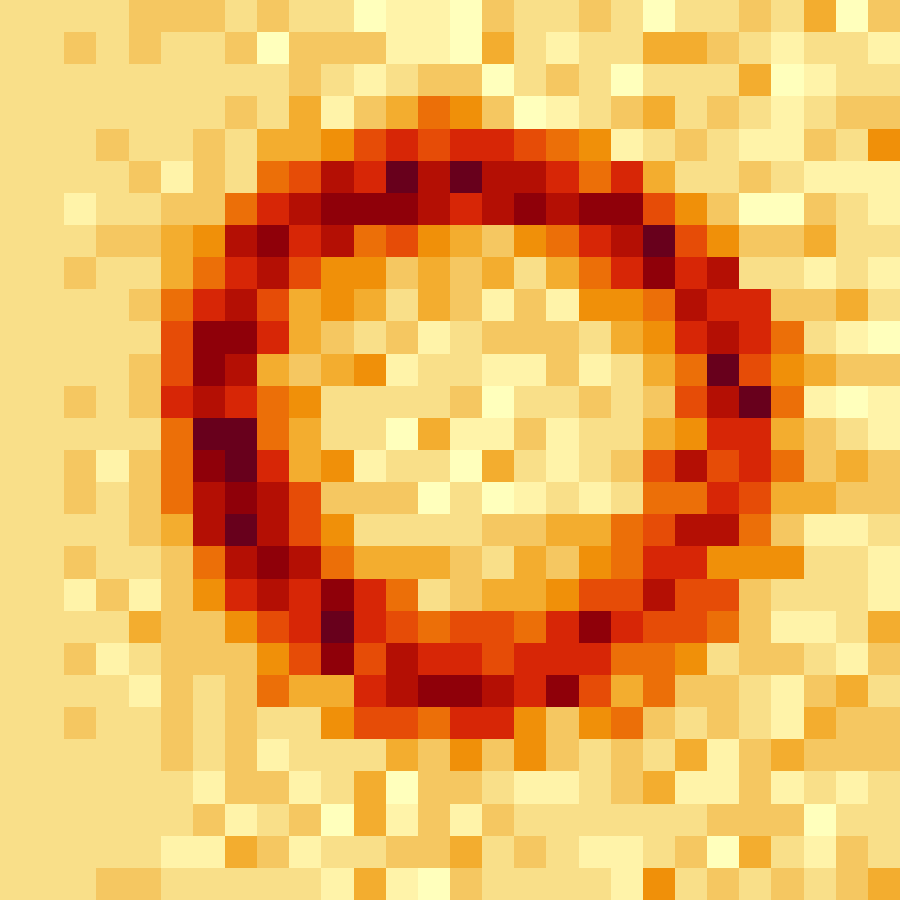}
\includegraphics[width=0.23\linewidth]{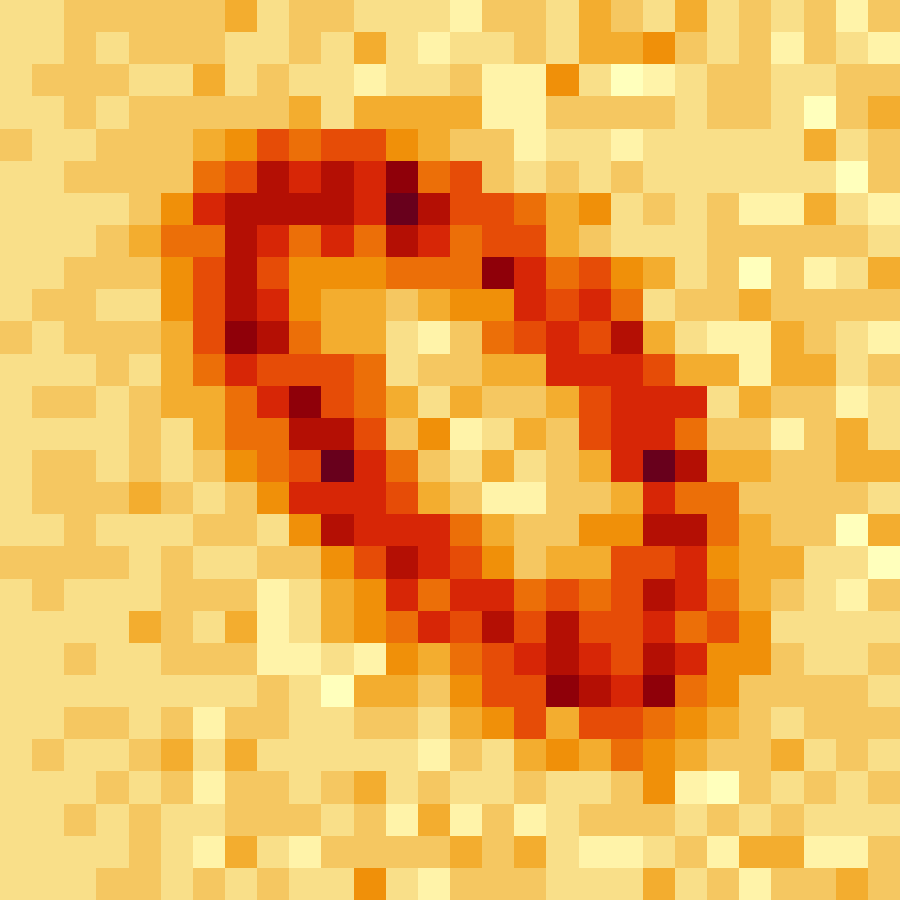} \\
\begin{centering}
\hspace{10mm}  \\
\end{centering}
\includegraphics[width=0.23\linewidth]{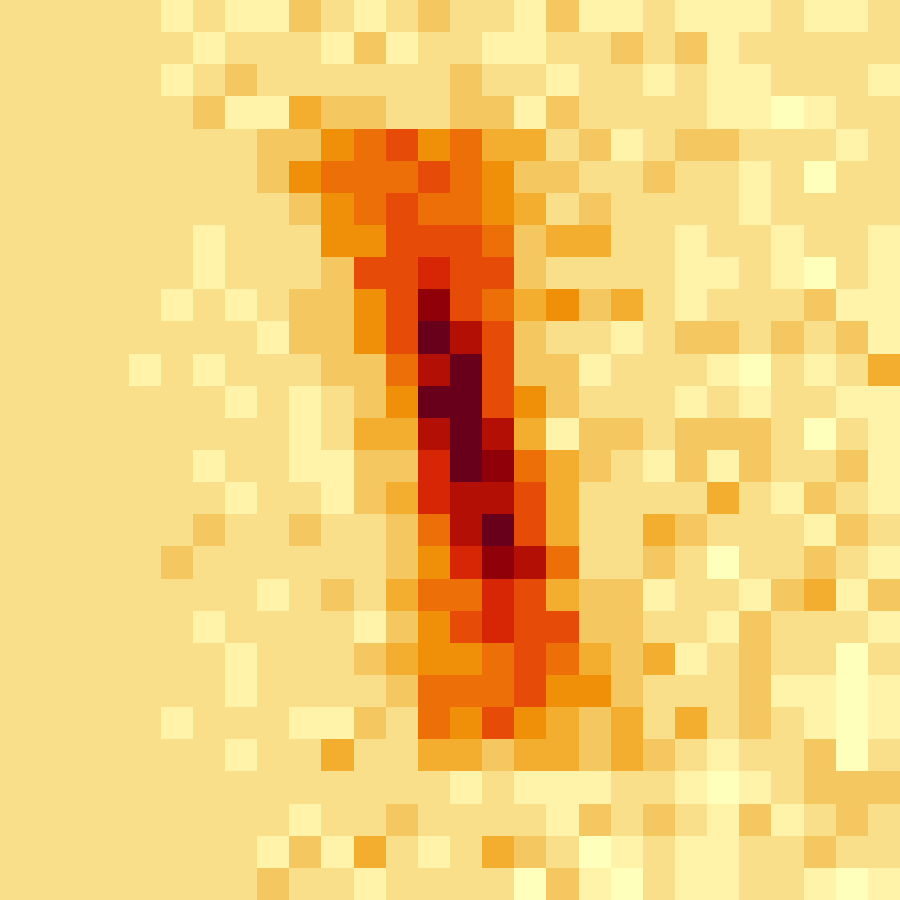}
\includegraphics[width=0.23\linewidth]{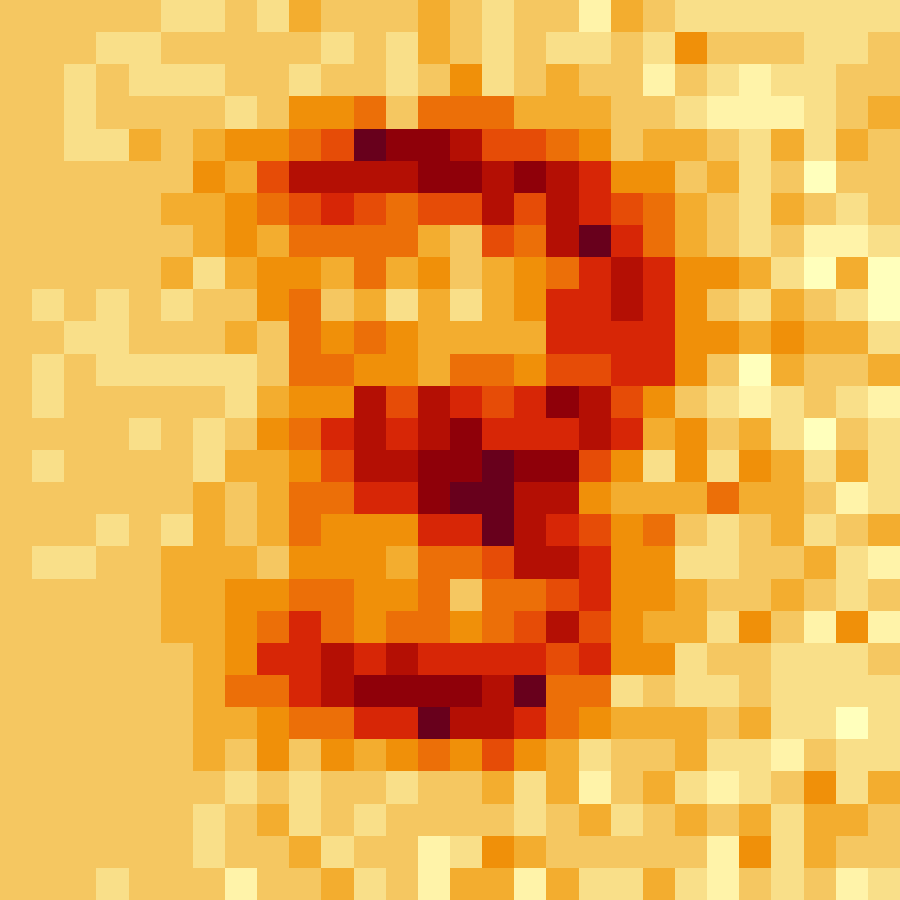}
\includegraphics[width=0.23\linewidth]{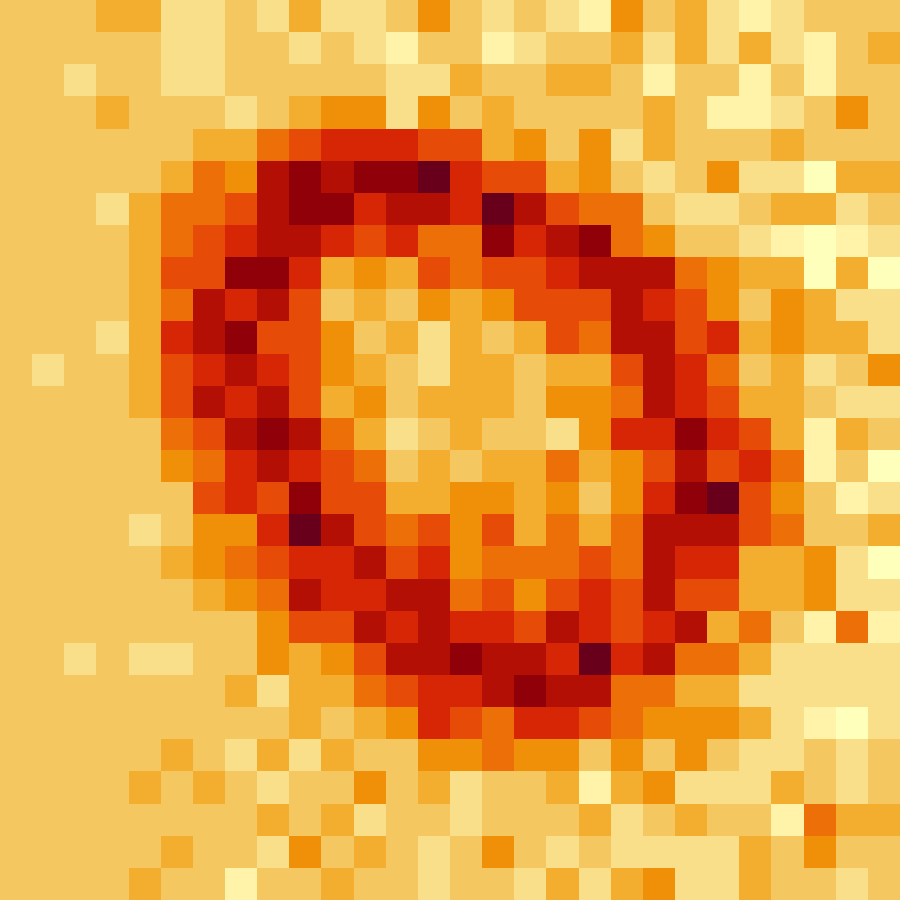}
\includegraphics[width=0.23\linewidth]{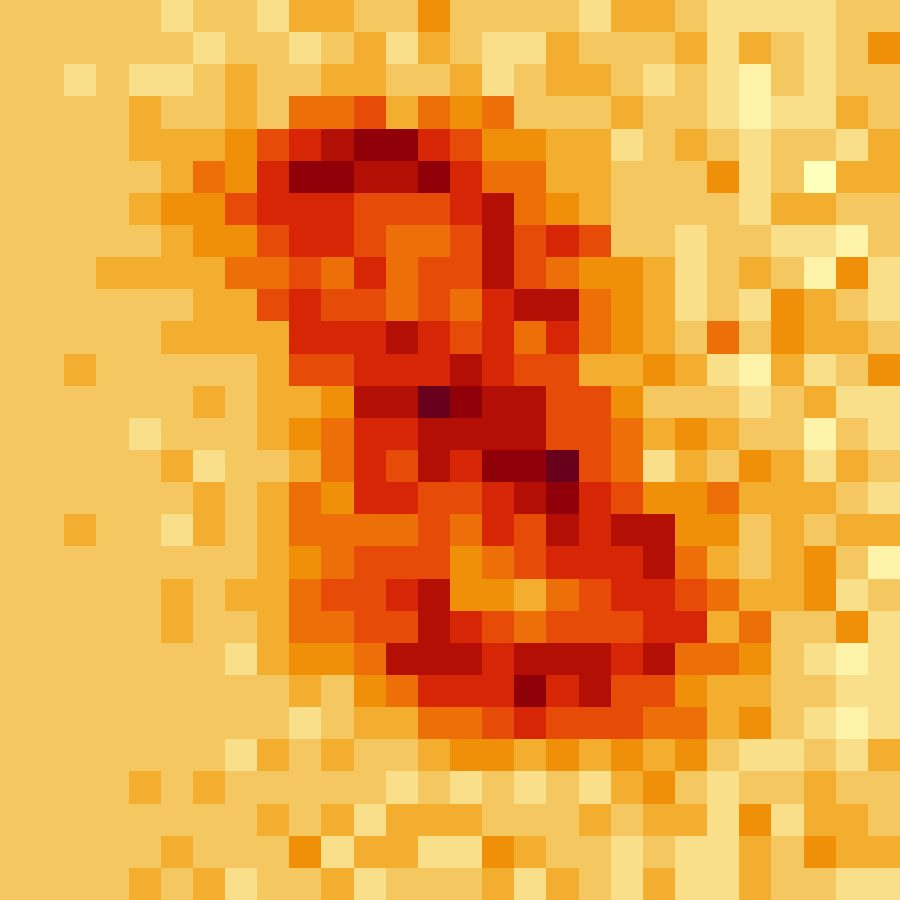}
\end{center}
\end{minipage}
\begin{minipage}{0.45\textwidth}
\resizebox*{\linewidth}{0.33\textheight}{%
\begin{tabular}{cccccc}
 & $H_0$  & $p_{\text{Naive}}$ & $\pKhat(\hat{\sigma}_{\text{MED}})$\\
``No cluster''& & & \\
&$\bar{{\mu}}_1 = \bar{{\mu}}_2$ & $<10^{-10}$ &  $<10^{-10}$\\
&$\bar{{\mu}}_1 = \bar{{\mu}}_3$ & $<10^{-10}$& $0.01$ \\
&$\bar{{\mu}}_1 = \bar{{\mu}}_4$ & $<10^{-10}$& $0.22$ \\
&$\bar{{\mu}}_1 = \bar{{\mu}}_5$ & $<10^{-10}$& $<10^{-10}$ \\
&$\bar{{\mu}}_1 = \bar{{\mu}}_6$ &  $<10^{-10}$& {$0.53$}  \\
&$\bar{{\mu}}_2 = \bar{{\mu}}_3$&  $<10^{-10}$ &  $<10^{-10}$\\
&$\bar{{\mu}}_2 = \bar{{\mu}}_4$&  $<10^{-10}$ & $1.4\times10^{-8}$ \\
&$\bar{{\mu}}_2 = \bar{{\mu}}_5$ & $<10^{-10}$ & $<10^{-10}$ \\
&$\bar{{\mu}}_2 = \bar{{\mu}}_6$&  $<10^{-10}$ & $0.49$\\
&$\bar{{\mu}}_3 = \bar{{\mu}}_4$&  $<10^{-10}$ &  $0.14$\\
&$\bar{{\mu}}_3 = \bar{{\mu}}_5$&  $<10^{-10}$ &  $1.1\times 10^{-6}$ \\
&$\bar{{\mu}}_3 = \bar{{\mu}}_6$&  $<10^{-10}$ & $0.45$ \\
&$\bar{{\mu}}_4 = \bar{{\mu}}_5$&  $<10^{-10}$ & $2.0\times10^{-8}$ \\
&$\bar{{\mu}}_4 = \bar{{\mu}}_6$ &  $<10^{-10}$ & $0.77$\\
& $\bar{{\mu}}_5 = \bar{{\mu}}_6$ &  $<10^{-10}$ & $0.48$ \\
``Cluster'' & & & & \\
&$\bar{{\mu}}_1 = \bar{{\mu}}_2$ & $<10^{-10}$ & $<10^{-10}$ \\
&$\bar{{\mu}}_1 = \bar{{\mu}}_3$ & $<10^{-10}$& $8.0\times10^{-6}$ \\
&$\bar{{\mu}}_1 = \bar{{\mu}}_4$ &  $<10^{-10}$& {$6.2\times 10^{-7}$}  \\
&$\bar{{\mu}}_2 = \bar{{\mu}}_3$&  $<10^{-10}$ & $6\times10^{-3}$ \\
&$\bar{{\mu}}_2 = \bar{{\mu}}_4$&  $<10^{-10}$ & $10^{-3}$ \\
&$\bar{{\mu}}_3 = \bar{{\mu}}_4$ &  $<10^{-10}$ & $4\times10^{-4}$\\
\end{tabular}}
\end{minipage}
\vspace*{-4mm}
\caption{\textit{Top left: } Centroids of six clusters from the ``no cluster'' dataset ($\hat{\mathcal{C}}_1$ to $\hat{\mathcal{C}}_6$ from left to right, top to bottom). \textit{Bottom left: } Same as top left, but for the ``cluster'' dataset. \textit{Right: } We test the null hypothesis of no difference between each pair of cluster centroids using $p_{\text{Naive}}$ and $\pKhat(\hat{\sigma}_{\text{MED}})$. Here, $\bar{{\mu}}_i = \sum_{j\in\hat{\mathcal{C}}_i}{\mu}_j/|\hat{\mathcal{C}}_i|$. }
\label{fig:mnist_data}
\end{figure}

We first construct a ``no cluster'' dataset by randomly sampling 1,500 images of the 0s; thus, $n=1,500$ and $q=784$. To de-correlate the pixels in each image, we whitened the data (see Section~\ref{section:full_cov}) using $\hat{{\Sigma}}^{-\frac{1}{2}} = {U}\qty({\Lambda}+0.01 \textbf{I}_n)^{-\frac{1}{2}}{U}^\top$ as in prior work~\citep{Coates2012-ag}, where ${U}{\Lambda}{U}^\top$ is the eigenvalue decomposition of the sample covariance matrix. 

We apply $k$-means clustering with $K=6$. The centroids are displayed in the top left panel of Figure~\ref{fig:mnist_data}. For each pair of estimated clusters, we compute the $p$-values $p_{\text{Naive}}$ and $\pKhat(\hat{\sigma}_{\text{MED}})$ (see Figure~\ref{fig:mnist_data}). The naive $p$-values are extremely small for all pairs of clusters under consideration, despite the resemblance of the centroids. By contrast, our approach yields modest $p$-values, congruent with the visual resemblance of the centroids. In addition, for the most part, the pairs for which $\pKhat(\hat{\sigma}_{\text{MED}})$ is small are visually quite different (e.g., clusters 1 and 2, clusters 1 and 5, and clusters 4 and 5). 

To demonstrate the power of the test based on $\pKhat(\hat{\sigma}_{\text{MED}})$, we also generated a ``cluster'' dataset by sampling 500 images each from digits $0,1,3,$ and 8; thus, $n=2,000$ and $q=784$. We again whitened the data to obtain uncorrelated features. 
After applying $k$-means clustering with $K=4$, we obtain four clusters that roughly correspond to four digits: cluster 1, 94.0\% digit 1; cluster 2, 72.4\% digit 3; cluster 3, 83.6\% digit 0; cluster 4, 62.4\% digit 8 (see the bottom left panel of Figure~\ref{fig:mnist_data}). Results from testing for a difference in means for each pair of clusters using $p_{\text{Naive}}$ and $\pKhat(\hat{\sigma}_{\text{MED}})$ are in Figure~\ref{fig:mnist_data}. Both sets of $p$-values are small on this ``cluster'' dataset.

\subsection{Single-cell RNA-sequencing data~\citep{Zheng2017-uv}}
\label{subsection:zheng_real_data}

In this section, we apply our proposal to single-cell RNA-sequencing data collected by~\citet{Zheng2017-uv}. Single-cell RNA-sequencing quantifies gene expression abundance at the resolution of single cells, thereby revealing cell-to-cell heterogeneity in transcription and allowing for the identification of cell types and marker genes. In practice, biologists often cluster the cells to identify putative cell types, and then perform a differential expression analysis, i.e., they test for a difference in gene expression between two clusters~\citep{Stuart2019-on,Lahnemann2020-xf,Grun2015-lh}. Because this approach ignores the fact that the clusters were estimated from the same data used for testing, it does not control the selective Type I error. 

\citet{Zheng2017-uv} profiled 68,000 peripheral blood mononuclear cells, and classified them based on their match to the expression profiles of 11 reference transcriptomes from known cell types. We consider the classified cell types to be the ``ground truth'', and use this information to demonstrate that our proposal in Section~\ref{section:hypothesis} yields reasonable results. 

As in prior work~\citep{Gao2020-yt,Duo2018-ap}, we first excluded cells with low numbers of expressed genes or total counts, as well as cells in which a large percentage of the expressed genes are mitochondrial. We then divided the counts for each cell by the total sum of counts in that cell. Finally, we applied a $\log_2$ transformation with a pseudo-count of 1 to the expression data, and considered only the subset of 500 genes with the largest average expression levels pre-normalization. 

We applied the aforementioned pre-processing pipeline separately to memory T cells ($N=10,224$) and a mixture of five types of cells (memory T cells, B cells, naive T cells, natural killer cells, and monocytes; $N=43,259$).

To investigate the selective Type I error in the absence of true clusters, we first constructed a ``no cluster'' dataset by randomly sampling 1,000  out of 10,224 memory T cells after pre-processing (thus, $n=1,000$ and $q=500$). Since the gene expression levels are highly correlated, we first whitened the data as described in Section~\ref{section:full_cov} by plugging in $\hat{{\Sigma}}^{-\frac{1}{2}} = {U}\qty({\Lambda}+0.01 \textbf{I}_n)^{-\frac{1}{2}}{U}^\top$~\citep{Coates2012-ag}, where ${U}{\Lambda}{U}^\top$ is the eigenvalue decomposition of the sample covariance matrix. 

We applied $k$-means clustering to the transformed data with $K=5$, and obtained five clusters consisting of 97, 223, 172, 165, and 343 cells, respectively (see Figure~\ref{fig:single_cell_visualization} left panel in Appendix~\ref{appendix:additional_real_data}). For each pair of estimated clusters, we computed the $p$-values $p_{\text{Naive}}$ and $\pKhat(\hat{\sigma}_{\text{MED}})$. The results are displayed in the top panel of Table~\ref{tbl:single_cell_Neg}. On this dataset, the naive p-values are extremely small for all pairs of estimated clusters, while our proposed $p$-values are quite large. In particular, at $\alpha=0.05$, the test based on $p_{\text{Naive}}$ concludes that all five estimated clusters correspond to distinct cell types (even after multiplicity correction), whereas our approach does not reject the null hypothesis that the expression levels (and thus cell types) are, in fact, the same between estimated clusters. Because this ``no cluster'' dataset consists only of memory T cells, we believe that conclusion based on $\pKhat(\hat{\sigma}_{\text{MED}})$ aligns better with the underlying biology. 

\begin{table}[htbp!]
\def~{\hphantom{0}}
\tbl{P-values $p_{\text{Naive}}$ in \eqref{eq:wald_pval} and $\pKhat$ in \eqref{eq:p_val_chen_estimate_sigma} with $\hat{\sigma}_{\text{MED}}$ defined in \eqref{eq:sigma_hat_MED} corresponding to the null hypothesis that the means of two estimated clusters are equal, for each pair of estimated clusters in the ``no cluster'' (top) and the ``cluster'' datasets (bottom)}{%
\resizebox{\linewidth}{!}{%
\begin{tabular}{ccccccccccc}
  $H_0$  & $\bar{{\mu}}_1 = \bar{{\mu}}_2$ & 
  $\bar{{\mu}}_1 = \bar{{\mu}}_3$ & $\bar{{\mu}}_1 = \bar{{\mu}}_4$ & $\bar{{\mu}}_1 = \bar{{\mu}}_5$ &  $\bar{{\mu}}_2 = \bar{{\mu}}_3$ & $\bar{{\mu}}_2 = \bar{{\mu}}_4$ &
  $\bar{{\mu}}_2 = \bar{{\mu}}_5$ & $\bar{{\mu}}_3 = \bar{{\mu}}_4$ & $\bar{{\mu}}_3 = \bar{{\mu}}_5$ & $\bar{{\mu}}_4 = \bar{{\mu}}_5$ \\
  $p_{\text{Naive}}$ & $<10^{-10}$ & $<10^{-10}$ & $<10^{-10}$ & $<10^{-10}$ & $<10^{-10}$ & $<10^{-10}$ & $<10^{-10}$ &$<10^{-10}$ & $<10^{-10}$ & $<10^{-10}$  \\  
  $\pKhat(\hat{\sigma}_{\text{MED}})$ & 0.30 & $0.31$ & 0.43 & 0.12 & 0.12 &$0.002$ &$0.10$ &0.005 &0.04  &0.05\\
\end{tabular}}}
\vspace{3mm}
{\resizebox{\linewidth}{!}{%
\begin{tabular}{ccccccccccc}
  $H_0$  & $\bar{{\mu}}_1 = \bar{{\mu}}_2$ & 
  $\bar{{\mu}}_1 = \bar{{\mu}}_3$ & $\bar{{\mu}}_1 = \bar{{\mu}}_4$ & $\bar{{\mu}}_1 = \bar{{\mu}}_5$ &  $\bar{{\mu}}_2 = \bar{{\mu}}_3$ & $\bar{{\mu}}_2 = \bar{{\mu}}_4$ &
  $\bar{{\mu}}_2 = \bar{{\mu}}_5$ & $\bar{{\mu}}_3 = \bar{{\mu}}_4$ & $\bar{{\mu}}_3 = \bar{{\mu}}_5$ & $\bar{{\mu}}_4 = \bar{{\mu}}_5$ \\
  $p_{\text{Naive}}$ & $<10^{-10}$ & $<10^{-10}$ & $<10^{-10}$ & $<10^{-10}$ & $<10^{-10}$ & $<10^{-10}$ & $<10^{-10}$ &$<10^{-10}$ & $<10^{-10}$ & $<10^{-10}$  \\
  $\pKhat(\hat{\sigma}_{\text{MED}})$ &  $4.0\times 10^{-4}$ & $<10^{-10}$ & $<10^{-10}$ & $<10^{-10}$ & $<10^{-10}$ & $<10^{-10}$ & $<10^{-10}$ &$<10^{-10}$ & $5.0\times10^{-8}$ & $<10^{-10}$ \\
\end{tabular}}
}
\label{tbl:single_cell_Neg}
\end{table}

Next, we construct a ``cluster'' dataset by randomly sampling 400 each of memory T cells, B cells, naive T cells, natural killer cells, and monocytes from the $43,259$ cells; thus, $n=2,000$ and $q=500$. After whitening the data, we applied $k$-means clustering to obtain five clusters. We see that these clusters approximately correspond to the five different cell types (cluster 1: 82.5\% naive T cells; cluster 2: 95.3\% memory T cells; cluster 3: 99.2\% B cells; cluster 4: 91.5\% nature killer cells; cluster 5: 83.3\% monocytes); estimated clusters are visualized in the right panel of Figure~\ref{fig:single_cell_visualization} in Appendix~\ref{appendix:additional_real_data}. We evaluate the $p$-values $p_{\text{Naive}}$ and $\pKhat(\hat{\sigma}_{\text{MED}})$ for all pairs of estimated clusters, and display results in the bottom panel of Table~\ref{tbl:single_cell_Neg}. Both sets of $p$-values are extremely small on this dataset, which suggests that the test based on our $p$-value has substantial power to reject the null hypothesis when it does not hold.

\section{Discussion}
\label{section:discussion}

We have proposed a test for a difference in means between two clusters estimated from $k$-means clustering, under \eqref{eq:data_gen}. Here, we outline several future research directions.

While the $p$-value in \eqref{eq:p_val_k_means_chen} leads to   selective Type I error control, it conditions on more information than is used to construct the hypothesis in \eqref{eq:null_clustering_mean}. In practice, data analysts likely only make use of the final cluster assignments (leading to the $p$-value in \eqref{eq:p_val_k_means_ideal}), as opposed to all the intermediate assignments (leading to the $p$-value in \eqref{eq:p_val_k_means_chen}). Empirically, conditioning on too much information results in a loss of power~\citep{Fithian2014-ow,jewell2019testing,Liu2018-zx}. In future work, we will investigate the possibility of leveraging recent developments in selective inference~\citep{Chen2021-hk,Le_Duy2021-iy,jewell2019testing} to compute the ``ideal'' $p$-value \eqref{eq:p_val_k_means_ideal}.

We could also consider extending our proposal to other data generating models. The normality assumption in \eqref{eq:data_gen} is critical to the proof of Proposition~\ref{prop:single_param_p}, because it guarantees that under $H_0$ in \eqref{eq:null_clustering_mean}, $\Vert {X}^\top \nu \Vert_2$, $\text{dir}({X}^\top \nu)$, and ${\Pi}_{\nu}^\perp {X}$ are pairwise independent. However, this normality assumption is often violated in practice; for instance, in single-cell genomics, the data are count-valued and the variance of gene expression levels varies drastically with the mean expression levels of that gene~\citep{Stuart2019-on,Eling2018-qy}. This has motivated some authors to work with alternative models for gene expression including Poisson~\citep{Witten2011-hv}, negative binomial~\citep{Risso2018-ew}, and curved normal~\citep{Lin2021-bw}. To extend our framework to other exponential family distributions, we may be able to leverage recent proposals to decompose ${X}$ into $f({X})$ and $g({X})$ such that both $f({X})$ and $g({X})|f({X})$ have a known, computationally-tractable distribution~\citep{Rasines2021-jc,Leiner2021-es}.

\section{Acknowledgments}
This work was partially supported by National Institutes of Health grants and a Simons Investigator Award to Daniela Witten.


\bibliographystyle{biometrika}
\bibliography{./kmeans_ref_biometrika.bib}
 
\clearpage

\begin{center}
{\LARGE \textbf{Supplementary material for\\
 ``Selective inference for $k$-means clustering''}}
\end{center}

\pagenumbering{arabic}



\appendix
\section{Appendix}
\renewcommand{\theequation}{A.\arabic{equation}}

\subsection{Proof of Proposition~\ref{prop:single_param_p}}
\label{appendix:proof_p_val_prop}

The proof of Proposition~\ref{prop:single_param_p} is similar to the proof of Theorem 1 in \citet{Gao2020-yt}, the proof of Theorem 3.1 in \citet{Loftus2014-eq}, the proof of Lemma 1 in \citet{Yang2016-km}, and the proof of Theorem 3.1 in \citet{Chen2020-rh}. 

For any non-zero $\nu\in\mathbb{R}^n$ and ${X} \in\mathbb{R}^{n\times q}$, we have that 
\begin{align} 
\label{eq:perp_nu_identity}
{X} = {\Pi}_\nu^\perp {X} + (\textbf{I}_n - {\Pi}_\nu^\perp) {X} = {\Pi}_\nu^\perp {X} + \frac{\nu\nu^\top{X}}{\Vert \nu \Vert_2^2} = {\Pi}_\nu^\perp {X} + \qty(\frac{\Vert {X}^\top \nu \Vert_2}{\Vert \nu \Vert_2^2}) \nu \qty{\text{dir}\qty({X}^\top \nu)}^\top.
\end{align} 

\begin{lemma}
\label{lemma:mutual_indep}
Under \eqref{eq:data_gen} and $H_0:{\mu}^\top \nu = 0_q$, we have that $\Vert {X}^{\top}\nu \Vert_2$, ${\Pi}_\nu^\perp {X}$, and $\emph{\text{dir}}({X}^{\top}\nu)$ are pairwise independent.
\end{lemma}
\begin{proof}
We first prove that ${X}^{\top}\nu $ is independent of ${\Pi}_\nu^\perp {X}$. The definition of ${\Pi}_\nu^\perp$ implies that ${\Pi}_\nu^\perp \nu = 0_n$, and it follows from the properties of the matrix normal distribution that ${\Pi}_\nu^\perp {X} $ and ${X}^\top \nu$ are independent. Therefore, $\Vert {X}^{\top}\nu \Vert_2$ and  $\text{dir}({X}^{\top}\nu)$ are independent of ${\Pi}_\nu^\perp {X} $ as well, since both are  functions of ${X}^\top \nu$. 

Next, we will show that $\Vert {X}^{\top}\nu \Vert_2$ and $\text{dir}({X}^{\top}\nu)$ are independent. Under \eqref{eq:data_gen} and $H_0:{\mu}^\top \nu = 0_q$, we have that ${X}^{\top}\nu \sim \mathcal{N}(0_q, \sigma^2\Vert\nu\Vert_2^2\textbf{I}_q)$. It follows that ${X}^{\top}\nu$ is rotationally invariant, and therefore $\Vert {X}^{\top}\nu \Vert_2$ is independent of $\text{dir}({X}^{\top}\nu)$ (see, e.g., Proposition 4.1 and Corollary 4.3 of \citet{Bilodeau1999-tl}).
\end{proof}

We now proceed to prove the statement in \eqref{eq:p_val_single_param}. Recalling the definition of $\pK$ in \eqref{eq:p_val_k_means_chen}, under $H_0:{\mu}^\top \nu = 0_q$ with $\nu$ defined in \eqref{eq:nu_def}, we have that
{{\small\begin{equation}
\label{eq:single_param_block}
\begin{aligned} 
 \hspace{-12mm}\pK &= \pr_{H_0}\left[ \Vert {X}^{\top}\nu \Vert_2 \geq  \Vert {x}^{\top}\nu  \Vert_2 \;\middle\vert\;  \bigcap_{t=0}^{T} \bigcap_{i=1}^{n}\left\{c_i^{(t)}\left({X}\right) = c_i^{(t)}\left({x}\right)\right\},\,{\Pi}_{\nu}^\perp {X} = {\Pi}_{\nu}^\perp {x},\, \text{dir}({X}^{\top}\nu) =  \text{dir}({x}^{\top}\nu)\right] \\
  \hspace{-8mm}&\overset{a.}{=}\pr_{H_0}\Bigg[ \Vert {X}^{\top}\nu \Vert_2 \geq  \Vert {x}^{\top}\nu  \Vert_2 \;\Bigg\vert\;  \bigcap_{t=0}^{T} \bigcap_{i=1}^{n}\left\{c_i^{(t)}\left({\Pi}_\nu^\perp {X} + \qty(\frac{\Vert {X}^\top \nu \Vert_2}{\Vert \nu \Vert_2^2}) \nu \qty{\text{dir}\qty({X}^\top \nu)}^\top\right) = c_i^{(t)}\left({x}\right)\right\},\,\\
  \hspace{-8mm}&\quad{\Pi}_{\nu}^\perp {X} = {\Pi}_{\nu}^\perp {x},\, \text{dir}({X}^{\top}\nu) =  \text{dir}({x}^{\top}\nu)\Bigg]  \\
 \hspace{-8mm} &\overset{b.}{=}\pr_{H_0}\Bigg[ \Vert {X}^{\top}\nu \Vert_2 \geq  \Vert {x}^{\top}\nu  \Vert_2 \;\Bigg\vert\;  \bigcap_{t=0}^{T} \bigcap_{i=1}^{n}\left\{c_i^{(t)}\left({\Pi}_{\nu}^\perp {x} + \qty(\frac{\Vert {X}^\top \nu \Vert_2}{\Vert \nu \Vert_2^2}) \nu \qty{\text{dir}\qty({x}^\top \nu)}^\top\right) = c_i^{(t)}\left({x}\right)\right\},\,\\
  \hspace{-8mm}&\quad{\Pi}_{\nu}^\perp {X} = {\Pi}_{\nu}^\perp {x},\, \text{dir}({X}^{\top}\nu) =  \text{dir}({x}^{\top}\nu)\Bigg]  \\
   \hspace{-8mm}&\overset{c.}{=}\pr_{H_0}\Bigg[ \Vert {X}^{\top}\nu \Vert_2 \geq  \Vert {x}^{\top}\nu  \Vert_2 \;\Bigg\vert\;  \bigcap_{t=0}^{T} \bigcap_{i=1}^{n}\left\{c_i^{(t)}\left({\Pi}_{\nu}^\perp {x} + \qty(\frac{\Vert {X}^\top \nu \Vert_2}{\Vert \nu \Vert_2^2}) \nu \qty{\text{dir}\qty({x}^\top \nu)}^\top\right) = c_i^{(t)}\left({x}\right)\right\}\Bigg] \\
   \hspace{-8mm}&\overset{d.}{=}\pr_{H_0}\Bigg[ \Vert {X}^{\top}\nu \Vert_2 \geq  \Vert {x}^{\top}\nu  \Vert_2 \;\Bigg\vert\;  \bigcap_{t=0}^{T} \bigcap_{i=1}^{n}\left\{c_i^{(t)}\left({x}'\qty(\Vert {X}^{\top}\nu \Vert_2 )\right) = c_i^{(t)}\left({x}\right)\right\}\Bigg].
\end{aligned} 
\end{equation}
}}
Here, step $a.$ follows from substituting ${X}$ with the expression in \eqref{eq:perp_nu_identity}, and step $b.$ follows from replacing ${\Pi}_\nu^\perp {X}$ and $\text{dir}({X}^{\top}\nu)$ with ${\Pi}_\nu^\perp {x}$ and $\text{dir}({x}^{\top}\nu)$, respectively. Next, in step $c.$, we used Lemma~\ref{lemma:mutual_indep}. Finally, step $d.$ follows from the definition of ${x}'(\phi)$ in \eqref{eq:xphi}.

Note that under \eqref{eq:data_gen} and $H_0:{\mu}^\top \nu = 0_q$, we have that $\Vert {X}^\top \nu \Vert_2 \sim \sigma \Vert \nu \Vert_2 \chi_q$, which concludes the proof of \eqref{eq:p_val_single_param}.

It remains to show that the test that rejects $H_0: {\mu}^\top \nu = 0$ when $\pK\leq \alpha$ controls the selective Type I error at level $\alpha$, in the sense of \eqref{eq:selective_type_1}. First of all, recall that we decided to test the  null hypothesis in \eqref{eq:null_clustering_mean} based on the output of Algorithm~\ref{algo:k_means_alt_min}. Therefore, $\pK$ controls the selective Type I error at level $\alpha$ if, for any $c_i^{(T)}({x}),\, i=1,\ldots,n$,
\begin{align}
\label{eq:appendix_selective_type_1}
\pr_{H_0}\qty[\text{reject $H_0$ at level $\alpha$} \;\middle\vert\; \bigcap_{i=1}^{n}\left\{c_i^{(T)}\left({X}\right) = c_i^{(T)}\left({x}\right)\right\}  ]\leq \alpha, \;\forall \alpha \in (0,1).
\end{align}
To prove \eqref{eq:appendix_selective_type_1}, we first note that the following holds for any $\alpha\in(0,1)$:
{\small
\begin{align}
\label{eq:proof_prob_int_trans}
\begin{split}
\hspace{-8mm}\pr_{H_0}&\left[ \pK(\Vert {X}^\top \nu\Vert_2)\leq \alpha \;\middle\vert\; \bigcap_{t=0}^{T} \bigcap_{i=1}^{n}\left\{c_i^{(t)}\left({X}\right) = c_i^{(t)}\left({x}\right)\right\},\,{\Pi}_{\nu}^\perp {X} = {\Pi}_{\nu}^\perp {x},\, \text{dir}({X}^{\top}\nu) =  \text{dir}({x}^{\top}\nu)\right]\\
  \hspace{-8mm}&\overset{a.}{=}\pr_{H_0}\Bigg[ \pK(\Vert {X}^\top \nu\Vert_2)\leq \alpha  \;\Bigg\vert\;  \bigcap_{t=0}^{T} \bigcap_{i=1}^{n}\left\{c_i^{(t)}\left({\Pi}_\nu^\perp {X} + \qty(\frac{\Vert {X}^\top \nu \Vert_2}{\Vert \nu \Vert_2^2}) \nu \qty{\text{dir}\qty({X}^\top \nu)}^\top\right) = c_i^{(t)}\left({x}\right)\right\},\,\\
  \hspace{-8mm}&\quad{\Pi}_{\nu}^\perp {X} = {\Pi}_{\nu}^\perp {x},\, \text{dir}({X}^{\top}\nu) =  \text{dir}({x}^{\top}\nu)\Bigg]  \\
 \hspace{-8mm} &\overset{b.}{=}\pr_{H_0}\Bigg[ \pK(\Vert {X}^\top \nu\Vert_2)\leq \alpha \;\Bigg\vert\;  \bigcap_{t=0}^{T} \bigcap_{i=1}^{n}\left\{c_i^{(t)}\left({\Pi}_{\nu}^\perp {x} + \qty(\frac{\Vert {X}^\top \nu \Vert_2}{\Vert \nu \Vert_2^2}) \nu \qty{\text{dir}\qty({x}^\top \nu)}^\top\right) = c_i^{(t)}\left({x}\right)\right\},\,\\
  \hspace{-8mm}&\quad{\Pi}_{\nu}^\perp {X} = {\Pi}_{\nu}^\perp {x},\, \text{dir}({X}^{\top}\nu) =  \text{dir}({x}^{\top}\nu)\Bigg]  \\
   \hspace{-8mm}&\overset{c.}{=}\pr_{H_0}\Bigg[\pK(\Vert {X}^\top \nu\Vert_2)\leq \alpha \;\Bigg\vert\;  \bigcap_{t=0}^{T} \bigcap_{i=1}^{n}\left\{c_i^{(t)}\left({\Pi}_{\nu}^\perp {x} + \qty(\frac{\Vert {X}^\top \nu \Vert_2}{\Vert \nu \Vert_2^2}) \nu \qty{\text{dir}\qty({x}^\top \nu)}^\top\right) = c_i^{(t)}\left({x}\right)\right\}\Bigg] \\
   \hspace{-8mm}&\overset{d.}{=}\pr_{H_0}\Bigg[ \pK(\Vert {X}^\top \nu\Vert_2)\leq \alpha  \;\Bigg\vert\;  \bigcap_{t=0}^{T} \bigcap_{i=1}^{n}\left\{c_i^{(t)}\left({x}'\qty(\Vert {X}^{\top}\nu \Vert_2 )\right) = c_i^{(t)}\left({x}\right)\right\}\Bigg] \\
   \hspace{-8mm}&\overset{e.}{=}\pr_{H_0}\qty[ 1-F_{q}^{\mathcal{S}_T}(\Vert {X}^\top \nu\Vert_2) \leq \alpha  \;\middle\vert\;  \bigcap_{t=0}^{T} \bigcap_{i=1}^{n}\left\{c_i^{(t)}\left({x}'\qty(\Vert {X}^\top \nu\Vert_2)\right) = c_i^{(t)}\left({x}\right)\right\} ] \\
   \hspace{-8mm}&\overset{f.}{=} \alpha. 
  \end{split}
\end{align}}
Here, steps $a.$ through $d.$ follow from the same line of argument in \eqref{eq:single_param_block}. Moreover, \eqref{eq:p_val_single_param} implies that, for a given sequence of cluster assignments $c_i^{(T)}({x}),\, i=1,\ldots,n$, $\pK$ is the survival function of a $\chi_q$ random variable, truncated to the set $\mathcal{S}_T$ defined in \eqref{eq:s_set}. Letting $F_{q}^{\mathcal{S}_T}(\cdot)$ denote the cumulative distribution function of this truncated $\chi_q$ random variable, we arrive at step $e.$ Finally, to prove $f.$, we first note that under $H_0$, the conditional cumulative distribution function of $\Vert {X}^\top \nu\Vert_2$ given $\bigcap_{t=0}^{T} \bigcap_{i=1}^{n}\left\{c_i^{(t)}\left({x}'\qty(\Vert {X}^\top \nu\Vert_2)\right) = c_i^{(t)}\left({x}\right)\right\}$ is exactly $F_{q}^{\mathcal{S}_T}$. The equality, therefore, follows from the probability integral transform, which states that for a continuous random variable $Z$, $F_{Z}(Z)$ follows the Uniform(0,1) distribution.

Finally, we have that
{
\begin{align*}
\pr_{H_0}&\qty[\pK(\Vert {X}^\top \nu\Vert_2)\leq \alpha \;\middle\vert\; \bigcap_{i=1}^{n}\left\{c_i^{(T)}\left({X}\right) = c_i^{(T)}\left({x}\right)\right\}  ] \\
&= \E_{H_0}\qty[1\qty{\pK(\Vert {X}^\top \nu\Vert_2)\leq \alpha} \;\middle\vert\; \bigcap_{i=1}^{n}\left\{c_i^{(T)}\left({X}\right) = c_i^{(T)}\left({x}\right)\right\} ]\\
&\overset{a.}{=} \E_{H_0}\Bigg( \E_{H_0}\qty[1\qty{\pK(\Vert {X}^\top \nu\Vert_2)\leq \alpha }\;\Bigg\vert\; \bigcap_{t=0}^{T} \bigcap_{i=1}^{n}\left\{c_i^{(t)}\left({X}\right) = c_i^{(t)}\left({x}\right)\right\},\,{\Pi}_{\nu}^\perp {X} = {\Pi}_{\nu}^\perp {x},\, \text{dir}({X}^{\top}\nu) =  \text{dir}({x}^{\top}\nu)  ]\\
&\quad\;\Bigg\vert\; \bigcap_{i=1}^{n}\left\{c_i^{(T)}\left({X}\right) = c_i^{(T)}\left({x}\right)\right\} \Bigg) \\
&\overset{b.}{=} \E_{H_0}\qty[\alpha\;\middle\vert\; \bigcap_{i=1}^{n}\left\{c_i^{(T)}\left({X}\right) = c_i^{(T)}\left({x}\right)\right\} ] \\
&=\alpha.
\end{align*}
}
In the proof above, $a.$ follows from the tower property of conditional expectation, and $b.$ is a direct consequence of \eqref{eq:proof_prob_int_trans}. 

Therefore, we conclude that the test based on $\pK$ controls the selective Type I error in \eqref{eq:selective_type_1}, which completes the proof of Proposition~\ref{prop:single_param_p}.


\subsection{Proof of Proposition~\ref{prop:phi_ineq}}
\label{appendix:proof_S_T_induction}

We will derive the expression for $\mathcal{S}_T$ in Proposition~\ref{prop:phi_ineq} using an induction argument.

The following two claims (Lemmas~\ref{lemma:base_t_0} and \ref{lemma:base_t_1}) serve as the ``base cases" for the proof.
\begin{lemma}
\label{lemma:base_t_0}
Recall that  $c_i^{(t)}({x})$ denotes the cluster to which the $i$th observation is assigned during the $t$th iteration of step 3b. of Algorithm~\ref{algo:k_means_alt_min} applied to data ${x}$, and that $m^{(0)}_{k}({x})$ denotes the $k$th centroid sampled from ${x}$ during step 1 of Algorithm~\ref{algo:k_means_alt_min}. For $\mathcal{S}_0$ defined as \begin{equation}
\label{eq:s_0_c}
\mathcal{S}_0 = \left\{ \phi\in\mathbb{R}: \bigcap_{i=1}^{n} \left\{c_i^{(0)}\left({x}'(\phi)\right) = c_i^{(0)}\left({x}\right)\right\}  \right\},
\end{equation}
we have that
\begin{equation}
\label{eq:lemma_base}
\mathcal{S}_0 = \bigcap_{i=1}^{n}\bigcap_{k=1}^{K} \qty{\phi:  \left\Vert \qty[{x}'(\phi)]_{i} - m^{(0)}_{c_i^{(0)}({x})}({x}'(\phi)) \right\Vert_2^2 \leq  \left\Vert \qty[{x}'(\phi)]_{i} - m^{(0)}_{k}({x}'(\phi)) \right\Vert_2^2}.
\end{equation} 
\end{lemma}
\begin{proof}

We first prove that the set in \eqref{eq:s_0_c} is a subset of the set in \eqref{eq:lemma_base}. For an arbitrary $\phi_0\in\eqref{eq:s_0_c}$ and $1 \leq i\leq n$, we have that
\begin{align*}
 &c_i^{(0)}({x}'(\phi_0)) = \underset{1\leq k \leq K}{\text{argmin}}  \left\Vert \qty[{x}'(\phi_0)]_{i} - m_{k}^{(0)}({x}'(\phi_0)) \right\Vert_2^2 \\
&\overset{a.}{\implies} \left\Vert \qty[{x}'(\phi_0)]_{i} - m_{c_i^{(0)}({x}'(\phi_0))}^{(0)}({x}'(\phi_0)) \right\Vert_2^2 \leq  \left\Vert \qty[{x}'(\phi_0)]_{i} - m_{k}^{(0)}({x}'(\phi_0)) \right\Vert_2^2,  k=1,\ldots,K \\
&\overset{b.}{\implies} \left\Vert \qty[{x}'(\phi_0)]_{i} - m_{c_i^{(0)}({x})}^{(0)}({x}'(\phi_0)) \right\Vert_2^2 \leq  \left\Vert \qty[{x}'(\phi_0)]_{i} - m_{k}^{(0)}({x}'(\phi_0)) \right\Vert_2^2, k=1,\ldots,K.
\end{align*}
Here, the first line follows from the definition of $c_i^{(0)}$ in step 2 of Algorithm~\ref{algo:k_means_alt_min}, and step $a.$ follows from the definition of the argmin function. Step $b.$ follows from the assumption that $\phi_0 \in\eqref{eq:s_0_c}$ satisfies $ c_i^{(0)}\left({x}'(\phi_0)\right) = c_i^{(0)}\left({x}\right)$. 
Because this holds for an arbitrary $1\leq i\leq n$, we have proven that $\phi_0 \in \eqref{eq:s_0_c} \implies \phi_0 \in \eqref{eq:lemma_base}$; or equivalently, $\eqref{eq:s_0_c}\subseteq \eqref{eq:lemma_base}$.

We proceed to prove the other direction. For an arbitrary $\phi_0\in\eqref{eq:lemma_base}$ and an arbitrary $1\leq i\leq n$, we have that
\begin{align*}
 &\left\Vert \qty[{x}'(\phi_0)]_{i} - m^{(0)}_{c_i^{(0)}({x})}({x}'(\phi_0)) \right\Vert_2^2 \leq  \left\Vert \qty[{x}'(\phi_0)]_{i} - m^{(0)}_{k}({x}'(\phi_0)) \right\Vert_2^2, k=1,\ldots,K \\
 &\overset{a.}{\implies} c_i^{(0)}({x}) = \underset{1\leq k \leq K}{\text{argmin}}  \left\Vert \qty[{x}'(\phi_0)]_{i} - m^{(0)}_{k}({x}'(\phi_0)) \right\Vert_2^2\\
&\overset{b.}{\implies} c_i^{(0)}({x}) = c_i^{(0)}({x}'(\phi_0)).
\end{align*}
Here, step $a.$ follows from the definition of argmin, and step $b.$ follows from combining the definition of $c_i^{(0)}({x}'(\phi))$ in step 2 of Algorithm~\ref{algo:k_means_alt_min}. We conclude that $\phi_0 \in \eqref{eq:lemma_base} \implies \phi_0 \in \eqref{eq:s_0_c}$.

Combining these two directions, we have proven that $\eqref{eq:lemma_base} =  \eqref{eq:s_0_c}$.
\end{proof}

\begin{lemma}
\label{lemma:base_t_1}
Recall that $c_i^{(t)}({x})$ denotes the cluster to which the $i$th observation is assigned in the $t$th iteration of step 3b. of Algorithm~\ref{algo:k_means_alt_min} applied to data ${x}$, and that $m^{(0)}_{k}({x})$ denotes the $k$th centroid sampled from ${x}$ during step 1 of Algorithm~\ref{algo:k_means_alt_min}. For $\mathcal{S}_1$ defined as \begin{equation}
\label{eq:s_1_c}
\mathcal{S}_1 = \left\{ \phi\in\mathbb{R}: \bigcap_{t=0}^1\bigcap_{i=1}^{n} \left\{c_i^{(t)}\left({x}'(\phi)\right) = c_i^{(t)}\left({x}\right)\right\}  \right\},
\end{equation}
and $w_{i}^{(t)}(k)$ defined in \eqref{eq:w_i_t_k},
we have that
{ \small
\begin{equation}
\begin{aligned}
\label{eq:lemma_base_t_1}
\mathcal{S}_1 &=  \bigcap_{i=1}^{n}\bigcap_{k=1}^{K} \qty{\phi:  \left\Vert \qty[{x}'(\phi)]_{i} - m^{(0)}_{c_i^{(0)}({x})}({x}'(\phi)) \right\Vert_2^2 \leq  \left\Vert \qty[{x}'(\phi)]_{i} - m^{(0)}_{k}({x}'(\phi)) \right\Vert_2^2} \cap\\ 
&\hspace{-3mm} \qty(\bigcap_{i=1}^{n}\bigcap_{k=1}^{K}\qty{ \phi: \left\Vert \qty[{x}'(\phi)]_i -  \sum_{i'=1}^n w_{i'}^{(0)}\qty(c_i^{(1)}({x}))\qty[{x}'(\phi)]_{i'} \right\Vert_2^2 \leq \left\Vert \qty[{x}'(\phi)]_i - \sum_{i'=1}^n w_{i'}^{(0)}\qty(k) \qty[{x}'(\phi)]_{i'}\right\Vert_2^2  }).
\end{aligned}
\end{equation}}
\end{lemma}
\begin{proof}

We first prove that \eqref{eq:s_1_c} $\subseteq$ \eqref{eq:lemma_base_t_1}. For an arbitrary $\phi_0 \in \eqref{eq:s_1_c}$ and an arbitrary  $1\leq i\leq n$, 
 we have that
 {\footnotesize\begin{align*}
 \hspace{-5mm}&c_i^{(1)}({x}'(\phi_0)) = \underset{1\leq k \leq K}{\text{argmin}}  \left\Vert \qty[{x}'(\phi_0)]_{i} - m^{(1)}_{k}({x}'(\phi_0)) \right\Vert_2^2 \\
 \hspace{-5mm}&\overset{a.}{\implies}  c_i^{(1)}({x}) = \underset{1\leq k \leq K}{\text{argmin}} \left\Vert \qty[{x}'(\phi_0)]_{i} - m^{(1)}_{k}({x}'(\phi_0)) \right\Vert_2^2 \\
  \hspace{-5mm}&\overset{b.}{\implies}  c_i^{(1)}({x}) = \underset{1\leq k \leq K}{\text{argmin}}  \left\Vert \qty[{x}'(\phi_0)]_i - \frac{ \sum_{i'=1}^n 1\qty{c_{i'}^{(0)}({x}'(\phi_0))= k}\qty[{x}'(\phi_0)]_{i'} }{\sum_{i'=1}^n 1\qty{c_{i'}^{(0)}({x}'(\phi_0)) = k}} \right\Vert_2^2 \\
  \hspace{-5mm}&\overset{c.}{\implies}\left\Vert \qty[{x}'(\phi_0)]_i - \frac{\sum_{i'=1}^n 1\left\{c_{i'}^{(0)}({x}'(\phi_0)) = c_i^{(1)}({x})\right\}\qty[{x}'(\phi_0)]_{i'}}{\sum_{i'=1}^n 1\qty{c_{i'}^{(0)}({x}'(\phi_0))= c_{i}^{(1)}({x})}}   \right\Vert_2^2 
 \leq \left\Vert \qty[{x}'(\phi_0)]_i - \frac{ \sum_{i'=1}^n 1\qty{c_{i'}^{(0)}({x}'(\phi_0))= k}\qty[{x}'(\phi_0)]_{i'} }{\sum_{i'=1}^n 1\qty{c_{i'}^{(0)}({x}'(\phi_0)) = k}} \right\Vert_2^2 , k=1,\ldots,K\\
 \hspace{-5mm}&\overset{d.}{\implies}
 \left\Vert \qty[{x}'(\phi_0)]_i - \frac{\sum_{i'=1}^n 1\left\{c_{i'}^{(0)}({x}) = c_i^{(1)}({x})\right\}\qty[{x}'(\phi_0)]_{i'}}{\sum_{i'=1}^n 1\qty{c_{i'}^{(0)}({x})= c_{i}^{(1)}({x})}}   \right\Vert_2^2 
 \leq \left\Vert \qty[{x}'(\phi_0)]_i - \frac{ \sum_{i'=1}^n 1\qty{c_{i'}^{(0)}({x})= k}\qty[{x}'(\phi_0)]_{i'} }{\sum_{i'=1}^n 1\qty{c_{i'}^{(0)}({x}) = k}} \right\Vert_2^2 , k=1,\ldots,K\\
 \hspace{-5mm}&\overset{e.}{\implies} \left\Vert \qty[{x}'(\phi_0)]_i -  \sum_{i'=1}^n w_{i'}^{(0)}\qty(c_i^{(1)}({x}))\qty[{x}'(\phi_0)]_{i'} \right\Vert_2^2 \leq \left\Vert \qty[{x}'(\phi_0)]_i - \sum_{i'=1}^n w_{i'}^{(0)}\qty(k) \qty[{x}'(\phi_0)]_{i'}\right\Vert_2^2, k=1,\ldots,K.
 \end{align*} }
 In the equations above, the first line follows from step 3b. of Algorithm~\ref{algo:k_means_alt_min} with $t=0$. Next, step $a.$ follows from the definition of \eqref{eq:s_1_c}, which implies that $c_i^{(1)}({x}'(\phi_0)) = c_i^{(1)}({x})$. Step $b.$ is a direct consequence of step 3a. of Algorithm~\ref{algo:k_means_alt_min} with $t=0$. In steps $c.$ and $d.$, we used the definitions of the argmin function and \eqref{eq:s_1_c}. Finally, we apply the definition of $w_i^{(t)}$ in \eqref{eq:w_i_t_k} to get $e.$ Because this holds for an arbitrary $1\leq i\leq n$, $\phi_0\in\eqref{eq:s_1_c}$ implies that $\phi_0$ is an element of the second set in the intersection in \eqref{eq:lemma_base_t_1}. 

 Moreover, $\phi_0\in\eqref{eq:s_1_c}$ implies that $\phi_0 \in \eqref{eq:s_0_c}$, which, according to Lemma~\ref{lemma:base_t_0}, further implies that $\phi_0$ is an element 
of the first set in the intersection in \eqref{eq:lemma_base_t_1}.
 To summarize, we have proven that $\phi_0 \in \eqref{eq:s_1_c} \implies \phi_0 \in \eqref{eq:lemma_base_t_1}$, and as a result, $\eqref{eq:s_1_c} \subseteq \eqref{eq:lemma_base_t_1}$. 

Next, we prove that the set in \eqref{eq:lemma_base_t_1} is a subset of the set in \eqref{eq:s_1_c}. For an arbitrary $\phi_0 \in \eqref{eq:lemma_base_t_1}$ and an arbitrary $1\leq i\leq n$, we have that 
{\footnotesize
\begin{align*}
\hspace{-10mm}&\phi_0 \in \eqref{eq:lemma_base_t_1} \overset{a.}{\implies}  \left\Vert \qty[{x}'(\phi_0)]_i - \frac{\sum_{i'=1}^n 1\left\{c_{i'}^{(0)}({x}) = c_i^{(1)}({x})\right\}\qty[{x}'(\phi_0)]_{i'}}{\sum_{i'=1}^n 1\qty{c_{i'}^{(0)}({x})= c_{i}^{(1)}({x})}}   \right\Vert_2^2 \leq \left\Vert \qty[{x}'(\phi_0)]_i - \frac{ \sum_{i'=1}^n 1\qty{c_{i'}^{(0)}({x})= k}\qty[{x}'(\phi_0)]_{i'} }{\sum_{i'=1}^n 1\qty{c_{i'}^{(0)}({x}) = k}} \right\Vert_2^2, k=1,\ldots,K \\
\hspace{-10mm}&\overset{b.}{\implies}   \left\Vert \qty[{x}'(\phi_0)]_i - \frac{\sum_{i'=1}^n 1\left\{c_{i'}^{(0)}({x}'(\phi_0)) = c_i^{(1)}({x})\right\}\qty[{x}'(\phi_0)]_{i'}}{\sum_{i'=1}^n 1\qty{c_{i'}^{(0)}({x}'(\phi_0))= c_{i}^{(1)}({x})}}   \right\Vert_2^2 \leq \left\Vert \qty[{x}'(\phi_0)]_i - \frac{ \sum_{i'=1}^n 1\qty{c_{i'}^{(0)}({x}'(\phi_0))= k}\qty[{x}'(\phi_0)]_{i'} }{\sum_{i'=1}^n 1\qty{c_{i'}^{(0)}({x}'(\phi_0)) = k}} \right\Vert_2^2, k=1,\ldots,K \\
\hspace{-10mm}&\overset{c.}{\implies} c_{i}^{(1)}({x}) = \underset{1\leq k \leq K}{\text{argmin}}  \left\Vert \qty[{x}'(\phi_0)]_{i} - \frac{ \sum_{i'=1}^n 1\qty{c_{i'}^{(0)}({x}'(\phi_0))= k}\qty[{x}'(\phi_0)]_{i'} }{\sum_{i'=1}^n 1\qty{c_{i'}^{(0)}({x}'(\phi_0)) = k}} \right\Vert_2^2  \\
\hspace{-10mm}&\overset{d.}{\implies}  c_{i}^{(1)}({x}) = \underset{1\leq k \leq K}{\text{argmin}}  \left\Vert \qty[{x}'(\phi_0)]_{i} - m_k^{(1)}({x}'(\phi_0))\right\Vert_2^2  \\
\hspace{-10mm}&\overset{e.}{\implies} c_{i}^{(1)}({x}) = c_{i}^{(1)}({x}'(\phi_0)).
\end{align*} }
Here, step $a.$ follows from the definition of \eqref{eq:lemma_base_t_1}. In step $b.$, we first apply Lemma~\ref{lemma:base_t_0}, which implies that $\eqref{eq:lemma_base_t_1}\subseteq \eqref{eq:lemma_base}$. Therefore, $\phi_0\in \eqref{eq:lemma_base_t_1} \implies c_i^{(0)}({x}) =  c_i^{(0)}({x}'(\phi_0))$, for all $ i=1,\ldots,n, k=1,\ldots, K$, yielding the desired equality. Next, step $c.$ follows from the definition of the argmin function. Finally, steps $d.$ and $e.$ follow directly from the definitions of $m_k^{(t)}$ and $c_i^{(t)}$ in steps 3a. and 3b. of Algorithm~\ref{algo:k_means_alt_min}, respectively. 

Because the result above holds for an arbitrary $i$, we have that $\phi_0 \in \eqref{eq:lemma_base_t_1}  \implies c_{i}^{(1)}({x}) = c_{i}^{(1)}({x}'(\phi))$, $ i =1,\ldots, n$. Combining this result with the observation that $\eqref{eq:lemma_base_t_1} \subseteq \eqref{eq:lemma_base}$, we have that $\eqref{eq:lemma_base_t_1} \subseteq \eqref{eq:s_1_c}$, which concludes the proof. 
\end{proof}

Next, we will prove the inductive step in the proof of Proposition~\ref{prop:phi_ineq}, which relies on the following claim.
\begin{lemma}
\label{lemma:inductive}
Recall that  $c_i^{(t)}({x})$ denotes the cluster to which the $i$th observation is assigned in the $t$th iteration of Algorithm~\ref{algo:k_means_alt_min} applied to the data ${x}$, and that $m^{(0)}_{k}({x})$ denotes the $k$th centroid sampled from ${x}$ during initialization. For some $1\leq\tilde{T}\leq T-1$, define \begin{equation}
\label{eq:s_T_tilde_c}
\mathcal{S}_{\tilde{T}} = \left\{ \phi\in\mathbb{R}: \bigcap_{t=0}^{\tilde{T}}\bigcap_{i=1}^{n} \left\{c_i^{(t)}\left({x}'(\phi)\right) = c_i^{(t)}\left({x}\right)\right\}  \right\}.
\end{equation}
Suppose that the following holds for $\tilde{T}$:
{\footnotesize
\begin{equation}
\begin{aligned}
\label{eq:inductive_hypo}
\hspace{-7mm}\mathcal{S}_{\tilde{T}} &= \bigcap_{i=1}^{n}\bigcap_{k=1}^{K} \qty{\phi:  \left\Vert \qty[{x}'(\phi)]_{i} - m^{(0)}_{c_i^{(0)}({x})}({x}'(\phi)) \right\Vert_2^2 \leq  \left\Vert \qty[{x}'(\phi)]_{i} - m^{(0)}_{k}({x}'(\phi)) \right\Vert_2^2}  \\
&\hspace{-10mm}\cap \qty(\bigcap_{t=1}^{\tilde{T}}\bigcap_{i=1}^{n}\bigcap_{k=1}^{K}\qty{ \phi: \left\Vert \qty[{x}'(\phi)]_i - \sum_{i'=1}^n w_{i'}^{(t-1)}\qty(c_i^{(t)}({x}))\qty[{x}'(\phi)]_{i'}   \right\Vert_2^2
 \leq \left\Vert \qty[{x}'(\phi)]_i - \sum_{i'=1}^n w_{i'}^{(t-1)}\qty(k)\qty[{x}'(\phi)]_{i'}  \right\Vert_2^2  }),
\end{aligned}
\end{equation}
 }
 where $w_{i}^{(t)}(\cdot)$ is defined in \eqref{eq:w_i_t_k}. 
Then, for $\mathcal{S}_{\tilde{T}+1}$ defined as \begin{equation}
\label{eq:s_T_tilde_plus_1_c}
\mathcal{S}_{\tilde{T}+1} = \left\{ \phi\in\mathbb{R}: \bigcap_{t=0}^{\tilde{T}+1}\bigcap_{i=1}^{n} \left\{c_i^{(t)}\left({x}'(\phi)\right) = c_i^{(t)}\left({x}\right)\right\}  \right\},
\end{equation}
 we have that
{\footnotesize
\begin{equation}
\begin{aligned}
\label{eq:tilde_1_form}
\mathcal{S}_{\tilde{T}+1} &=  \bigcap_{i=1}^{n}\bigcap_{k=1}^{K} \qty{\phi:  \left\Vert \qty[{x}'(\phi)]_{i} - m^{(0)}_{c_i^{(0)}({x})}({x}'(\phi)) \right\Vert_2^2 \leq  \left\Vert \qty[{x}'(\phi)]_{i} - m^{(0)}_{k}({x}'(\phi)) \right\Vert_2^2}   \\
&\cap \qty(\bigcap_{t=1}^{\tilde{T}+1}\bigcap_{i=1}^{n}\bigcap_{k=1}^{K}\qty{ \phi: \left\Vert \qty[{x}'(\phi)]_i - \sum_{i'=1}^n w_{i'}^{(t-1)}\qty(c_i^{(t)}({x}))\qty[{x}'(\phi)]_{i'}   \right\Vert_2^2
 \leq \left\Vert \qty[{x}'(\phi)]_i - \sum_{i'=1}^n w_{i'}^{(t-1)}\qty(k)\qty[{x}'(\phi)]_{i'}  \right\Vert_2^2  }).
\end{aligned}
\end{equation}
}
\end{lemma}

\begin{proof}
Using the definitions in \eqref{eq:s_T_tilde_c} and \eqref{eq:s_T_tilde_plus_1_c}, we have that
\begin{align}
\label{eq:defn_form_equiv}
 \mathcal{S}_{\tilde{T}+1} = \mathcal{S}_{\tilde{T}}\cap\qty( \bigcap_{i=1}^n  \qty{\phi\in \mathbb{R}: c_i^{(\tilde{T}+1)}\left({x}'(\phi)\right) = c_i^{(\tilde{T}+1)}\left({x}\right) }).
\end{align}

Therefore, it suffices to prove that $\eqref{eq:defn_form_equiv}=\eqref{eq:tilde_1_form}$, under the inductive hypothesis \eqref{eq:inductive_hypo}.

We start by proving that $\eqref{eq:defn_form_equiv}\subseteq\eqref{eq:tilde_1_form}$. For an arbitrary $\phi_0 \in \eqref{eq:defn_form_equiv}$ and an arbitrary $1 \leq i \leq n$, we have that 
{\footnotesize
\begin{align*}
 \hspace{-5mm}&c_i^{(\tilde{T}+1)}\left({x}'(\phi_0)\right) = c_i^{(\tilde{T}+1)}\left({x}\right) \overset{a.}{\implies} c_i^{(\tilde{T}+1)}({x})= \underset{1\leq k \leq K}{\text{argmin}}  \left\Vert \qty[{x}'(\phi_0)]_{i} - m_{k}^{(\tilde{T}+1)}({x}'(\phi_0)) \right\Vert_2^2 \\
 \hspace{-5mm}&\overset{b.}{\implies} \left\Vert \qty[{x}'(\phi_0)]_i - \frac{\sum_{i'=1}^n 1\left\{c_{i'}^{(\tilde{T})}({x}'(\phi_0)) = c_i^{(\tilde{T}+1)}({x})\right\}\qty[{x}'(\phi_0)]_{i'}}{\sum_{i'=1}^n 1\qty{c_{i'}^{(\tilde{T})}({x}'(\phi_0))= c_{i}^{(\tilde{T}+1)}({x})}}   \right\Vert_2^2
 \leq \left\Vert \qty[{x}'(\phi_0)]_i - \frac{ \sum_{i'=1}^n 1\qty{c_{i'}^{(\tilde{T})}({x}'(\phi_0))= k}\qty[{x}'(\phi_0)]_{i'} }{\sum_{i'=1}^n 1\qty{c_{i'}^{(\tilde{T})}({x}'(\phi_0)) = k}} \right\Vert_2^2, k =1,\ldots,K \\
  \hspace{-5mm}&\overset{c.}{\implies} \left\Vert \qty[{x}'(\phi_0)]_i - \frac{\sum_{i'=1}^n 1\left\{c_{i'}^{(\tilde{T})}({x}) = c_i^{(\tilde{T}+1)}({x})\right\}\qty[{x}'(\phi_0)]_{i'}}{\sum_{i'=1}^n 1\qty{c_{i'}^{(\tilde{T})}({x})= c_{i}^{(\tilde{T}+1)}({x})}}   \right\Vert_2^2
 \leq \left\Vert \qty[{x}'(\phi_0)]_i - \frac{ \sum_{i'=1}^n 1\qty{c_{i'}^{(\tilde{T})}({x})= k}\qty[{x}'(\phi_0)]_{i'} }{\sum_{i'=1}^n 1\qty{c_{i'}^{(\tilde{T})}({x}) = k}} \right\Vert_2^2, k =1,\ldots,K   \\
 \hspace{-5mm}&\overset{d.}{\implies} \phi_0 \in \eqref{eq:tilde_1_form}.
\end{align*}
} 
Here, the first statement follows from the definition of $\mathcal{S}_{\tilde{T}+1}$. Next, steps $a.$  and $b.$ follow from the definitions of 
$c_i^{(\tilde{T}+1)}$ and $m_{k}^{(\tilde{T}+1)}({x}'(\phi_0))$ in steps 3b. and 3a. of Algorithm~\ref{algo:k_means_alt_min}, respectively. In step $c.$, we used the fact that $\phi_0 \in\eqref{eq:defn_form_equiv} \implies \phi_0 \in \mathcal{S}_{\tilde{T}} \implies c_i^{\tilde{T}}({x}'(\phi_0))=c_i^{\tilde{T}}({x})$. Finally, $d.$ follows from the definition of $w_i^{(t)}$ in \eqref{eq:w_i_t_k}.

We continue with the reverse direction. Applying the inductive hypothesis \eqref{eq:inductive_hypo}, together with the definition of $S_{\tilde{T}+1}$ in \eqref{eq:tilde_1_form} and the definition of $w_i^{(t)}$ in \eqref{eq:w_i_t_k}, we have that
{\footnotesize
\begin{equation}
\begin{aligned}
\label{eq:equiv_t_tilde_2}
\hspace{-4mm}&\eqref{eq:tilde_1_form}=S_{\tilde{T}}\cap  \vast(\bigcap_{i=1}^{n}\bigcap_{k=1}^{K} 
\vast\{ \phi: \left\Vert \qty[{x}'(\phi)]_i - \frac{\sum_{i'=1}^n 1\left\{c_{i'}^{(\tilde{T})}({x}) = c_i^{(\tilde{T}+1)}({x})\right\}\qty[{x}'(\phi)]_{i'}}{\sum_{i'=1}^n 1\qty{c_{i'}^{(\tilde{T})}({x})= c_{i}^{(\tilde{T}+1)}({x})}} \right\Vert_2^2 \leq \\ 
&\hspace{5mm}\left\Vert \qty[{x}'(\phi)]_i - \frac{ \sum_{i'=1}^n 1\qty{c_{i'}^{(\tilde{T})}({x})= k}\qty[{x}'(\phi)]_{i'} }{\sum_{i'=1}^n 1\qty{c_{i'}^{(\tilde{T})}({x}) = k}} \right\Vert_2^2  \vast\}\vast).
\end{aligned}
\end{equation}
}
For an arbitrary $\phi_0 \in \eqref{eq:tilde_1_form}$ and any $1 \leq i \leq n$, the following holds:
{\footnotesize\begin{align*}
\hspace{-10mm}&\left\Vert \qty[{x}'(\phi_0)]_i - \frac{\sum_{i'=1}^n 1\left\{c_{i'}^{(\tilde{T})}({x}) = c_i^{(\tilde{T}+1)}({x})\right\}\qty[{x}'(\phi_0)]_{i'}}{\sum_{i'=1}^n 1\qty{c_{i'}^{(\tilde{T})}({x})= c_{i}^{(\tilde{T}+1)}({x})}} \right\Vert_2^2 \leq \left\Vert \qty[{x}'(\phi_0)]_i - \frac{ \sum_{i'=1}^n 1\qty{c_{i'}^{(\tilde{T})}({x})= k}\qty[{x}'(\phi_0)]_{i'} }{\sum_{i'=1}^n 1\qty{c_{i'}^{(\tilde{T})}({x}) = k}} \right\Vert_2^2 , k=1,\ldots,K \\
\hspace{-10mm}&\overset{a.}{\implies} \left\Vert \qty[{x}'(\phi_0)]_i - \frac{\sum_{i'=1}^n 1\left\{c_{i'}^{(\tilde{T})}({x}'(\phi_0)) = c_i^{(\tilde{T}+1)}({x})\right\}\qty[{x}'(\phi_0)]_{i'}}{\sum_{i'=1}^n 1\qty{c_{i'}^{(\tilde{T})}({x}'(\phi_0))= c_{i}^{(\tilde{T}+1)}({x})}}   \right\Vert_2^2 \leq \left\Vert \qty[{x}'(\phi_0)]_i - \frac{ \sum_{i'=1}^n 1\qty{c_{i'}^{(\tilde{T})}({x}'(\phi_0))= k}\qty[{x}'(\phi_0)]_{i'} }{\sum_{i'=1}^n 1\qty{c_{i'}^{(\tilde{T})}({x}'(\phi_0)) = k}} \right\Vert_2^2, k=1,\ldots,K \\
\hspace{-10mm}&\overset{b.}{\implies} c_i^{(\tilde{T}+1)}({x})= \underset{1\leq k \leq K}{\text{argmin}}  \left\Vert \qty[{x}'(\phi_0)]_{i} - m_{k}^{(\tilde{T}+1)}({x}'(\phi_0)) \right\Vert_2^2 \\
\hspace{-10mm}&\overset{c.}{\implies} c_i^{(\tilde{T}+1)}({x})= c_i^{(\tilde{T}+1)}({x}'(\phi_0)).
\end{align*}
}

Here, to derive step $a.$, we first note that by \eqref{eq:equiv_t_tilde_2}, any element $\phi_0 $ of $\eqref{eq:tilde_1_form}$ is also an element of $\mathcal{S}_{\tilde{T}}$. Therefore, using the definition of $\mathcal{S}_{\tilde{T}}$ in \eqref{eq:s_T_tilde_c}, we have that $\bigcap_{t=1}^{\tilde{T}}\qty{c_i^{({t})}\left({x}'(\phi_0)\right) = c_i^{({t})}\left({x}\right)}$, and step $a.$ follows directly. Next, steps $b.$ and $c.$ follow directly from steps 3a. and 3b. of Algorithm~\ref{algo:k_means_alt_min} with $t=\tilde{T}$. By inspecting the form of \eqref{eq:defn_form_equiv}, we conclude that $\phi_0 \in \eqref{eq:tilde_1_form} \implies \phi_0 \in \eqref{eq:defn_form_equiv}$.

In conclusion, we have proven that $\eqref{eq:tilde_1_form}=\eqref{eq:defn_form_equiv}$, which completes the proof.
\end{proof}

The inductive proof of Proposition~\ref{prop:phi_ineq} follows  from combining Lemmas~\ref{lemma:base_t_0}, \ref{lemma:base_t_1} and \ref{lemma:inductive}.

\subsection{Proof of Lemmas~\ref{lemma:simple_norm} and \ref{lemma:canonical_norm}}
\label{appendix:quad_ineq}

We first prove Lemma~\ref{lemma:simple_norm}, which is also Lemma 2 in \citet{Gao2020-yt}.

\begin{proof}
We first express the inner product $\langle \qty[{x}'(\phi)]_i, \qty[{x}'(\phi)]_j \rangle$ as a function of $\phi$. 
From \eqref{eq:xphi}, we have that $[{x}'(\phi)]_i = {x}_i +  \nu_i  \left( \frac{\phi-\Vert {x}^{\top}\nu  \Vert_2 }{\Vert \nu  \Vert_2^2 }\right) \text{dir}({x}^{\top}\nu) = {x}_i  - \nu_i \frac{|| {x}^{\top}\nu||_2}{||\nu||_2^2} \text{dir}({x}^{\top}\nu) + \qty(\frac{\nu_i}{||\nu||_2^2}  \phi )\text{dir}({x}^{\top}\nu).$

Therefore, 
{\small
\begin{align*}
\left\langle \qty[{x}'(\phi)]_i, \qty[{x}'(\phi)]_j \right\rangle &= \left\langle {x}_i   - \nu_i \frac{|| {x}^{\top}\nu||_2}{||\nu||_2^2} \text{dir}({x}^{\top}\nu) + \qty(\frac{\nu_i}{||\nu||_2^2}  \phi )\text{dir}({x}^{\top}\nu), {x}_j  - \nu_j \frac{|| {x}^{\top}\nu||_2}{||\nu||_2^2} \text{dir}({x}^{\top}\nu) + \qty(\frac{\nu_j}{||\nu||_2^2}  \phi )\text{dir}({x}^{\top}\nu) \right\rangle \\
&= \qty(\frac{{(\nu_i\nu_j)^{1/2}}}{||\nu||_2^2}\phi )^2 + \left\langle {x}_i   - \nu_i \frac{|| {x}^{\top}\nu||_2}{||\nu||_2^2} \text{dir}({x}^{\top}\nu), \qty(\frac{\nu_j}{||\nu||_2^2})\text{dir}({x}^{\top}\nu) \right\rangle \cdot  \phi \\
&+ \left\langle {x}_j -\nu_j \frac{|| {x}^{\top}\nu||_2}{||\nu||_2^2} \text{dir}({x}^{\top}\nu), \qty(\frac{\nu_i}{||\nu||_2^2})\text{dir}({x}^{\top}\nu) \right\rangle \cdot  \phi   \\&
+  \left\langle {x}_i   - \nu_i \frac{|| {x}^{\top}\nu||_2}{||\nu||_2^2} \text{dir}({x}^{\top}\nu), {x}_j  - \nu_j \frac{|| {x}^{\top}\nu||_2}{||\nu||_2^2} \text{dir}({x}^{\top}\nu) \right\rangle \\
&= \qty(\frac{{(\nu_i\nu_j)^{1/2}}}{||\nu||_2^2} )^2 \phi^2 + \left( \frac{\nu_j}{||\nu||_2^2} \left\langle  {x}_i ,  \text{dir}({x}^{\top}\nu) \right\rangle +  \frac{\nu_i}{||\nu||_2^2} \left\langle  {x}_j,  \text{dir}({x}^{\top}\nu) \right\rangle - 2\frac{\nu_i\nu_j||{x}^{\top}\nu||_2}{||\nu||_2^4} \right) \phi \\
&+\left\langle {x}_i   - \nu_i \frac{|| {x}^{\top}\nu||_2}{||\nu||_2^2} \text{dir}({x}^{\top}\nu), {x}_j  - \nu_j \frac{|| {x}^{\top}\nu||_2}{||\nu||_2^2} \text{dir}({x}^{\top}\nu) \right\rangle.
\end{align*} 
}

Next, using the expression for $\langle \qty[{x}'(\phi)]_i, \qty[{x}'(\phi)]_j \rangle$ above, we have that
{
\begin{align*}
\left\Vert\qty[{x}'(\phi)]_i- \qty[{x}'(\phi)]_j\right\Vert_2^2 &= \left\langle \qty[{x}'(\phi)]_i- \qty[{x}'(\phi)]_j ,\qty[{x}'(\phi)]_i- \qty[{x}'(\phi)]_j\right\rangle \\
&\overset{}{=}  \Big\langle {x}_i  - {x}_j  - (\nu_i-\nu_j) \frac{|| {x}^{\top}\nu||_2}{||\nu||_2^2} \text{dir}({x}^{\top}\nu) + \qty(\frac{(\nu_i-\nu_j)}{||\nu||_2^2}  \phi )\text{dir}({x}^{\top}\nu),\\
&{x}_i  - {x}_j  - (\nu_i-\nu_j) \frac{|| {x}^{\top}\nu||_2}{||\nu||_2^2} \text{dir}({x}^{\top}\nu) + \qty(\frac{(\nu_i-\nu_j)}{||\nu||_2^2}  \phi )\text{dir}({x}^{\top}\nu) \Big\rangle \\
&\overset{}{=} \qty(\frac{\nu_i-\nu_j}{||\nu||_2^2} )^2 \phi^2 + 2\left( \frac{\nu_i-\nu_j}{||\nu||_2^2} \left\langle  {x}_i -{x}_j,  \text{dir}({x}^{\top}\nu) \right\rangle  -\qty(\frac{\nu_i-\nu_j}{||\nu||_2^2} )^2 ||{x}^{\top}\nu||_2 \right) \phi \\
&+\left\Vert {x}_i   -  {x}_j -(\nu_i-\nu_j) \frac{{x}^{\top}\nu}{||\nu||_2^2} \right\Vert_2^2.
\end{align*}}
\end{proof}
This completes the proof of Lemma~\ref{lemma:simple_norm}. 

We continue with the proof of Lemma~\ref{lemma:canonical_norm}.  Using the definition of $w_{i}^{(t-1)}(k)$ in \eqref{eq:w_i_t_k}, we have that
{\begin{align*}
\left\Vert \qty[{x}'(\phi)]_i - \frac{\sum_{i'=1}^n 1\qty{c_{i'}^{(t-1)}({x})= k}\qty[{x}'(\phi)]_{i'} }{\sum_{i'=1}^n 1\qty{c_{i'}^{(t-1)}({x}) = k}}  \right\Vert_2^2 = \left\Vert \qty[{x}'(\phi)]_i - \sum_{i'=1}^n w_{i'}^{(t-1)}(k)\qty[{x}'(\phi)]_{i'}   \right\Vert_2^2,
\end{align*}}
where
{\begin{align*}
\qty[{x}'(\phi)]_i - \sum_{i'=1}^n w_{i'}^{(t-1)}(k)\qty[{x}'(\phi)]_{i'} = \qty(\sum_{i'=1}^n w_{i'}^{(t-1)}(k) \frac{\nu_i}{||\nu||_2^2})  \phi +\sum_{i'=1}^n w_{i'}^{(t-1)}(k)\qty({x}_i  -   \nu_i \frac{|| {x}^{\top}\nu||_2}{||\nu||_2^2} \text{dir}({x}^{\top}\nu))
\end{align*}}
is a linear function of $\phi$. The rest of the proof follows directly from the same  set of calculations in the proof of Lemma~\ref{lemma:simple_norm}. 

\subsection{Proof of Proposition~\ref{prop:time_complexity}}
\label{appendix:proof_S_timing}

Recall that $n,q,K,T$ denote the number of samples (see \eqref{eq:data_gen}), the number of features (see \eqref{eq:data_gen}), the number of clusters (see Algorithm~\ref{algo:k_means_alt_min}), and the maximum number of iterations for which Algorithm~\ref{algo:k_means_alt_min} is run. 

According to Proposition~\ref{prop:phi_ineq}, to compute the set $\mathcal{S}_T$ in \eqref{eq:s_set}, it suffices to compute the intersection of the two sets in \eqref{eq:S_T_quad_ineq_part_1} and \eqref{eq:S_T_quad_ineq_part_2}. 

We first make the following observations for our timing complexity analysis:
\begin{itemize}
\item
Observation 1: according to Lemma~\ref{lemma:simple_norm}, the set in \eqref{eq:S_T_quad_ineq_part_1} is an intersection of $nK$ quadratic inequalities. 

\item
Observation 2: according to Lemma~\ref{lemma:canonical_norm}, the set in \eqref{eq:S_T_quad_ineq_part_2} is an intersection of $nKT$ quadratic inequalities. 

\item
Observation 3: we can solve a quadratic inequality in $\mathcal{O}(1)$ time using the quadratic formula.

\item
Observation 4: we can intersect the solution sets of $N$ quadratic inequalities in $\mathcal{O}(N\log N)$ time~\citep{intervals_pacakge}.

\end{itemize}

Equipped with these observations, we will analyze the timing complexity of computing the set \eqref{eq:S_T_quad_ineq_part_1}. Note that the coefficients for each of the $nK$ quadratic inequalities can be computed in $\mathcal{O}(nq)$ operations: first, using the property that ${x}^\top \nu = \sum_{i\in {\hat{\mathcal{C}}}_1}{x}_i/|\hat{\mathcal{C}}_1|-\sum_{i\in \hat{\mathcal{C}}_2}{x}_i/|\hat{\mathcal{C}}_2|$, we can compute $\Vert {x}^\top \nu \Vert_2$ and $\text{dir}({x}^\top \nu)$ in $\mathcal{O}(nq)$ operations. Then, computing the coefficients $a, b,$ and $\gamma$ in Lemma~\ref{lemma:simple_norm} takes $\mathcal{O}(1)$, $\mathcal{O}(q)$, and $\mathcal{O}(q)$ operations, respectively. For each inequality, obtaining the solution set requires $\mathcal{O}(1)$ operations (see Observation 3). Finally, intersecting the solution sets of the $n(K-1)$  quadratic inequalities incurs another  $\mathcal{O}(nK\log(nK))$ operations. Thus, the computational cost for \eqref{eq:S_T_quad_ineq_part_1} totals to $\mathcal{O}(nKq+nK\log(nK))$ operations.

Next, we analyze the cost of computing the set \eqref{eq:S_T_quad_ineq_part_2}. Note that using Observation 2, we need to solve $nKT$ quadratic inequalities. Here, for each quadratic inequality of the form in Lemma~\ref{lemma:canonical_norm}, it takes $\mathcal{O}(n)$,\,$\mathcal{O}(n+q)$, and $\mathcal{O}(n+q)$ operations to compute the coefficients $\tilde{a},\tilde{b}$, and $\tilde{\gamma}$, respectively. Therefore, obtaining the $nKT$ solution sets will take $\mathcal{O}(nKT(n+q))$ time. Finally, intersecting these sets using Observation 4 adds another $\mathcal{O}(nKT\log(nKT))$ operations.

Combining the costs for computing the set in \eqref{eq:S_T_quad_ineq_part_1} and the set in \eqref{eq:S_T_quad_ineq_part_2}, we conclude that the cost for computing the set $\mathcal{S}_T$ in \eqref{eq:s_set} is $\mathcal{O}(nKT(n+q)+nKT\log(nKT))$ operations.

\subsection{Proof of Proposition~\ref{prop:full_cov_single_param_p} and computation of $\pKSigma$}
\label{appendix:proof_p_val_prop_full_cov}

The proof of Proposition~\ref{prop:full_cov_single_param_p} is similar to that of Proposition~\ref{prop:single_param_p}. 

First note that for any non-zero $\nu\in\mathbb{R}^n$ and ${X} \in\mathbb{R}^{n\times q}$, we have that 
{
\begin{align} 
\label{eq:perp_nu_identity_full_cov}
{X} =  {\Pi}_\nu^\perp {X} + \frac{\nu\nu^\top{X}{\Sigma}^{-\frac{1}{2}} {\Sigma}^{\frac{1}{2}} }{\Vert \nu \Vert_2^2} = {\Pi}_\nu^\perp {X} + \qty(\frac{\Vert {\Sigma}^{-\frac{1}{2}}{X}^\top \nu \Vert_2}{\Vert \nu \Vert_2^2}) \nu \qty{\text{dir}\qty({\Sigma}^{-\frac{1}{2}}{X}^\top \nu)}^\top {\Sigma}^{\frac{1}{2}}.
\end{align} 
}
\begin{lemma}
\label{lemma:mutual_indep_full_cov}
Under \eqref{eq:data_full_cov_gen} and $H_0:{\mu}^\top \nu = 0_q$, $\Vert {\Sigma}^{-\frac{1}{2}} {X}^{\top}\nu \Vert_2$, ${\Pi}_\nu^\perp {X}$, and $\emph{\text{dir}}\qty({\Sigma}^{-\frac{1}{2}}{X}^{\top}\nu)$ are pairwise independent.
\end{lemma}
\begin{proof}
As in the proof of Lemma~\ref{lemma:mutual_indep}, ${\Pi}_\nu^\perp \nu = 0_n$, and it follows from the property of the matrix normal distribution that ${X}^{\top}\nu$ is independent of ${\Pi}_\nu^\perp {X} $. Because both $\Vert {\Sigma}^{-\frac{1}{2}} {X}^{\top}\nu \Vert_2$ and  $\text{dir}( {\Sigma}^{-\frac{1}{2}}{X}^{\top}\nu)$ are functions of ${X}^{\top}\nu$, both are independent of ${\Pi}_\nu^\perp {X}$.

Next, we will show that $\Vert {\Sigma}^{-\frac{1}{2}} {X}^{\top}\nu \Vert_2$ and $\text{dir}( {\Sigma}^{-\frac{1}{2}}{X}^{\top}\nu)$ are independent. Under \eqref{eq:data_full_cov_gen} and $H_0:{\mu}^\top \nu = 0_q$, we have that $ {\Sigma}^{-\frac{1}{2}} {X}^{\top}\nu \sim \mathcal{N}(0_q, \Vert\nu\Vert_2^2\textbf{I}_q)$. It then follows that $ {\Sigma}^{-\frac{1}{2}} {X}^{\top}\nu$ is rotationally invariant, and therefore $\Vert  {\Sigma}^{-\frac{1}{2}}  {X}^{\top}\nu \Vert_2$ is independent of $\text{dir}( {\Sigma}^{-\frac{1}{2}} {X}^{\top}\nu)$~\citep{Bilodeau1999-tl}.
\end{proof}
Then, recalling the definition of $\pKSigma$ in \eqref{eq:full_cov_p_val_k_means_chen}, we have that
{\small
\begin{align*} 
\hspace{-10mm}\pKSigma &= \pr_{H_0}\Big[ \Vert {\Sigma}^{-\frac{1}{2}}{X}^{\top}\nu \Vert_2 \geq  \Vert {\Sigma}^{-\frac{1}{2}}{x}^{\top}\nu  \Vert_2 \;\big\vert\;  \bigcap_{t=0}^{T} \bigcap_{i=1}^{n}\left\{c_i^{(t)}\left({X}\right) = c_i^{(t)}\left({x}\right)\right\},\, {\Pi}_{\nu}^\perp {X} = {\Pi}_{\nu}^\perp {x},\,\text{dir}\qty({\Sigma}^{-\frac{1}{2}}{X}^{\top}\nu) =  \text{dir}\qty({\Sigma}^{-\frac{1}{2}}{x}^{\top}\nu)\Big] \\
 &\hspace{-10mm}\overset{a.}{=}\pr_{H_0}\Bigg[ \Vert {\Sigma}^{-\frac{1}{2}}{X}^{\top}\nu \Vert_2 \geq  \Vert{\Sigma}^{-\frac{1}{2}} {x}^{\top}\nu  \Vert_2 \;\Bigg\vert\;  \bigcap_{t=0}^{T} \bigcap_{i=1}^{n}\left\{c_i^{(t)}\left({\Pi}_\nu^\perp {X} + \frac{\Vert {\Sigma}^{-\frac{1}{2}}{X}^\top \nu \Vert_2}{\Vert \nu \Vert_2^2} \nu \qty{\text{dir}\qty({\Sigma}^{-\frac{1}{2}}{X}^\top \nu)}^\top {\Sigma}^{\frac{1}{2}}\right) = c_i^{(t)}\left({x}\right)\right\},\,\\
 &\quad{\Pi}_{\nu}^\perp {X} = {\Pi}_{\nu}^\perp {x},\, \text{dir}({\Sigma}^{-\frac{1}{2}}{X}^{\top}\nu) =  \text{dir}({\Sigma}^{-\frac{1}{2}}{x}^{\top}\nu)\Bigg]  \\
 &\hspace{-10mm}\overset{b.}{=}\pr_{H_0}\Bigg[ \Vert {\Sigma}^{-\frac{1}{2}}{X}^{\top}\nu \Vert_2 \geq  \Vert {\Sigma}^{-\frac{1}{2}}{x}^{\top}\nu  \Vert_2 \;\Bigg\vert\;  \bigcap_{t=0}^{T} \bigcap_{i=1}^{n}\left\{c_i^{(t)}\left({\Pi}_\nu^\perp {x} + \frac{\Vert {\Sigma}^{-\frac{1}{2}}{X}^\top \nu \Vert_2}{\Vert \nu \Vert_2^2} \nu \qty{\text{dir}\qty({\Sigma}^{-\frac{1}{2}}{x}^\top \nu)}^\top {\Sigma}^{\frac{1}{2}}\right) = c_i^{(t)}\left({x}\right)\right\},\,\\
 &\quad{\Pi}_{\nu}^\perp {X} = {\Pi}_{\nu}^\perp {x},\, \text{dir}({\Sigma}^{-\frac{1}{2}}{X}^{\top}\nu) =  \text{dir}({\Sigma}^{-\frac{1}{2}}{x}^{\top}\nu)\Bigg]  \\
  &\hspace{-10mm}\overset{c.}{=}\pr_{H_0}\Bigg[ \Vert {\Sigma}^{-\frac{1}{2}}{X}^{\top}\nu \Vert_2 \geq  \Vert{\Sigma}^{-\frac{1}{2}} {x}^{\top}\nu  \Vert_2 \;\Bigg\vert\;  \bigcap_{t=0}^{T} \bigcap_{i=1}^{n}\left\{c_i^{(t)}\left({\Pi}_{\nu}^\perp {x} + \frac{\Vert{\Sigma}^{-\frac{1}{2}} {X}^\top \nu \Vert_2}{\Vert \nu \Vert_2^2} \nu \qty{\text{dir}\qty({\Sigma}^{-\frac{1}{2}}{x}^\top \nu)}^\top{\Sigma}^{\frac{1}{2}} \right) = c_i^{(t)}\left({x}\right)\right\}\Bigg].
\end{align*} 
}
Here, step $a.$ follows from substituting ${X}$ with the expression in \eqref{eq:perp_nu_identity_full_cov}. Step $b.$ follows from replacing ${\Pi}_\nu^\perp {X}$ and $\text{dir}({\Sigma}^{-\frac{1}{2}}{X}^{\top}\nu)$ with ${\Pi}_\nu^\perp {x}$ and $\text{dir}({\Sigma}^{-\frac{1}{2}}{x}^{\top}\nu)$, respectively. Finally, in step $c.$, we used Lemma~\ref{lemma:mutual_indep_full_cov}.
Now, under \eqref{eq:data_full_cov_gen} and $H_0:{\mu}^\top \nu = 0_q$, we have that $\Vert{\Sigma}^{-\frac{1}{2}} {X}^\top \nu \Vert_2 \sim  \Vert \nu \Vert_2 \chi_q$, which concludes the proof of \eqref{eq:p_val_single_param_full_cov}.

It remains to show that the test that rejects $H_0: {\mu}^\top \nu = 0$ when $\pKSigma\leq \alpha$ controls the selective Type I error, in the sense of \eqref{eq:selective_type_1}. We omit the proof here, as it follows directly from the proof of Proposition~\ref{prop:single_param_p} in Appendix~\ref{appendix:proof_p_val_prop}.

Next, we discuss how we could modify the results in Section~\ref{section:method} to compute the $p$-value $\pKSigma$. First note that according to Proposition~\ref{prop:full_cov_single_param_p}, it suffices to compute the set
\begin{align}
\label{eq:s_t_sigma}
\mathcal{S}_T^{\Sigma} = \qty{\phi:\bigcap_{t=0}^{T} \bigcap_{i=1}^{n}\left\{c_i^{(t)}\left({\Pi}_{\nu}^\perp {x} + \qty(\frac{\phi}{\Vert \nu \Vert_2^2}) \nu \qty{\text{dir}\qty({\Sigma}^{-\frac{1}{2}}{x}^\top \nu)}^\top{\Sigma}^{\frac{1}{2}} \right) = c_i^{(t)}\left({x}\right)\right\}}.
\end{align}
In addition, letting $\tilde{{x}}'(\phi)$ denote ${\Pi}_{\nu}^\perp {x} + \qty(\frac{\phi}{\Vert \nu \Vert_2^2}) \nu \qty{\text{dir}\qty({\Sigma}^{-\frac{1}{2}} {x}^\top \nu)}^\top{\Sigma}^{\frac{1}{2}}$, we see that $\tilde{{x}}'(\phi)$ is in fact a linear function of $\phi$ with
\begin{align}
\label{eq:x_phi_Sigma}
[\tilde{{x}}'(\phi)]_i = {x}_i  - \nu_i \frac{|| {x}^{\top}\nu||_2}{||\nu||_2^2} \text{dir}({x}^{\top}\nu) + \qty(\frac{|| {x}^{\top}\nu||_2}{||{\Sigma}^{-\frac{1}{2}} {x}^{\top}\nu||_2}\frac{\nu_i}{||\nu||_2^2}  \phi )\text{dir}({x}^{\top}\nu).
\end{align}


Therefore, a minor modification of Proposition~\ref{prop:phi_ineq} yields the following corollary.
\begin{corollary}
\label{cor:phi_ineq_Sigma}
Suppose the $k$-means clustering algorithm (see Algorithm~\ref{algo:k_means_alt_min}) with $K$ clusters the data ${x}$, when applied to the data ${x}$, runs for $T$ steps. Then, for the set $\mathcal{S}^{\Sigma}_T$ defined in \eqref{eq:s_t_sigma}, we have that
{\footnotesize
\begin{equation}
\begin{aligned}
&\hspace{-5mm}\mathcal{S}^{\Sigma}_T = \qty( \bigcap_{i=1}^{n}\bigcap_{k=1}^{K} \qty{\phi:  \left\Vert \qty[\tilde{{x}}'(\phi)]_{i} - m^{(0)}_{c_i^{(0)}({x})}(\phi) \right\Vert_2^2 \leq  \left\Vert \qty[\tilde{{x}}'(\phi)]_{i} - m^{(0)}_{k}(\phi) \right\Vert_2^2} ) \cap \\
&\hspace{-15mm}\vast(\bigcap_{t=1}^{T}\bigcap_{i=1}^{n}\bigcap_{k=1}^{K}\vast\{ \phi: \left\Vert \qty[\tilde{{x}}'(\phi)]_i - \sum_{i'=1}^n w_{i'}^{(t-1)}\qty(c_{i'}^{(t-1)}({x}))\qty[\tilde{{x}}'(\phi)]_{i'}\right\Vert_2^2 \leq \left\Vert \qty[\tilde{{x}}'(\phi)]_i - \sum_{i'=1}^n  w_{i'}^{(t-1)}(k)\qty[\tilde{{x}}'(\phi)]_{i'}  \right\Vert_2^2  \vast\}\vast) .
\end{aligned}
\end{equation}
 }
\end{corollary}

We also have the following extensions of Lemmas~\ref{lemma:simple_norm} and \ref{lemma:canonical_norm}, which enable efficient computation of the expressions in Corollary~\ref{cor:phi_ineq_Sigma}.

\begin{lemma}[Section 4.2 in \citet{Gao2020-yt}]
\label{lemma:simple_norm_Sigma}
For $\tilde{{x}}'(\phi)$  in \eqref{eq:x_phi_Sigma} and $\nu$  in \eqref{eq:nu_def},
{\small$\left\Vert\qty[\tilde{{x}}'(\phi)]_i- \qty[\tilde{{x}}'(\phi)]_j\right\Vert_2^2 = a' \phi^2 + b' \phi + \gamma'$}, where {\small$a' = \qty(\frac{|| {x}^{\top}\nu||_2}{||{\Sigma}^{-\frac{1}{2}}{x}^{\top}\nu||_2})^2\qty(\frac{\nu_i-\nu_j}{\Vert\nu\Vert_2^2} )^2 $, $b'=2 \qty(\frac{|| {x}^{\top}\nu||_2}{||{\Sigma}^{-\frac{1}{2}}{x}^{\top}\nu||_2})\left( \frac{\nu_i-\nu_j}{\Vert\nu\Vert_2^2} \left\langle {x}_i-{x}_j,  \emph{\text{dir}}({x}^{\top}\nu) \right\rangle  -\qty(\frac{\nu_i-\nu_j}{\Vert\nu\Vert_2^2} )^2 \Vert{x}^{\top}\nu\Vert_2 \right)$}, and {\small$\gamma'=\left\Vert {x}_i  -  {x}_j -(\nu_i-\nu_j) \frac{{x}^{\top}\nu}{||\nu||_2^2} \right\Vert_2^2$}. 
\end{lemma}

\begin{lemma} 
\label{lemma:canonical_norm_Sigma}
For $\tilde{{x}}'(\phi)$ in \eqref{eq:x_phi_Sigma}, $\nu$  in \eqref{eq:nu_def}, and $w_{i}^{(t)}(k)$  in \eqref{eq:w_i_t_k}, 
{\small $\left\Vert \qty[\tilde{{x}}'(\phi)]_i - \sum_{i'=1}^n w_{i'}^{(t-1)}(k)\qty[\tilde{{x}}'(\phi)]_{i'}   \right\Vert_2^2 = \tilde{a}' \phi^2 + \tilde{b}'\phi +\tilde{\gamma}'$}, where {\small $$\tilde{a}' = \frac{1}{\Vert\nu\Vert_2^4}\qty(\frac{|| {x}^{\top}\nu||_2}{||{\Sigma}^{-\frac{1}{2}}{x}^{\top}\nu||_2})^2\qty(\nu_i - \sum_{i'=1}^n w_{i'}^{(t-1)}(k) \nu_{i'} )^2,$$}
{\small $$\tilde{b}' = \qty(\frac{2|| {x}^{\top}\nu||_2}{\Vert\nu\Vert_2^2||{\Sigma}^{-\frac{1}{2}}{x}^{\top}\nu||_2}) \Bigg\{ \qty(\nu_i - \sum_{i'=1}^n  w_{i'}^{(t-1)}(k) \nu_{i'} ) \left\langle {x}_i - \sum_{i'=1}^n  w_{i'}^{(t-1)}(k) {x}_{i'},\emph{\text{dir}}({x}^{\top}\nu) \right\rangle - 
\frac{\Vert{x}^{\top}\nu \Vert_2}{\Vert\nu\Vert_2^4}\qty(\nu_i -  \sum_{i'=1}^n w_{i'}^{(t-1)}(k) \nu_{i'} )^2  \Bigg\},$$} and {\small$$\tilde{\gamma}' =  \left\Vert {x}_i - \sum_{i'=1}^n w_{i'}^{(t-1)}(k)  {x}_{i'} - \qty(\nu_i - \sum_{i'=1}^n w_{i'}^{(t-1)}(k) \nu_{i'} )\frac{{x}^\top \nu}{\Vert\nu\Vert_2^2} \right\Vert_2^2.$$}
\end{lemma}

Proofs of Lemmas~\ref{lemma:simple_norm_Sigma} and \ref{lemma:canonical_norm_Sigma} follow from the same set of calculations in the proofs of Lemmas~\ref{lemma:simple_norm} and \ref{lemma:canonical_norm} in Appendix~\ref{appendix:quad_ineq}.


\subsection{Proof of Proposition~\ref{prop:estimated_sigma}}
\label{appendix:proof_p_val_estimated_sigma}

Proof of Proposition~\ref{prop:estimated_sigma} is similar to the proof of Lemma 1 in \citet{markovic2018unifying} and the proof of Lemma 7 in \citet{Tibshirani2018-eo}.

We first present an auxiliary lemma. 
\begin{lemma}
\label{lemma:pkhat_cont_function}
For any $c_i^{(t)}(x),i=1,\ldots,n;t=1,\ldots, T$, $\pKhat(\hat\sigma)$ defined in \eqref{eq:p_val_chen_estimate_sigma} is a continuous and monotonically increasing function of $\hat\sigma$.
\end{lemma} 
\begin{proof}
By the definition in \eqref{eq:p_val_chen_estimate_sigma}, we have that
\begin{align}
\label{eq:p_khat_chi_sq}
\pKhat(\hat\sigma) &= \frac{\int_{\Vert {x}^\top \nu \Vert_2}^\infty (\frac{1}{2})^{q/2-1}\frac{t^{q-1}}{\Gamma(q/2)}\Vert \nu \Vert_2^{-q}\hat\sigma^{-q}\exp\qty(-\frac{t^2}{2\hat\sigma^2\Vert \nu \Vert_2^2}) 1\qty{t\in\mathcal{S}_T}dt}{\int_{0}^{\infty} (\frac{1}{2})^{q/2-1}\frac{t^{q-1}}{\Gamma(q/2)}\Vert \nu \Vert_2^{-q}\hat\sigma^{-q}\exp\qty(-\frac{t^2}{2\hat\sigma^2\Vert \nu \Vert_2^2}) 1\qty{t\in\mathcal{S}_T}dt},
\end{align}
where $\mathcal{S}_T$ defined in \eqref{eq:s_set} is a function of $c_i^{(t)}(x),i=1,\ldots,n;t=1,\ldots, T$. By inspection, \eqref{eq:p_khat_chi_sq} is a continuous function of $\hat\sigma$, because the product or ratio of two continuous functions is still continuous. It remains to show that \eqref{eq:p_khat_chi_sq} is increasing in $\hat\sigma$. This follows directly from Lemma S3. of \citet{Gao2020-yt}.
\end{proof}



Provided that $\hat\sigma$ converges to $\sigma$ in probability, we can combine Lemma~\ref{lemma:pkhat_cont_function} and the continuous mapping theorem to see that $\pKhat(\hat\sigma)$ converges to $\pK(\sigma)$ in probability, i.e., for all $\epsilon>0, \lim_{q\to\infty}\pr\qty(|\pKhat(\hat\sigma)-\pK(\sigma)|\geq \epsilon) = 0$. Next, letting $A_{q}$ denote the event $\bigcap_{t=0}^{T}\bigcap_{i=1}^{n} \left\{c_i^{(t)}\left({X}^{(q)}\right) = c_i^{(t)}\left({x}^{(q)}\right)\right\}$, we will show that under the assumptions in Proposition~\ref{prop:estimated_sigma}, $\pKhat(\hat\sigma)$ converges to $\pK(\sigma)$ in probability, \emph{conditional on $A_{q}$}. For any $\epsilon > 0$, we have that
{
\begin{align*}
\lim_{q\to\infty}&\pr_{H_0^{(q)}}\qty{|\pKhat(\hat\sigma)-\pK(\sigma)| \leq \epsilon \;\middle\vert\; A_{q}} \\
&\overset{a.}{=} \lim_{q\to\infty} \frac{\pr_{H_0^{(q)}}\qty{|\pKhat(\hat\sigma)-\pK(\sigma)| \leq \epsilon , A_{q} }}{\pr_{H_0^{(q)}}\qty(A_{q})} \\
&\overset{b.}{\geq} \lim_{q\to\infty} \frac{\pr_{H_0^{(q)}}\qty(A_{q} ) - \pr_{H_0^{(q)}}\qty{|\pKhat(\hat\sigma)-\pK(\sigma)| > \epsilon}}{\pr_{H_0^{(q)}}\qty(A_{q})} \\
&\overset{c.}{=}  \frac{\lim_{q\to\infty} \pr_{H_0^{(q)}}\qty(A_{q})-\lim_{q\to\infty}\pr_{H_0^{(q)}}\qty{|\pKhat(\hat\sigma)-\pK(\sigma)| > \epsilon}}{\lim_{q\to\infty} \pr_{H_0^{(q)}}\qty(A_{q})} \\
&\overset{d.}{=} \frac{\delta}{\delta} = 1.
\end{align*}}
Here, step $a.$ follows from Bayes rule, and the observation that the denominator is non-zero for finite $q$. In step $b.$, we used the  lower bound that for events $A,B$ defined on the same probability space, $\pr(A\cap B) = \pr(A) - \pr(A\setminus B) \geq  \pr(A) - \pr(B^C)$. Next, $c.$ follows from distributing the limit, which is valid because of the assumption that $\lim_{q\to\infty} \pr_{H_0^{(q)}}\qty(A_{q}) = \delta >0$; finally, $d.$ follows from the fact that $\pKhat(\hat\sigma)$ converges to $\pK(\sigma)$ in probability for any sequence of $\mu^{(q)},q=1,2,\ldots$, which implies the convergence under $H_0:{\mu^{(q)}}^{\top} \nu^{(q)} = 0$ as well.

Finally, we have that
\begin{align}
\lim_{q\to\infty}\pr_{H_0^{(q)}}\qty{\pKhat(\hat\sigma)\leq \alpha \;\middle\vert\; A_q }  \overset{a.}{=} \lim_{q\to\infty} \pr_{H_0^{(q)}}\qty{\pK(\sigma)\leq \alpha \;\middle\vert\; A_q}  \overset{b.}{=} \lim_{q\to\infty} \alpha = \alpha.
\end{align}
Here, step $a.$ follows from $\pKhat(\hat\sigma)$ converging to $\pK(\sigma)$ in probability, \emph{conditional on $A_{q}$}. Step $b.$ follows from the fact that the result of Proposition~\ref{prop:single_param_p} applies for any positive integer $q$. This completes the proof of Proposition~\ref{prop:estimated_sigma}.

Proposition~\ref{prop:estimated_sigma} assumes that we have a consistent estimator $\hat\sigma$ of $\sigma$. In Appendix~\ref{appendix:variance_estimation}, we analyze different estimators of $\sigma$ in \eqref{eq:data_gen}, and prove that, under appropriate sparsity assumptions on ${\mu}$ in \eqref{eq:data_gen},  $\hat\sigma_{\text{MED}}$ in \eqref{eq:sigma_hat_MED} is a consistent estimator for $\sigma$.

As an alternative, we can also use an asymptotically conservative estimator of $\sigma$ as in \citet{Gao2020-yt}. This leads to an asymptotically conservative $p$-value; details are stated in Corollary~\ref{cor:estimated_conserv_sigma}.

\begin{corollary}
\label{cor:estimated_conserv_sigma}
For $q=1,2,\ldots,$ suppose that ${X}^{(q)} \sim \mathcal{MN}_{n\times q}\qty({\mu}^{(q)}, \emph{\textbf{I}}_n, \sigma^2 \emph{\textbf{I}}_q)$. Let ${x}^{(q)}$ be a realization from ${X}^{(q)}$ and 
$c_i^{(t)}(\cdot)$ be the cluster to which the $i$th observation is assigned during the $t$th iteration of step 3b. of Algorithm~\ref{algo:k_means_alt_min}. Consider the sequence of null hypotheses $H_0^{(q)}: {\mu^{(q)}}^\top \nu^{(q)} = 0_q$, where $\nu^{(q)}$ defined in \eqref{eq:nu_def} is the contrast vector resulting from applying $k$-means clustering on ${x}^{(q)}$. Suppose that (i) $\hat\sigma$ is an asymptotically conservative estimator of $\sigma$, i.e., $\lim_{q\to\infty}\pr\qty(\hat\sigma({X}^{(q)})\geq \sigma) = 1$; and (ii) there exists $\delta \in (0,1)$ such that $\lim_{q\to\infty}\pr_{H_0^{(q)}}\qty[\bigcap_{t=0}^{T}\bigcap_{i=1}^{n} \left\{c_i^{(t)}\left({X}^{(q)}\right) = c_i^{(t)}\left({x}^{(q)}\right)\right\} ] > \delta$.
Then, $\forall \alpha \in (0,1)$, we have that 
$\lim_{q\to\infty}\pr_{H_0^{(q)}}\qty[\pKhat(\hat\sigma)\leq \alpha \;\middle\vert\; \bigcap_{t=0}^{T}\bigcap_{i=1}^{n} \left\{c_i^{(t)}\left({X}^{(q)}\right) = c_i^{(t)}\left({x}^{(q)}\right)\right\} ] \leq \alpha$.
\end{corollary}
We omit the proof of Corollary~\ref{cor:estimated_conserv_sigma}, as it follows directly from combining the proof of Proposition~\ref{prop:estimated_sigma} and the fact that $\pKhat(\hat\sigma)$ is a  monotonically increasing function of $\hat\sigma$ (see Lemma~\ref{lemma:pkhat_cont_function}).

Finally, we remark that, in principle, the result in Proposition~\ref{prop:estimated_sigma} can be extended to an unknown covariance matrix ${\Sigma}$. However, estimating ${\Sigma}$ is challenging, especially when $q$ is comparable to, or larger than, $n$~\citep{Rousseeuw1987-gb,Bickel2008-if,Avella-Medina2018-qd}. It may be possible to leverage recent advances in robust covariance matrix estimation (e.g., \citet{Han2014-uw,Chen2018-qg,Belomestny2019-ni}) to obtain a consistent estimator of ${\Sigma}$ under model~\eqref{eq:data_full_cov_gen}. 

\subsection{Estimating $\sigma$ in \eqref{eq:data_gen}}
\label{appendix:variance_estimation}
Proposition~\ref{prop:estimated_sigma} states that, under appropriate assumptions, a consistent estimator of $\sigma$ in \eqref{eq:data_gen} leads to asymptotic selective Type I error control. In this section, we analyze the asymptotic behavior of the two variance estimators considered in Section~\ref{section:sim}, $\hat\sigma^2_{\text{MED}}$ and $\hat\sigma^2_{\text{Sample}}$. In particular, we prove that under model \eqref{eq:data_gen} and a sparsity assumption on ${\mu}$ (defined in \eqref{eq:data_gen}), a close analog of $\hat\sigma_{\text{MED}}^2$ in \eqref{eq:sigma_hat_MED} that does not subtract the column median is a consistent estimator of $\sigma^2$. Moreover, we prove that $\hat\sigma_{\text{Sample}}^2$ is a conservative estimator of $\sigma^2$, and characterize its exact bias.



We first introduce an auxiliary result that specifies the rate of convergence for a median-based estimator of the variance in the sparse vector model~\citep{Comminges2021-gm}. For a vector $\theta \in \mathbb{R}^n$, we use $\Vert\theta\Vert_0$ to denote its $\ell_0$ norm, i.e. $\Vert\theta\Vert_0 = \sum_{i=1}^n 1\qty{\theta_i\neq 0}$.

\begin{lemma}[Proposition 6 in \citet{Comminges2021-gm}]
\label{eq:lemma_sigma_estimation}
Consider the model 
\begin{align}
\label{eq:med_model}
Y_i = \theta_i + \sigma \xi_i, \quad i=1,\ldots, d,
\end{align}
where $\sigma$ is unknown, and the independently and identically distributed noise $\xi_i$ satisfies that (i) $\emph{\E}(\xi_i)= 0$; (ii) $\emph{\E}(\xi_i^2)=1$; and (iii) $\emph{\E}\qty(|\xi_i|^{2+\epsilon})<\infty$ for some $\epsilon>0$.
We further assume that the signal $\theta$ is $s$-sparse, i.e., $\Vert \theta \Vert_0 \leq s$. Denoting by $M_{\xi_1^2}$ the median of $\xi_1^2$, we consider the following estimator of $\sigma^2$:
\begin{align}
\label{eq:sigma_med_proof_linear}
\bar\sigma_{\text{MED}}^2 = \emph{\text{median}}(Y_1^2,\ldots, Y_d^2)/M_{\xi_1^2}.
\end{align}
Then, there exist constants $\gamma \in (0,1/8)$, $C>0$ depending only on the cumulative distribution function of $\xi_1$ such that for all integers $s$ and $d$ satisfying $1\leq s < \gamma d$,
\begin{align}
\sup_{\sigma>0}\sup_{\Vert \theta\Vert_0 \leq s} \frac{1}{\sigma^2}\emph{\E}\qty{\qty|\bar\sigma_{\text{MED}}^2-\sigma^2|} \leq C \max\qty(\frac{1}{{d}^{1/2}}, \frac{s}{d}).
\end{align}
\end{lemma}

Building on Lemma~\ref{eq:lemma_sigma_estimation}, in Corollary~\ref{prop:sigma_med_consistency}, we analyze the properties of an estimator closely related to $\hat{\sigma}^2_{\text{MED}}$ in \eqref{eq:sigma_hat_MED}. In particular, this estimator $\tilde\sigma_{\text{MED}}^2$ does not subtract the median of each column in the input data. While $\hat{\sigma}^2_{\text{MED}}$ and $\tilde\sigma_{\text{MED}}^2$ are very similar provided that ${\mu}$ is sparse, we expect $\hat\sigma^2_{\text{MED}}$ to perform better empirically in scenarios where ${\mu}$ is sparse \emph{up to a constant shift}, i.e., there exists a matrix ${C}$ such that (i) each column of ${C}$ takes on the same value; and (ii) ${\mu}+{C}$ is sparse.

\begin{corollary}
\label{prop:sigma_med_consistency}
Under model \eqref{eq:data_gen}, consider 
\begin{align}
\label{eq:sigma_med_proof_cluster}
\tilde\sigma_{\text{MED}}^2({X}) = \qty{\underset{1\leq i\leq n, 1\leq j\leq q}{\emph{\text{median}}}\qty({{X}}_{ij}^2)}/M_{\chi_1^2},
\end{align}
where $M_{\chi_1^2}$ is the median of the $\chi_1^2$ distribution.
Then, there exist constants $\gamma_0 \in (0,1/8)$, $c_0>0$ such that for all integers $s$ and $q$ satisfying $1\leq s < \gamma_0 q$,
\begin{align}
\sup_{\sigma>0}\;\sup_{\underset{1\leq i\leq n}{\max}\Vert {\mu}_i \Vert_0 \leq s} \frac{1}{\sigma^2}\emph{\E}\qty{\qty|\tilde\sigma_{\text{MED}}^2-\sigma^2|} \leq c_0 \max\qty(\frac{1}{{(nq)}^{1/2}}, \frac{s}{q}).
\end{align}
\end{corollary}
\begin{proof}
First note that \eqref{eq:data_gen} can be re-written into the form of \eqref{eq:med_model}:
\begin{align}
\label{eq:equiv_data_gen}
X_{ij} = \mu_{ij} + \sigma \xi_{ij}, \quad i=1,\ldots, n,\, j = 1,\ldots, q,
\end{align}
where $\xi_{ij}$ is independently and identically distributed as $\mathcal{N}(0,1)$. Therefore, the estimator $\tilde\sigma_{\text{MED}}^2({X})$ in \eqref{eq:sigma_med_proof_cluster} is the estimator \eqref{eq:sigma_med_proof_linear} applied to the model \eqref{eq:equiv_data_gen}. Moreover, $\underset{1\leq i\leq n}{\max}\Vert {\mu}_i \Vert_0 \leq s$ implies that $ \sum_{i=1}^n \sum_{j=1}^q 1\qty{\mu_{ij}\neq 0} \leq ns$. Applying  Lemma~\ref{eq:lemma_sigma_estimation}, we have that 
\begin{align*}
\sup_{\sigma>0}\;\sup_{\underset{1\leq i\leq n}{\max}\Vert {\mu}_i \Vert_0 \leq s} \frac{1}{\sigma^2}\E\qty{\qty|\tilde\sigma_{\text{MED}}^2({X})-\sigma^2|} \leq c_0 \max\qty{\frac{1}{{(nq)}^{1/2}}, \frac{ns}{nq}} = c_0 \max\qty{\frac{1}{{(nq)}^{1/2}}, \frac{s}{q}},
\end{align*}
where $c_0$ is some universal constant.
\end{proof}
In words, Corollary~\ref{prop:sigma_med_consistency} states that under model \eqref{eq:data_gen}, the rate of convergence of $\tilde\sigma_{\text{MED}}^2$ in mean (and therefore, in probability) is $\max\qty{{1}/{{(nq)}^{1/2}}, {s}/{q}}$. In particular, $\tilde\sigma_{\text{MED}}^2$ is a consistent estimator of $\sigma^2$ provided that $s/q \to 0$ as $q \to \infty$.

Next, we investigate the property of the sample variance estimator $\hat\sigma^2_{\text{Sample}}$.
\begin{proposition}
\label{prop:sigma_sample_consistency}
Under model \eqref{eq:data_gen}, for $\hat\sigma_{\text{Sample}}^2({X}) = \sum_{i=1}^n\sum_{j=1}^q \qty({X}_{ij}-\bar{{X}}_j)^2/(nq-q)$, we have that
\begin{align}
\label{eq:expectation_var_sample}
\emph{\E}\qty{\hat\sigma_{\text{Sample}}^2({X})} - \sigma^2 =  \frac{1}{2n(n-1)q}\sum_{j=1}^q\sum_{i=1}^n\sum_{i'=1}^n (\mu_{ij}-\mu_{i'j})^2.
\end{align}
Moreover, for any integers $s$ and $q$ such that $ns\leq q$, we have that, for some constant $\tilde{c}_0$,
\begin{align}
\label{eq:sigma_sample_statement}
\sup_{\sigma>0}\;\sup_{\underset{1\leq i\leq n}{\max}\Vert {\mu}_i \Vert_0 \leq s} \frac{1}{\sigma^2}\emph{\E}\qty{\qty|\hat\sigma_{\text{Sample}}^2({X})-\sigma^2|} \geq \tilde{c}_0\frac{s}{q}.
\end{align}
\end{proposition}
\begin{proof}
We start with the proof of \eqref{eq:expectation_var_sample}. Under \eqref{eq:data_gen}, the following holds:
\begin{align*}
\E\qty{\hat\sigma_{\text{Sample}}^2({X})} &= \E\qty{\sum_{i=1}^n\sum_{j=1}^q \qty({X}_{ij}-\bar{{X}}_j)^2/(nq-q)} \\
&= \frac{1}{(n-1)q} \E\qty[\sum_{i=1}^n\sum_{j=1}^q  \qty{X_{ij}^2 - (\bar{{X}}_j)^2} ] \\
&=  \frac{1}{(n-1)q} \sum_{i=1}^n\sum_{j=1}^q \qty[\sigma^2 + \mu_{ij}^2 - \qty{\frac{\sigma^2}{n}+\frac{1}{n^2}\qty(\sum_{i'=1}^n \mu_{i'j})^2} ] \\
&= \sigma^2 + \frac{1}{n^2(n-1)q}\sum_{i=1}^n\sum_{j=1}^q  \qty{ n^2\mu_{ij}^2  - \qty(\sum_{i'=1}^n \mu_{i'j})^2} \\
&= \sigma^2 + \frac{1}{n(n-1)q}\sum_{j=1}^q \qty{ \qty(\sum_{i=1}^nn\mu_{ij}^2)  - \qty(\sum_{i'=1}^n \mu_{i'j})^2} \\
&= \sigma^2 + \frac{1}{2n(n-1)q}\sum_{j=1}^q  \sum_{i=1}^n  \sum_{i'=1}^n \qty(\mu_{ij}-\mu_{i'j})^2.
\end{align*}
Here, the last equality follows from Langrange's identity, which states that $\qty(\sum_{i=1}^n a_i^2 )\qty(\sum_{i=1}^n b_i^2 ) - \qty(\sum_{i=1}^n a_ib_i )^2 = 1/2\sum_{i=1}^n\sum_{i'=1}^n \qty(a_i b_{i'}-a_{i'}b_i)^2$. 

To prove the second statement, we consider a specific matrix $\tilde{{\mu}} \in \mathbb{R}^{n\times q}$ with exactly $ns\leq q$ non-zero entries. In addition, each column of $\tilde{{\mu}}$ has at most one non-zero entry and each row of $\tilde{{\mu}}$ has exactly $s$ non-zero entries. This is possible because $ns$ is assumed to be less than $q$. Finally, we assume that the square of the minimal non-zero entry of $\tilde{{\mu}}$, $\underset{i,j:\tilde{\mu}_{ij} \neq 0}{\min}\tilde{\mu}_{ij}^2$, is lower bounded by some universal constant $M$. Then, we have that
{
\begin{align*}
\sup_{\sigma>0}\;\sup_{\underset{1\leq i\leq n}{\max}\Vert {\mu}_i \Vert_0 \leq s}& \frac{1}{\sigma^2}\E\qty{\qty|\hat\sigma_{\text{Sample}}^2({X})-\sigma^2|} \\
&\overset{a.}{\geq} \sup_{\sigma>0} \frac{1}{\sigma^2}\E_{X\sim\mathcal{MN}(\tilde{{\mu}},\textbf{I}_n,\sigma^2\textbf{I}_q)}\qty{\qty|\hat\sigma_{\text{Sample}}^2({X})-\sigma^2|} \\
&\overset{b.}{\geq} \sup_{\sigma>0} \frac{1}{\sigma^2}\E_{X\sim\mathcal{MN}(\tilde{{\mu}},\textbf{I}_n,\sigma^2\textbf{I}_q)}\qty{\hat\sigma_{\text{Sample}}^2({X})-\sigma^2} \\
&\overset{c.}{\geq} \sup_{\sigma>0}\frac{1}{\sigma^2}\frac{1}{2n(n-1)q}\sum_{j=1}^q\sum_{i=1}^n\sum_{i'=1}^n (\tilde{\mu}_{ij}-\tilde{\mu}_{i'j})^2 \\
&\overset{}{\geq} \sup_{\sigma>0}\frac{1}{\sigma^2}\frac{1}{2n(n-1)q}\sum_{j=1}^q\sum_{i=1}^n\sum_{i'=1}^n 1\qty{\tilde{\mu}_{ij}\neq 0}1\qty{\tilde{\mu}_{i'j}= 0} (\tilde{\mu}_{ij}-\tilde{\mu}_{i'j})^2 \\
&\overset{d.}{\geq} \sup_{\sigma>0}\frac{1}{\sigma^2}\frac{M(n-1)ns}{2n(n-1)q} \\
&\overset{}{\geq}\tilde{c}_0\frac{s}{q}.
\end{align*}
}
Here, $a.$ follows from picking any $\tilde{{\mu}}$ satisfying the conditions outlined above, since by construction, $\underset{1\leq i\leq n}{\max}\Vert \tilde{{\mu}}_i \Vert_0 = s$. Steps $b.$ and $c.$ follow from the inequality $\E\qty(|X|) \geq \E(X)$ and the expression for $\E\qty{\hat\sigma_{\text{Sample}}^2({X})}$ in \eqref{eq:expectation_var_sample}, respectively. Finally, to prove $d.$, we note that for each of the $ns$ columns with exactly one non-zero element, there are $n-1$ pairs of $(i,i'), i=1,\ldots, n;i'=1,\ldots ,n$ such that the product $1\qty{\tilde{\mu}_{ij}\neq 0}1\qty{\tilde{\mu}_{i'j}= 0}$ is non-zero. Moreover, each of pair contributes at least $M$ by the assumption that $\min_{i,j: \tilde{\mu}_{ij} \neq 0}\tilde{\mu}_{ij}^2\geq M$. 
\end{proof}

Contrasting the results in Corollary~\ref{prop:sigma_med_consistency} and Proposition~\ref{prop:sigma_sample_consistency}, we note that, under \eqref{eq:data_gen}, the convergence of $\tilde\sigma^2_{\text{MED}}$ depends critically on the sparsity parameter $s$ (or, equivalently, the $\ell_0$ norm of ${\mu}_i$), whereas the convergence of $\hat\sigma_{\text{Sample}}^2$ is determined by $\sum_{j=1}^q\sum_{i=1}^{n}\sum_{i'=1}^{n}\qty({\mu}_{ij}-{\mu}_{i'j})^2$. Thus, in scenarios where the underlying means ${\mu}_i,i=1,\ldots,n$ are sparse (e.g., \eqref{eq:power_model} in Section~\ref{section:sim}), we expect $\tilde\sigma_{\text{MED}}^2$ (and therefore its ``centered'' analog $\hat\sigma_{\text{MED}}^2$ in \eqref{eq:sigma_hat_MED}) to be a less conservative estimator of $\sigma^2$. As a result, we expect the test based on $\pKhat(\hat\sigma_{\text{MED}})$ to be more powerful than that based on $\pKhat(\hat\sigma_{\text{Sample}})$, as shown in Figure~\ref{fig:sim_power} of Section~\ref{section:sim}.

\subsection{Additional power comparisons}
\label{appendix:additional_simulation}

In Section~\ref{section:cond_power}, we compared the conditional power of the tests based on $\pK$, $\pKhat(\hat\sigma_{\text{MED}})$, and $\pKhat(\hat\sigma_{\text{Sample}})$ under \eqref{eq:power_model}. Here, we conduct two additional analyses.

In the first analysis, we consider a different notion of power that does not condition on $\hat{\mathcal{C}}_1$ and $\hat{\mathcal{C}}_2$ being true clusters. In this case, comparing the power of the tests requires a bit of care, because the effect size $\Vert{\mu}^\top \nu\Vert_2$ may differ across simulated datasets from the same data-generating distribution. As a result, we consider the power of the tests \emph{as a function of $\Vert{\mu}^\top \nu\Vert_2$}. We fit a regression spline using the \texttt{gam} function in the \texttt{R} package \texttt{mgcv}~\citep{wood_2017} to obtain a smooth estimate of power on the same simulated datasets from Section~\ref{section:cond_power}. 
The results are in Figure~\ref{fig:smooth_power}. The power of the tests that reject $H_0$ if 
$\pK$, $\pKhat(\hat\sigma_{\text{MED}})$, or $\pKhat(\hat\sigma_{\text{Sample}})$ is less than $\alpha=0.05$ increases as $\Vert{\mu}^\top \nu\Vert_2$ increases. For a given value of $\Vert{\mu}^\top \nu\Vert_2$ and $\sigma$, the test based on $\pK$ has the highest power, followed by that based on $\pKhat(\hat\sigma_{\text{MED}})$; the test based on $\pKhat(\hat\sigma_{\text{Sample}})$ has the lowest power. 

\begin{figure}[htbp!]
\centering
\includegraphics[width=0.7\linewidth]{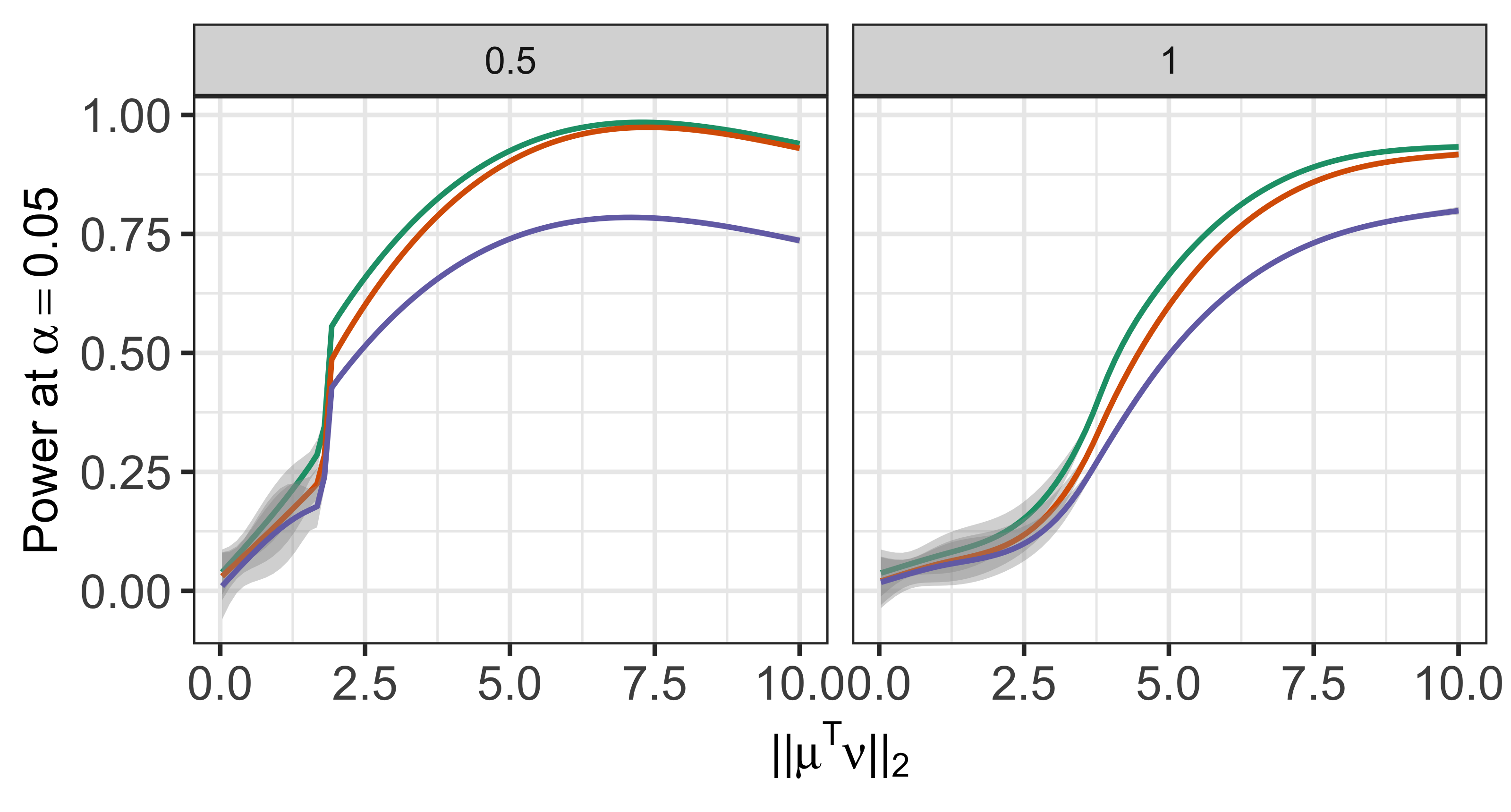}
\vspace*{-4mm}
\caption{\textit{Left }: Additional analysis of the data in Section~\ref{section:cond_power} with $\sigma=0.5$. We fit a regression spline to display the power of the tests based on $\pK$ (green line), $\pKhat(\hat\sigma_{\text{MED}})$ (orange line), and $\pKhat(\hat\sigma_{\text{Sample}})$ (purple line) as a function of $\Vert{\mu}^\top \nu\Vert_2$. \textit{Right }: Same as left, but for $\sigma=1$.}
\label{fig:smooth_power}
\end{figure}

In the second analysis, we consider the conditional power (defined in \eqref{eq:conditional_power}) of the tests based on $\pK$, $\pKhat(\hat\sigma_{\text{MED}})$, and $\pKhat(\hat\sigma_{\text{Sample}})$ under a different data generating model than \eqref{eq:power_model}. We generate data from \eqref{eq:data_gen} with $n=150$ and 
{
\begin{align}
\label{eq:less_sparse_power_model}
\hspace{-6mm}{\mu}_1 =\ldots = {\mu}_{\frac{n}{3}} = \begin{bmatrix}
{\theta}_1 \\ 0_{0.9q}
\end{bmatrix}, \; {\mu}_{\frac{n}{3}+1}=\ldots = {\mu}_{\frac{2n}{3}} = \begin{bmatrix}
 {\theta}_2\\0_{0.9q} 
\end{bmatrix} ,\;
{\mu}_{\frac{2n}{3}+1}=\ldots = {\mu}_{n} = \begin{bmatrix}
{\theta}_3 \\ 0_{0.9q}
\end{bmatrix},
\end{align}}
where, $q$ is taken to be a multiple of 10, and for $\delta>0$, ${\theta} \in \mathbb{R}^{3\times0.1q}$ has orthogonal rows, with $\Vert{\theta}_i\Vert_2^2 = \delta/2$ for $i=1,2,3$. As in Section~\ref{section:cond_power}, we can think of $\mathcal{C}_1 = \{1,\ldots,n/3\},\mathcal{C}_2 = \{(n/3)+1,\ldots,(2n/3)\},\mathcal{C}_3 = \{(2n/3)+1,\ldots,n\}$ as ``true clusters''. Under \eqref{eq:less_sparse_power_model}, the pairwise distance between each pair of true clusters is $\delta$.

\begin{figure}[htbp!]
\begin{centering}
\hspace{10mm} (a) \hspace{55mm} (b) \\
\end{centering}
\centering
\includegraphics[width=0.8\linewidth]{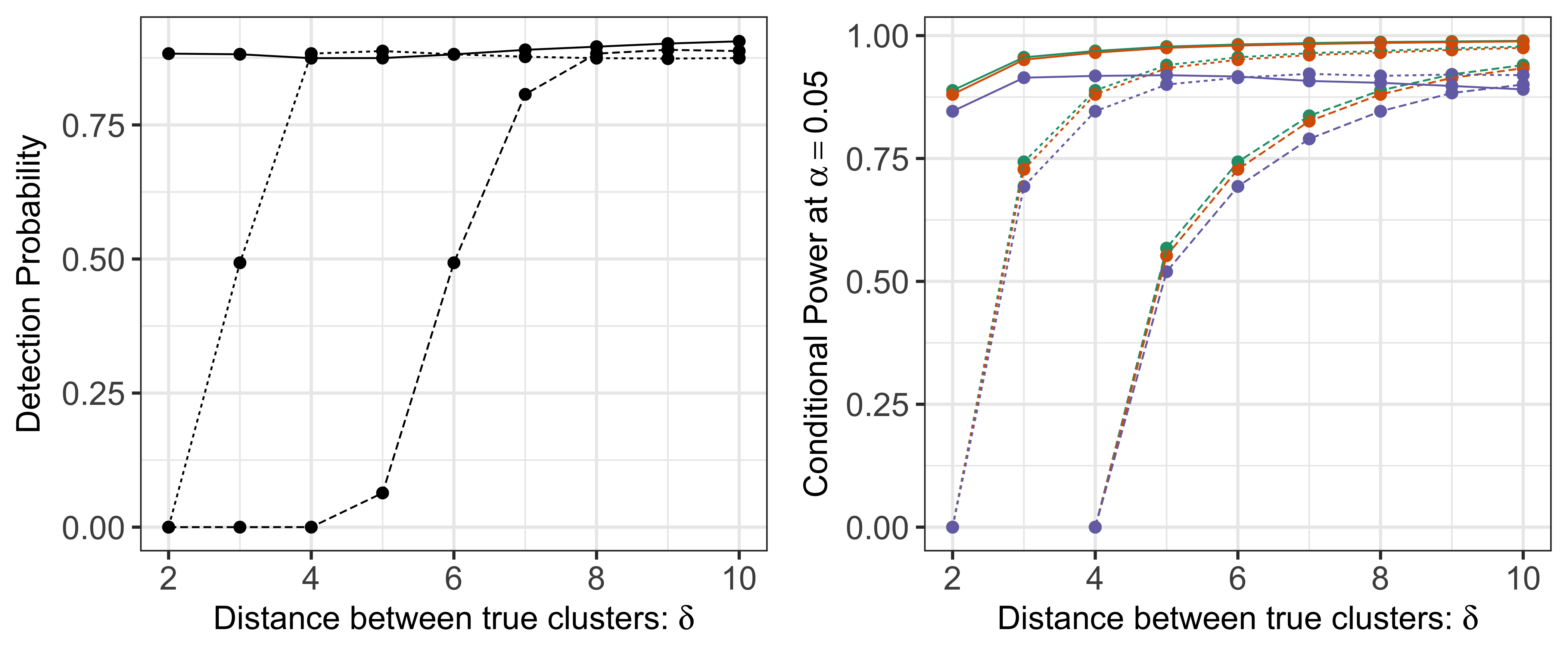}

\begin{centering}
\hspace{10mm} (c) \hspace{55mm} (d) \\
\end{centering}
\centering
\includegraphics[width=0.8\linewidth]{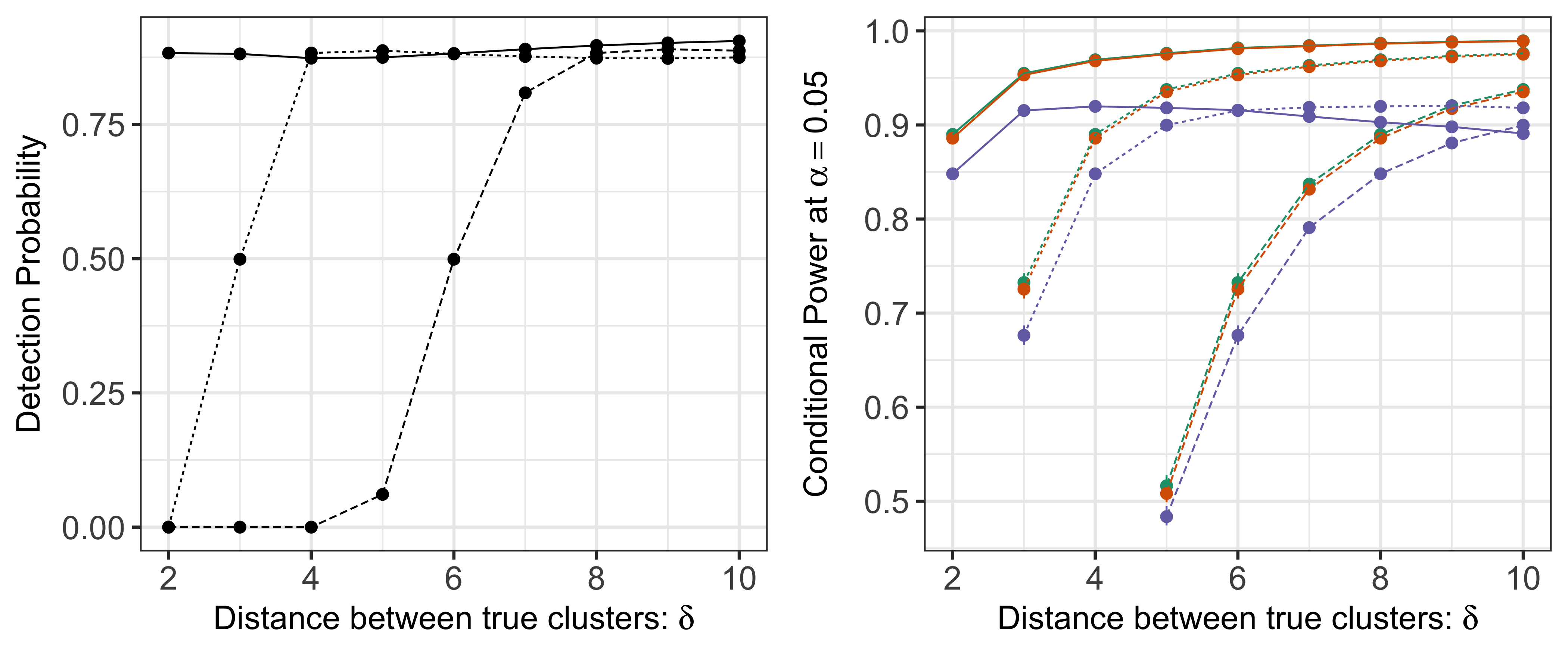}
\vspace*{-4mm}
\caption{\textit{(a): } Detection probability defined in \eqref{eq:detect_p} for $k$-means clustering with $K=3$ under model \eqref{eq:data_gen} with $n=150$, $q=50$, and ${\mu}$ in \eqref{eq:less_sparse_power_model}, across $\delta  = \Vert {\theta}_i - {\theta}_j \Vert_2$ in \eqref{eq:less_sparse_power_model} and $\sigma=0.25$ (solid lines), $0.5$ (dashed lines), and $1$ (long-dashed lines).
\textit{(b): } The conditional power \eqref{eq:conditional_power} at $\alpha=0.05$ for the tests based on $\pK$ (green), $\pKhat(\hat{\sigma}_{\text{MED}})$ (orange), and $\pKhat(\hat{\sigma}_{\text{Sample}})$ (purple), under model \eqref{eq:data_gen} with $n=150$, $q=50$, and ${\mu}$ in \eqref{eq:less_sparse_power_model}. \textit{(c): } Same as (a), but for ${\mu}$ in \eqref{eq:power_model}. \textit{(d): } Same as (b), but for ${\mu}$ in \eqref{eq:power_model}.}
\label{fig:sim_power_q_50}
\end{figure}

We generate $M=100,000$ datasets from \eqref{eq:less_sparse_power_model} with $q=50$,$\sigma=0.25, 0.5,1$, and $\delta=2,3,\ldots,10$. For each simulated dataset, we apply $k$-means clustering with $K=3$ and reject $H_0:{\mu}^\top \nu = 0_q$ if $\pK$, $\pKhat(\hat\sigma_{\text{MED}})$, or $\pKhat(\hat\sigma_{\text{Sample}})$ is less than $\alpha = 0.05$. Figure~\ref{fig:sim_power_q_50}(a) displays the detection probability \eqref{eq:detect_p} of $k$-means clustering as a function of $\delta$ in \eqref{eq:less_sparse_power_model}. Under model \eqref{eq:data_gen}, the detection probability increases as a function of $\delta$ in \eqref{eq:less_sparse_power_model} across all values of $\sigma$. For a given value of $\delta$, a larger value of $\sigma$ leads to lower detection probability. Figure~\ref{fig:sim_power_q_50}(b) displays the conditional power \eqref{eq:conditional_power} for the tests based on $\pK$, $\pKhat(\hat\sigma_{\text{MED}})$, and $\pKhat(\hat\sigma_{\text{Sample}})$. For some combinations of $\delta$ and $\sigma$, the conditional power is not displayed, because the true clusters are never recovered in simulation. For all tests and values of $\sigma$ under consideration, conditional power is an increasing function of $\delta$. For a given test and a value of $\delta$, smaller $\sigma$ leads to higher conditional power. Moreover, for the same values of $\delta$ and $\sigma$, the test based on $\pK$ has the highest conditional power, followed closely by the test based of $\pKhat(\hat{\sigma}_{\text{MED}})$. Using $\pKhat(\hat{\sigma}_{\text{Sample}})$ leads to a less powerful test, especially for larger values of $\delta$. As a comparison, we included the detection probability and conditional power under model \eqref{eq:power_model} with $q=50$ in panels (c) and (d) of Figure~\ref{fig:sim_power_q_50}. The tests under consideration behave qualitatively similarly as a function of $\delta$ and $\sigma$. Under \eqref{eq:power_model}, we observe an even larger gap between the power of the test based on $\pKhat(\hat{\sigma}_{\text{Sample}})$ and the power of the test based on $\pKhat(\hat{\sigma}_{\text{MED}})$.

\clearpage
\subsection{Additional results for real data applications}
\label{appendix:additional_real_data}

In this section, we visualize the estimated clusters for the single cell RNA-sequencing data in Section~\ref{subsection:zheng_real_data}.

\begin{figure}[htbp!]
\centering
\begin{subfigure}{0.4\linewidth}
\includegraphics[width=\linewidth]{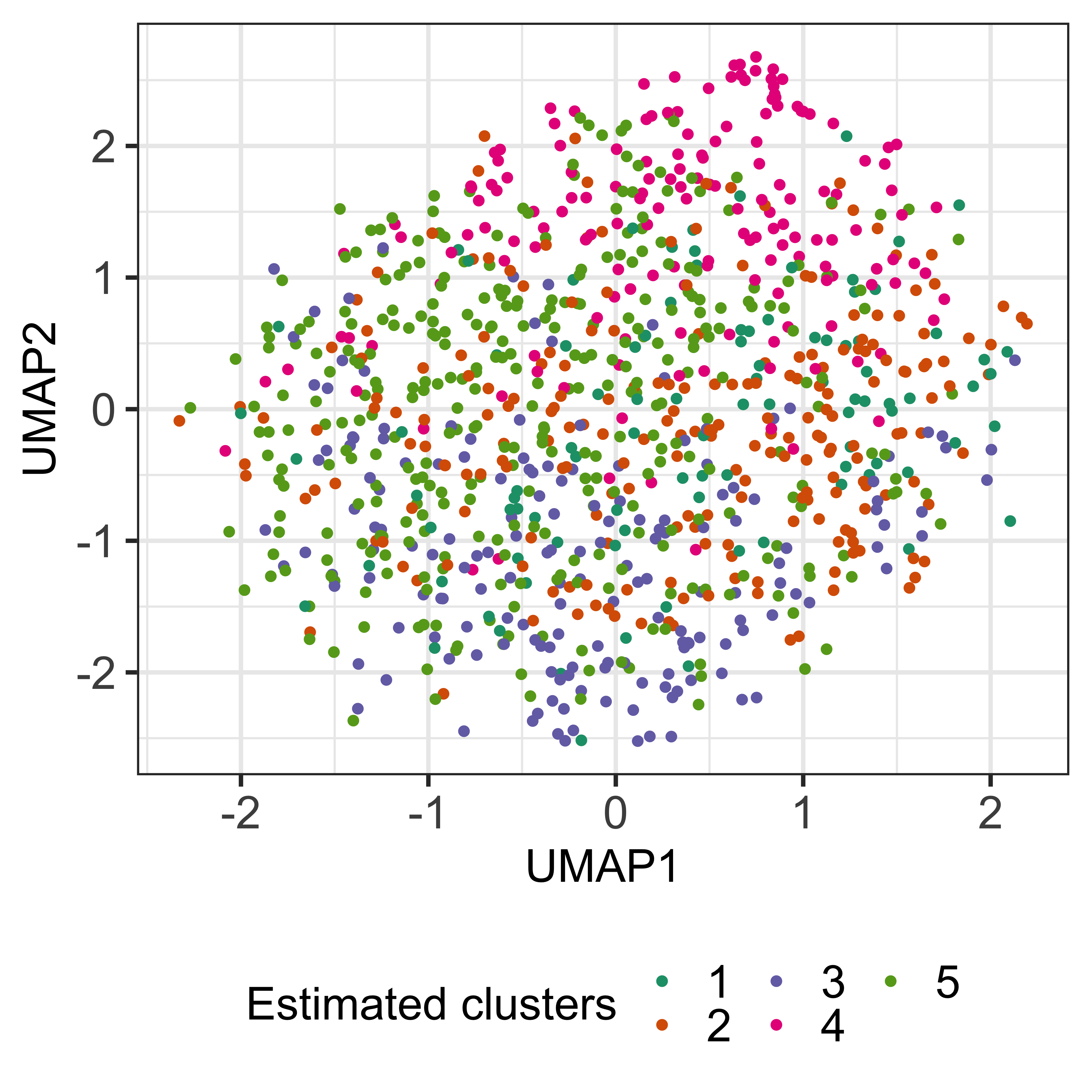}
\end{subfigure}
\begin{subfigure}{0.4\linewidth}
\includegraphics[width=\linewidth]{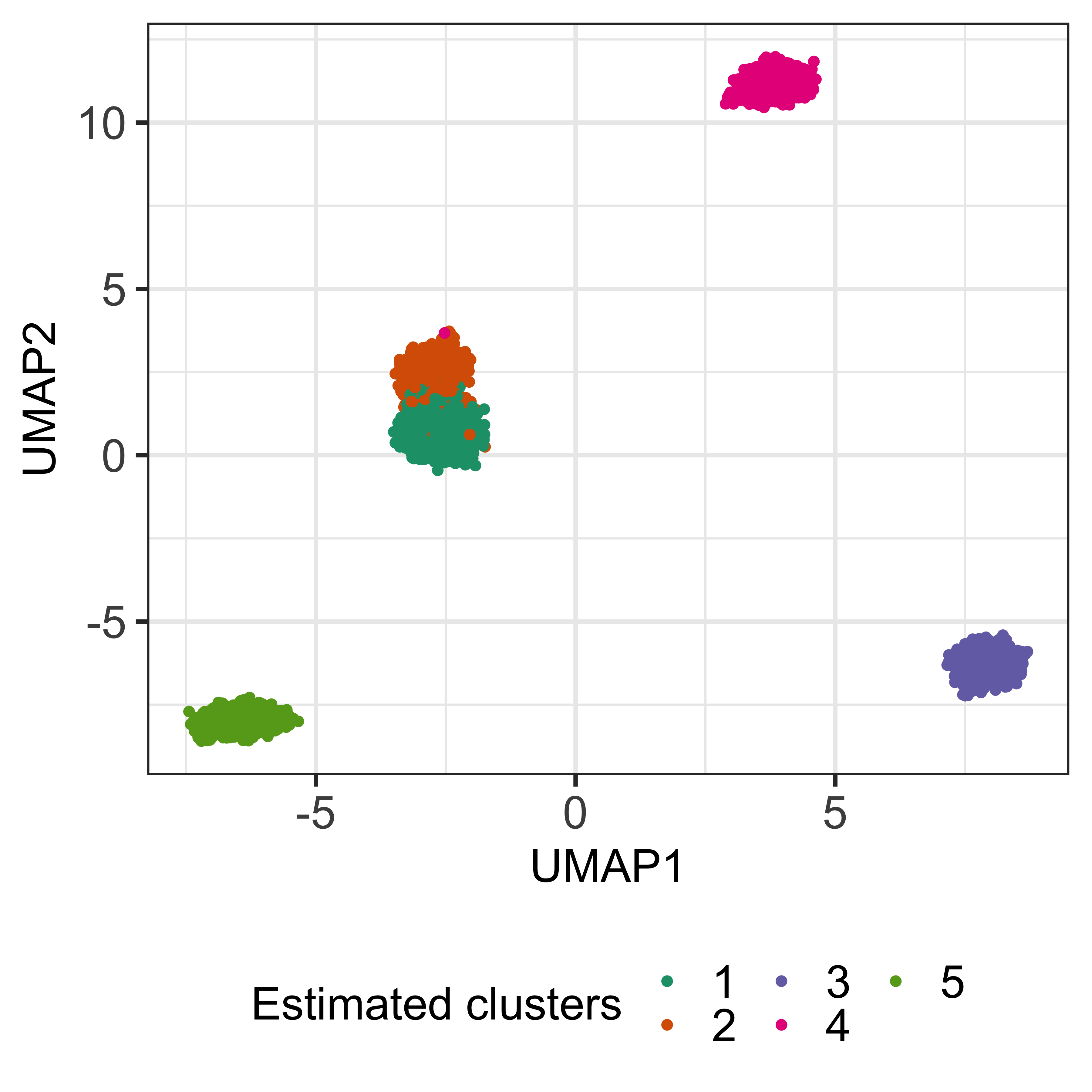}
\end{subfigure}
\vspace*{-4mm}
\caption{\textit{Left: } The two-dimensional UMAP embedding~\citep{McInnes2018-gn} of the ``no cluster'' dataset after preprocessing (as described in Section~\ref{subsection:zheng_real_data}), colored by the estimated cluster membership via $k$-means clustering.
 \textit{Right: } Same as left, but for the ``cluster'' dataset.
}
\label{fig:single_cell_visualization}
\end{figure}



\end{document}